\documentclass[11pt]{article}
\usepackage{qcircuit}
\usepackage{physics}
\usepackage{caption,graphics}
\usepackage{xcolor}
\usepackage{hyperref}
\usepackage{tikz,pgfplots}
\usepackage[font={small}]{caption}
\usepackage{amsmath, amsthm, amssymb, amsfonts,ulem, multicol,cancel,mathtools,esint,subcaption}
\hypersetup{
    colorlinks=true,
    linkcolor=red,
    filecolor=magenta,      
    urlcolor=cyan,
}
\usepackage{pgfplots}
\pgfplotsset{compat=1.18}
\usepackage{geometry}
\usepackage[capitalize]{cleveref}

\def\bs{\ensuremath\boldsymbol}
\newtheorem{theorem}{Theorem}

\newtheorem{problem}[theorem]{Problem}
\newtheorem{lemma}[theorem]{Lemma}

\newtheorem{remark}[theorem]{Remark}

\newtheorem{assumption}{Assumption}
\newtheorem{definition}{Definition}

\newcommand{\bm}{\boldsymbol}

\usepackage{authblk}

\author[1]{Zhiyan Ding \thanks{zding.m@math.berkeley.edu}}
\author[2]{Taehee Ko \thanks{kthmomo@kias.re.kr}\thanks{Ding and Ko are co-first authors with equal contribution.}}
\author[1]{Jiahao Yao \thanks{jiahao@math.berkeley.edu}}
\author[1,3,4]{Lin Lin \thanks{linlin@math.berkeley.edu}}
\author[5]{Xiantao Li \thanks{xiantao.li@psu.edu}}
\affil[1]{Department of Mathematics, University of California, Berkeley, CA 94720, USA}
\affil[2]{School of Computational Sciences, Korea Institute for Advanced Study, 02455, Seoul, South Korea}
\affil[3]{Applied Mathematics and Computational Research Division, Lawrence Berkeley National Laboratory,  Berkeley,  CA 94720, USA
}
\affil[4]{Challenge Institute for Quantum Computation, University of California, Berkeley, CA 94720, USA
}
\affil[5]{Department of Mathematics, Pennsylvania State University, University Park, PA 16802, USA
}

\title{Random coordinate descent: a simple alternative for optimizing parameterized quantum circuits}

\begin{document}

\maketitle

\begin{abstract}
Variational quantum algorithms rely on the optimization of parameterized quantum circuits in noisy settings. The commonly used back-propagation procedure in classical machine learning is not directly applicable in this setting due to the collapse of quantum states after measurements. Thus, gradient estimations constitute a significant overhead in a gradient-based optimization of such quantum circuits. This paper introduces a random coordinate descent algorithm as a practical and easy-to-implement alternative to the full gradient descent algorithm. This algorithm only requires one partial derivative at each iteration. 
Motivated by the behavior of measurement noise in the practical optimization of parameterized quantum circuits, this paper presents an optimization problem setting that is amenable to analysis. Under this setting,
the random coordinate descent algorithm exhibits the same level of stochastic stability as the full gradient approach, making it as resilient to noise. The complexity of the random coordinate descent method is generally no worse than that of the gradient descent and can be much better for various quantum optimization problems with anisotropic Lipschitz constants. Theoretical analysis and extensive numerical experiments validate our findings. 
\end{abstract}

\section{Introduction}\label{sec:intro}

Variational quantum algorithms have emerged as a promising application for near-term quantum devices, addressing various computational challenges with enhanced efficiency~\cite{moll2018quantum, cerezo2021variational}. These algorithms encompass several notable approaches, such as the variational quantum eigensolver \cite{peruzzo2014variational}, the variational quantum simulation~\cite{benedetti2021hardware}, the quantum approximate optimization algorithm \cite{farhi2014quantum,cerezo2021variational,leng2022differentiable}, and quantum machine learning \cite{larocca2021theory, biamonte2017quantum, schuld2015introduction, cerezo2022challenges}. They are designed to operate in a hybrid quantum-classical fashion~\cite{mcclean2016theory, endo2021hybrid}.
In these algorithms, the quantum component involves the implementation of parameterized quantum gate operations. By performing measurements, a cost function (and optionally, its gradient) is obtained as the output. The classical computational procedure then utilizes an iterative method to produce updates for the parameters, which are subsequently leveraged to refine and reprogram the quantum circuits. This iterative process continues until convergence is achieved, forming a feedback loop that continues to improve the algorithm's performance.

In variational quantum algorithms, the optimizable parameters are defined within parameterized quantum circuits (PQCs)~\cite{benedetti2019parameterized, ostaszewski2021structure, sim2019expressibility,altares2021automatic}. A PQC is a sequence of unitary operators represented by parameterized quantum gates that can be readily implemented on a quantum computer. Assuming we are working in an $n$-qubit Hilbert space, a parameterized quantum circuit can be expressed as follows:
\begin{equation}\label{eq: pqc}
    U(\bm \theta) = \prod_{j=1}^J U_j(\bm \theta_j) W_j. 
\end{equation}
Here ${\bm \theta} \!= \!\{\bm \theta_j\}_{j=1}^J$ are the parameters that we need to optimize, $U_j(\bm \theta_j)\in \mathbb{C}^{2^n \times 2^n}$ are the parameterized unitary operators, and $W_j \in \mathbb{C}^{2^n \times 2^n}$ are fixed unitary operators. For instance, a simple example of a PQC consisting only of one-qubit Pauli rotation operators takes the form 
 \[ U_j(\bm \theta_j) = \bigotimes_{m=1}^{M} e^{- i\theta_{j,k_{j,m}} \sigma_{j,k_{j,m}}},  \]
where $\sigma_{j,k_{j,m}}\in\mathbb{C}^{2\times 2}$ is a single-qubit Pauli matrix that acts on $k_{j,m}$-th qubit, $\theta_{j,k_{j,m}}$ represents one of the parameters in $\bm \theta$, and $W_j$'s can be used to represent quantum gates that do not require parametrization, such as the controlled-NOT (CNOT) gate.

Let $d$ be the dimension of the parameters, and we write $\bm\theta=(\theta_1,\theta_2,\cdots,\theta_d)$.
We then optimize the parameter $\bm \theta$ by minimizing a properly chosen cost function $f(\bm \theta)$. 
As an example, the variation quantum eigensolvers (VQE) finds the smallest eigenvalue (ground-state energy) and its corresponding eigenvector (ground state) of a given Hamiltonian matrix $H$  by minimizing the energy of the state:
\begin{equation}\label{eq:vqe}
\theta^*=\mathrm{argmin}_{\bm\theta}  f(\bm \theta)=\mathrm{argmin}_{\bm\theta}\bra{U(\bm \theta) \psi_0} H \ket{U(\bm \theta) \psi_0}\,.
\end{equation}
Here $\ket{\psi_0}\in\mathbb{C}^{2^n}$ is a predetermined initial state that can be easily prepared on a quantum computer. For each given $\bm{\theta}$, $U(\bm{\theta})$ is implemented on a quantum computer to evolve $\ket{\psi_0}$, and the corresponding energy $f(\bm{\theta})$ and its gradient $\nabla_{\bm \theta} f(\bm{\theta})$ can be  consequently obtained with measurements. By solving the optimization problem \eqref{eq:vqe}, the minimum value gives an approximation to the smallest eigenvalue of $H$, while $U(\bm{\theta}^*)\ket{\psi_0}$ approximates the corresponding eigenvector.

\subsection{Problem setup}\label{sec:problem_setup}

Although the problem of optimizing parameters in VQAs resembles classical optimization problems in machine learning, there exist key differences, particularly in how the cost function is evaluated and the level of accuracy that can be obtained for function and gradient evaluations. Firstly, quantum circuits used for estimating partial derivatives in various directions are typically different. This is predominantly because there is no straightforward method (in parallel to backpropagation) to estimate the entire gradient at once, given the inherent nature of quantum states. The predominant method for computing partial derivatives in a PQC is called the parameter-shift rule~\cite{crooks2019gradients, wierichs2022general, banchi2021measuring}, which can only be applied to evaluate one component of the partial derivatives at a time.
As a result, the estimation of the gradient, $\nabla f(\bm\theta)$, typically incurs a cost that is $d$ times greater than the cost associated with merely estimating a single partial derivative, $\partial_i f(\bm\theta)$.

Secondly, the evaluation of any given quantity, a function value or a partial derivative, requires measurement from quantum computers and is subject to measurement noise. We note that this noise is associated with a finite sampling space. For example, a measurement of the Hamiltonian in \eqref{eq:vqe}, which is defined in a finite-dimensional Hilbert space, yields one of its eigenvalues corresponding to the ansatz. Thus, with an increased number of samples or measurements, the central limit theorem suggests that the distribution of the sample average of the function value or the partial derivative can be approximated by a Gaussian distribution, and as a result, the accuracy of function and gradient evaluations can be relatively low. Therefore, the optimization algorithm must be designed to be resilient to measure noise.

In an idealized scenario, we may assume that both the function value and the partial derivatives incorporated into the optimization routine are subject to some Gaussian noise. However, the magnitude of corresponding noises can differ up to a constant, especially in situations where the parameter shift rule is applicable (see \cite{sweke2020stochastic}). With this consideration, the problem of optimizing PQCs can be stated as follows:

\begin{problem}[Optimizing parameterized quantum circuits]\label{problem_1} Finding an efficient algorithm to solve the optimization problem,
\begin{equation}
{\bm \theta}^*=\mathrm{argmin}_{\bm \theta\in\mathbb{R}^d}f(\bm \theta),
\end{equation}
under the following assumptions:
\begin{enumerate}

\item The cost of evaluating a partial derivative scales linearly with that of a function value. 

\item Every evaluation of the function and partial derivative is susceptible to Gaussian noise:
\begin{equation}\label{eqn:g_i}
f(\bm \theta) \to f(\bm \theta)+N(0,\sigma^2_1(\bm\theta))\,, \quad 
\partial_i f(\bm \theta) \to \partial_i f(\bm \theta)+N(0,\sigma^2_2(\bm\theta))\,.
\end{equation}

\end{enumerate}
\end{problem}

Here, $\sigma_1(\bm\theta)$ and $\sigma_2(\bm\theta)$ depend on the real implementation and are not necessarily the same (see \cite{sweke2020stochastic} for example). 
In practical applications, it is common to adaptively adjust the number of samples for varying values of $\bm\theta$ when evaluating $f$ or $\nabla f$. This adjustment ensures that the noise strength, $\sigma_1(\bm\theta)$ and $\sigma_2(\bm\theta)$, remains comparable. For simplicity, in our later analysis, we assume that $\sigma_2(\bm\theta)$ has a uniform upper bound $\sigma_\infty$ (see \cref{as:noise_bound}).

\subsection{Optimization methods}\label{sec:optimization}

One widely used approach for optimizing VQA is through the application of gradient descent (GD)~\cite{sweke2020stochastic, Yuan2019}. The classical gradient descent method involves iteratively updating the parameters $\bm \theta$ by utilizing the gradient of the cost function.
\begin{equation}\label{eq:def_GD}
    \bs{\theta}_{n+1} = \bs{\theta}_n - a_n\nabla f(\bs{\theta}_n),
\end{equation}
where $a_n$ denotes the learning rate. In light of the measurement process in quantum computing, we consider the noisy gradient descent: Rather than implementing \cref{eq:def_GD} with exact $\nabla f(\bs{\theta}_n)$, we apply an unbiased estimator $g(\bs{\theta})$ \footnote{$g(\bs{\theta})$ satisfies $\mathbb E[g(\bs{\theta})] = \nabla f(\bs{\theta})$.} [for example, \eqref{eqn:g_i}]. Consequently,  the parameter update  involves the following iteration,
\begin{equation}\label{eq:def_noisy_GD}
    \bs{\theta}_{n+1} = \bs{\theta}_n - a_n g(\bs{\theta}_n).
\end{equation}
Since $g(\bs{\theta}_n)$ is an unbiased estimation, Eq. \eqref{eq:def_noisy_GD} is equivalent to Eq. \eqref{eq:def_GD} in the expectation sense. Specifically, by taking the conditional expectation on both sides, we have
\begin{equation}\label{eqn:expect_gd}
\mathbb{E}(\bs{\theta}_{n+1}|\bs{\theta_n})=\bs{\theta}_n - a_n\nabla f(\bs{\theta}_n)
\end{equation}
where $\mathbb{E}(\cdot|\bs{\theta_n})$ denotes the conditional expectation given $\bs{\theta}_n$. 

While noisy gradient descent avoids the need for precise gradient information, it still requires the approximated full gradient information at each iteration. As argued before, in the context of VQA, it is often necessary to compute $d$ partial derivatives separately for each direction, which makes the cost of each updating step at least $d$. In this paper, we introduce an alternative optimization method called random coordinate descent (RCD)~\cite{Ste-2015, nesterov2012efficiency, RT_2014} for addressing \cref{problem_1}, with the goal of eliminating the cost dependency on $d$ in each step. RCD can be viewed as a variant of GD where the full gradient in GD is approximated by a randomly selected component of $\nabla f(\bs{\theta}_n)$ in each iteration. Specifically, one RCD iteration can be expressed as:
\begin{equation}\label{eq: def_RCD} 
    \bs{\theta}_{n+1} = \bs{\theta}_n - a_n\bs{e}_{i_n}\partial_{i_n}f(\bs{\theta}_n)\,.
\end{equation}Here $\bs{e}_{i_n}$ is the $i_n$-th unit direction,  $f_{i_n}'(\bs{\theta}_n)$ is the corresponding partial derivative of the cost function, and  $i_n$ is  a random index uniformly drawn from $\{1,2,\cdots, d\}$. Similarly to Eq. \eqref{eq:def_noisy_GD}, 
we can write the noisy RCD as:
\begin{equation}\label{eq:def_noisy_RGD}
    \bs{\theta}_{n+1} = \bs{\theta}_n - a_n\bs{e}_{i_n}g_{i_n}(\bs{\theta}_n)\,.
\end{equation}
It is important to emphasize that in each iteration of RCD \eqref{eq:def_noisy_RGD}, only one partial derivative information is needed. Consequently, within the scope of VQA (as stated in the first assumption of \cref{problem_1}), the cost per step of RCD is $d$ times smaller than that of GD.

\subsection{Contribution}
This paper primarily focuses on the investigation of RCD in the context of noisy gradient evaluation. Our analysis is conducted in a specific comparison with GD, and we illustrate that, under specific conditions, RCD can serve as a favorable alternative for optimizing parameterized quantum circuits. The main contributions of this study can be summarized as follows:
\begin{itemize}
    \item  We show that RCD is theoretically no worse than GD when measuring the complexity by the number of partial derivative calculations (Theorems \ref{thm:GD_informal} and \ref{thm:RCD_informal}), assuming the presence of noise and the \textbf{local} Polyak-Łojasiewicz (PL) condition, which is more appropriate for PQCs (see \cref{remark: no global}).  A summary of the complexities of the two methods is presented in Table \ref{table:complexity} for comparison.  It is important to highlight that the inequality $L_{\mathrm{avg}}\leq L\leq dL_{\mathrm{avg}}$ always holds. Consequently, when the optimization problem is highly anisotropic, i.e., $L\gg L_{\mathrm{avg}}$, RCD is more cost-effective than GD. In the most extreme case when $L$ is nearly equal to $dL_{\mathrm{avg}}$, RCD can reduce the complexity by a factor of $d$ compared to GD. 

    \item We demonstrate that (noisy) GD and RCD converge with high probability under the \textbf{local} PL condition (Assumption \ref{as: local PL}) and are stable under noisy gradient information. Specifically, if the initial parameter $\bs{\theta}_0$ resides within the basin $\mathcal{N}(\mathcal{X})$ surrounding the global minimum, then both noisy methods ensure that the subsequent parameters $\bs{\theta}_n$ will remain consistently within this basin until they converge with the same high probability (Lemmas \ref{lemma: stability of GD} and \ref{lemma: stability of RCD}). To the best of our knowledge, such stochastic stability has not been established for the optimization methods in variational quantum algorithms. We employ a, to the best of our knowledge, novel supermartingale approach and utilize Markov's inequality to quantify the average behavior of the iterates in the optimization method. We anticipate that, as far as we know, this new analytical framework would also be useful for other optimization algorithms for PQCs.

    \item We provide extensive empirical evidence demonstrating that RCD consistently delivers superior performance compared to GD (Sections \ref{sec:numerical_1} and \ref{sec:numerical_2}). Our numerical findings support the theoretical observation that RCD can take a larger learning rate than GD, leading to faster convergence. 
\end{itemize}

\begin{table}
\centering
\begin{tabular}{||c c c c ||}
 \hline
 Algorithm & Iteration cost & Iterations to reach $\epsilon$ tolerance & Total cost \\ [0.5ex] 
 \hline
 GD  & $\Omega(d)$ & $\tilde{\mathcal{O}}\left(\max\left\{\frac{L\sigma^2_\infty d}{\mu^2\epsilon},\frac{L}{\mu}\log\left(\frac{1}{\epsilon}\right)\right\}\right)$ &  $\Omega(\frac{L\sigma^2_\infty d^2}{\epsilon})$\\
 \hline
 RCD  & $\Omega(1)$ & $\tilde{\mathcal{O}}\left(\max\left\{\frac{L_{\mathrm{avg}}\sigma^2_\infty d^2}{\mu^2\epsilon},\frac{dL_{\max}}{\mu}\log\left(\frac{1}{\epsilon}\right)\right\}\right)$ &  $\Omega(\frac{L_{\mathrm{avg}}\sigma^2_\infty d^2}{\epsilon})$\\
 \hline
\end{tabular}
\caption{Comparison of the gradient descent and the randomized coordinate descent methods with an unbiased noisy gradient estimation. $d$ is the dimension of the parameter, and the smoothness constants $L$ and $L_{\mathrm{avg}}$ are defined in \eqref{L2} and \eqref{Lmax}, respectively. $\sigma^2_\infty$ is a bound for the measurement noise defined in \eqref{eqn:sigma_infinite}. In the table, we limit our attention to the situation where the learning rate is fixed.
}
\label{table:complexity}
\end{table}

\subsection{Related works}

\noindent{\bf Gradient descent with noise}
The noisy gradient descent \eqref{eq:def_noisy_GD} is a popular optimization method in the classical machine learning community. Notable examples are the stochastic gradient descent (SGD)~\cite{bottou2018optimization} or the perturbed gradient descent~\cite{jin2021nonconvex}. The convergence properties of the noisy gradient descent method in (\ref{eq:def_noisy_GD}) have been extensively studied ~\cite{moulines2011non, nguyen2018sgd, rakhlin2011making, nguyen2018tight,sweke2020stochastic, liu2022loss}. These previous works established that when the cost function is $L$-smooth, and $\mu$ strong convex (or PL~\cite{Polyak_1963}) and satisfies other mild conditions, $f(\bs{\theta}_n)$ converges linearly to an approximation of $f_{\min}$. 
In the recent study by Sweke $et$ $al$. \cite{sweke2020stochastic}, a comparable theoretical result was demonstrated for the application of the noisy GD method to quantum optimization problems. The authors establish the $L$-smooth property of the loss function \eqref{eq:vqe} and achieve the convergence of noisy gradient descent under the assumption of unbiased gradient estimation and global PL condition. Another approach to interpreting noise involves assuming that the variational quantum state is influenced by a noisy quantum map~\cite{PhysRevA.102.052414}, ultimately resulting in a noisy loss function. In~\cite{PhysRevA.102.052414}, the authors investigate the convergence properties of noisy GD, considering the impact of both noise and the stochastic nature of quantum measurement outcomes. 
They specifically illustrate that, when assuming the global PL condition, noisy GD has the capability to converge to a region where the perturbation in loss is bounded by a quantity associated with the quantum fisher information of the variational state. With an important observation that PQCs typically do not obey the global PL condition (see \cref{remark: no global}), we formulate a local PL condition, a property that is more appropriate for PQCs, and we establish the convergence of noisy GD and RCD under this weaker condition and unbiased gradient estimation. We substantiate the efficacy of both optimization algorithms with high probability. To the best of our knowledge, this result is novel and has not been previously explored.

%~\re{Very recently,~\cite{an2024convergencestochasticgradientdescent} provided a convergence proof for Stochastic Gradient Descent (SGD) under the local Łojasiewicz condition for deep neural networks. The high-level nature of their result is similar to our convergence result for noisy GD under the local PL condition, as detailed in Lemma \ref{lemma: stability of GD} and Theorem \ref{thm:GD_noisy}. Specifically, they demonstrated that when the loss function satisfies the local Łojasiewicz condition and the initial optimization point is sufficiently good, SGD can converge to the global minimum with high probability.} 

\noindent{\bf Randomized coordinate descent}
  The (RCD) method has proven its efficiency over GD in many large-scale optimization problems. The convergence properties of RCD have been extensively explored in the fields of machine learning and optimization~\cite{Ste-2015, nesterov2012efficiency, RT_2014, nesterov2017efficiency, lee2013efficient,chen2023global}. For example, it was shown in \cite{nesterov2012efficiency} that when $f$ is strongly convex, the convergence complexity of RCD can be consistently lower than or equal to that of GD. Here complexity refers to the total number of partial derivative calculations required for convergence. Later, for strongly convex functions, RCD accelerations were achieved with adaptive momentum-based strategies in various regimes \cite{lee2013efficient,nesterov2017efficiency}. For the nonconvex optimization, recent work \cite{chen2023global} shows the global convergence behavior of RCD with a focus on saddle point avoidance. 
  Nevertheless, 
  convergence rates of RCD have been scarcely studied for nonconvex optimization problems. More importantly, most related works focused on the case where partial derivatives are computed exactly, while in this work, we deal with the case where partial derivatives are estimated, which is subject to noise, and we will refer to it as {\it noisy} RCD \eqref{eq:def_noisy_RGD}.

\noindent{\bf Locally-defined convex conditions for convergence analysis}
One limitation of the conventional convergence analysis is its reliance on assumptions of global convex \cite{bottou2018optimization} or global PL \cite{sweke2020stochastic} conditions for the cost function $f(\bm \theta)$. However, we show that such global assumptions are not satisfied in quantum problem applications with PQCs, as elaborated on \cref{remark: no global}.
Thus, one must weaken such a global assumption to a local one in the analysis. Convergence analysis under local assumptions requires more sophisticated techniques (see \cite{fehrman2020convergence,ko2022local,patel2022global,mertikopoulos2020almost} and therein), but it provides important insights that help to interpret empirical results. In our work, we make a local nonconvex condition based on the local PL condition~\cite{liu2022loss}. Under this condition and suitable assumptions for the cost function, we show a local convergence analysis for the optimization of PQCs. By employing a stochastic stability argument, we demonstrate that the noisy GD and RCD methods maintain a comparable convergence rate under our local PL condition with high probability (refer to \cref{thm:GD_informal} and \cref{thm:RCD_informal}). To the best of the authors' knowledge, this paper is the first to provide a rigorous result for the complexity of noisy GD and RCD under a local PL condition designed for variational quantum algorithms built from PQCs. The recent work ~\cite{an2024convergencestochasticgradientdescent} also provided a local convergence analysis for the noisy GD but their local assumption and result are more suitable for deep neural networks in the classical machine learning.

\noindent{\bf Other quantum optimization methods}

Simultaneous perturbation stochastic approximation (SPSA)~\cite{Spall_1992,Spall_1987}, is a zero th-order method (i.e., it only involves function values) that has some  similarity to the RCD method. Each iteration of SPSA employs a finite-difference formula to estimate one directional derivative of the loss function. SPSA has been used in~\cite{Kandala_2017} to update the control parameters in VQE and achieves a level of accuracy comparable to standard gradient descent methods. We highlight that when each directional or partial derivative is approximated using only two loss function values through finite-differences, the cost per step of SPSA and RCD should be comparable. However, in scenarios where analytic partial derivatives are available, such as those where the parameter shift rule \cite{sweke2020stochastic} is applicable, RCD becomes significantly more robust than SPSA, since the small parameter in finite-difference formulas tends to amplify the noise inherence in the measurement.  
%In contrast to finite-differences, where the perturbation parameter in the denominator amplifies the estimation noise, 
In contrast, derivative estimations from analytic derivatives are more stable in the presence of measurement and quantum noise, as substantiated by the numerical example provided in Appendix \ref{sec:RCD_vs_SPSA}\footnote{One can also use the  parameter shift rule to calculate the directional derivative in SPSA. Nevertheless, the computation of a random directional derivative always requires evaluating $d$ partial derivatives, making it $d$ times more expensive than the cost of RCD per step.
}.

Besides SPSA, there are other promising zero-order methods in variational quantum optimization, more commonly known as gradient-free methods. Notably, policy gradient-based techniques have shown their effectiveness in noise robust optimization in the NISQ~\cite{yao2020policy}. Sung $et$ $al$.\cite{sung2020using} construct models based on the previous method and further improve the sample efficiency of the methods. Furthermore, these zero-order optimization methods leverage the strengths of reinforcement learning~\cite{pmlr-v145-yao22a, fosel2021quantum, PhysRevB.98.224305, PhysRevX.8.031086}, Monte Carlo tree search~\cite{yao2022monte, meng2021quantum, rosenhahn2023monte}, and natural evolutionary strategies~\cite{anand2021natural, zhao2020natural, giovagnoli2023qneat},  Bayesian~\cite{tibaldi2022bayesian, tamiya2022stochastic}, as well as Gaussian processes~\cite{zhu2019training}. 

In addition to these zero-order methods, several other optimization methods have been proposed recently~\cite{jordan2005fast,rebentrost2019quantum,stokes2020quantum,gilyen2019optimizing,gao2021quantum}. One interesting example is the quantum natural gradient (QNG)~\cite{stokes2020quantum}, an approximate second-order method, that incorporates the quantum geometric tensor, which is similar to the natural gradient in classical machine learning. While an outcome of measurement is used as an estimate of the gradient in the QNG or the noisy gradient \eqref{eq:def_noisy_GD} from \eqref{eq: pqc},
the Jordan algorithm \cite{jordan2005fast} encodes the partial derivatives as binary numbers in the computational basis. 
This algorithm was later  improved by Gilyen $et$ $al$.  \cite{gilyen2019optimizing} using high-order finite-difference approximations,  and applications to VQAs for a certain class of smooth functions were considered. However, the methods \cite{jordan2005fast,gilyen2019optimizing} require a significant number of ancilla qubits and complex control logics, due to the binary encoding of partial derivatives. Alternatively, \cite{abbas2023quantum} proposed a quantum backpropagation algorithm, which uses $\log d$ copies of the quantum state to compute $d$ derivatives. The overhead for computing $d$ derivatives is $\text{polylog}(d)$ times that of function evaluation (therefore mimicking backpropagation). One of the main drawbacks of their algorithm is that there is an exponential classical cost associated with the process.
For a more restrictive class of cost functions (polynomial functions),
\cite{rebentrost2019quantum} proposed a framework to implement the gradient descent and Newton's methods. This method also requires the coherent implementation of the cost function on a quantum computer using e.g., sparse input oracle, and thus can be challenging to implement in near-term devices.

\subsection{A numerical 
illustration: Variational quantum eigenvalue solver}\label{sec:numerical_1}

As a brief illustration of the performance of noisy GD versus RCD methods, we consider the transverse-field Ising model (TIFM),
\begin{equation}\label{eq:isingmodel}
    H = J\sum_{j=1}^{N-1}Z_jZ_{j+1}+\Delta\sum_{j=1}^NX_j,
\end{equation}with the coefficient $J=1$ and $\Delta = 1.5$.  Here $N$ denotes the number of qubits, and $X_j,Z_j$ are Pauli operators acting on the $j$-th qubit. In \cref{fig:vqe}, we set $N=10$. To implement the quantum circuits, we use Qiskit Aer-simulator \cite{Qiskit}
with the command ``result.get\_counts" that outputs measurement outcomes as classical bitstrings. We utilize the resulting classical bitstrings to compute partial derivatives by applying the parameter shift rule \cite{sweke2020stochastic}. Thus, the result in \cref{fig:vqe} takes into account the measurement noise. 

In each experiment, 10 independent simulations are used  with a fixed initialization. The parameterized quantum circuit used for estimating the ground state energy of the Hamiltonian \eqref{eq:isingmodel} is given in  Figure~\ref{fig:circuit} (Appendix \ref{sec:circuit}).

We compare the optimization performance of the two methods in terms of the number of partial derivative evaluations. The optimization results in \cref{fig:vqe} suggest that RCD requires nearly 4 times fewer partial derivative evaluations than GD to converge to an energy ratio of 0.99 and a fidelity of 0.97, both of which are higher than the energy ratio and the fidelity obtained from GD. This observation can be explained by the analysis in \cref{sec:main_result}, i.e.,  RCD can be more efficient than GD when the ratio of Lipschitz constants ($L/L_{\mathrm{avg}}$ or $L/L_{\max}$) is significantly larger than $1$. Specifically, the ratio of the total computational cost of GD to RCD can be linked to the Lipschitz ratios, as summarized in \cref{table:complexity}.  For instance, in the lower panels of \cref{fig:vqe}, we observe that the ratio $L/L_{\mathrm{avg}}$ and $L/L_{\max}$ remains above 20 and 11 throughout the iterations on average. The faster convergence of RCD can be attributed to these large Lipschitz ratios.

\begin{figure}[t]
   \centering
   \begin{subfigure}[b]{0.32\textwidth}
    \centering
    \includegraphics[width=\textwidth]{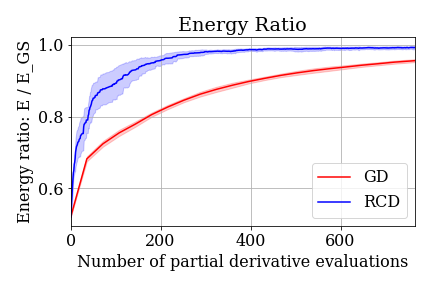}
    \end{subfigure}
   \begin{subfigure}[b]{0.32\textwidth}
    \centering
    \includegraphics[width=\textwidth]{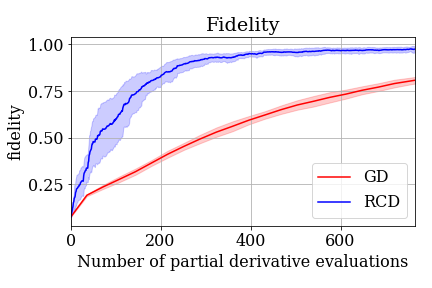}
    \end{subfigure}
    
    \begin{subfigure}[b]{0.32\textwidth}
    \centering
    \includegraphics[width=\textwidth]{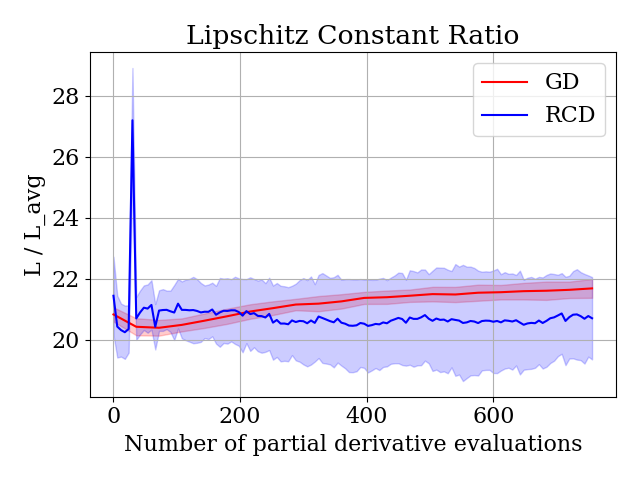}
    \end{subfigure}
   \begin{subfigure}[b]{0.32\textwidth}
    \centering
    \includegraphics[width=\textwidth]{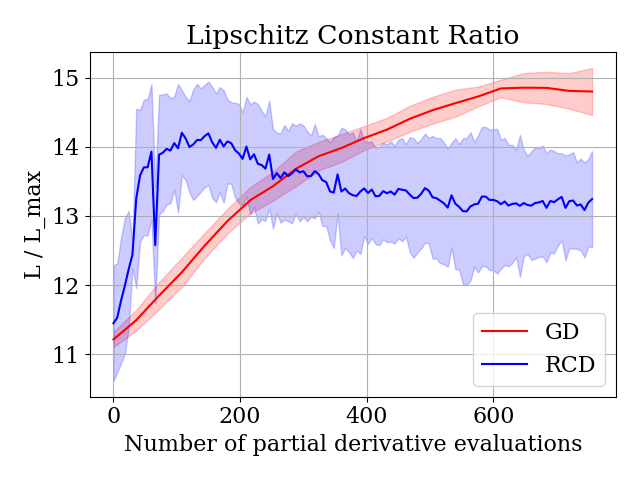}
    \end{subfigure}
    \caption{The comparison of the performance of GD (red) and RCD (blue) for optimizing the Hamiltonian  \eqref{eq:isingmodel}. The unit of the $x$ axis labels the number of partial derivative evaluations as an indication of the computational complexity. 
    The top  panels show the approximation of the ground state, including the energy ratio (left) and  fidelity (right). In the bottom panels, we show the ratios of Lipschitz constants obtained from the two methods are compared: $\frac{L}{L_{avg}}$ (left) and $\frac{L}{L_{\max}}$ (right). }
    \label{fig:vqe}
\end{figure}

\section{Preliminaries and main results}
Before we establish results pertinent to the performance of RCD, we first establish consistent notations and assumptions, which are presented in Section \ref{sec:notation}. Following that, we outline our key theoretical findings in Section \ref{sec:main_result}.

\subsection{Notations and assumptions}\label{sec:notation}
Given a vector $\bm v\in\mathbb{R}^d$, we use standard norms for $\bm v$, including the $2$-norm $\|\bm v\|_2:=\sqrt{\sum_iv_i^2}$ and the $\infty$-norm $\|\bm v\|_{\infty}:=\max_i\left|v_i\right|$. In order to ensure the convergence of gradient-based methods, we list several technical assumptions. 

We assume the cost function $f$ satisfies the $L$-smoothness. Specifically, it satisfies the following assumption:    
\begin{assumption}\label{as:f_L_smooth} The cost function $f$ is $L$-smooth, in that, 
    \begin{equation}\label{L2}
    \|\nabla f(\bs{\theta}) - \nabla f(\bs{\theta}')\|_2\leq L\|\bs{\theta}-\bs{\theta}'\|_2,\quad \text{for all }\bs{\theta},\bs{\theta}'\in\mathbb{R}^d.
\end{equation}
\end{assumption}
It is worth noting that assuming $L$-smoothness is standard when analyzing the convergence of gradient-based algorithms. Moreover, in the setting of PQC, verifying the $L$-smoothness of the loss function is straightforward. An example of this can be seen in~\cite[Theorem 2]{sweke2020stochastic}, where the author demonstrates that a broad class of PQCs can generate loss functions that satisfy $L$-smoothness properties.

Since the gradient is Lipschitz continuous, the partial derivatives are Lipschitz continuous as well. We define the componentwise Lipschitz constants,    
\begin{definition}We say that a function $f$ is $L_i$-smooth with respect to the $i$-th component if
    \begin{equation}\label{Li}
        |\partial_i f(\bs{\theta}+\bs{e}_ih)-\partial_i f(\bs{\theta})|\leq L_i|h| \quad \forall h\in \mathbb{R},
    \end{equation}where $\partial_i f(\bs{\theta})$ denotes the partial derivative in the $i$-th direction. 
    
\end{definition}From these componentwise Lipschitz constants, we denote the maximum and average of those constants as
\begin{equation}\label{Lmax}
    L_{\max}:=\max_{1\leq i \leq d} L_i,\quad L_{\mathrm{avg}}=\frac{1}{d}\sum^d_{i=1}L_i\,.
\end{equation}
As shown in~\cite{Ste-2015}, in general we have,
\begin{equation} \label{eq:LiL}
L_i \le L_{avg}\le L_{\max} \le L \le d L_{\max}\,.
\end{equation}
Another interpretation is through the Hessian:  When $f$ is twice continuously differentiable, the condition \eqref{L2} is equivalent to $\nabla^2 f(x)\preceq L I_d$, and similiarly, the condition \eqref{Li} is equivalent to $\sup_{\bs{\theta}}\left|\partial_{i}^2 f(\bs{\theta})\right|\leq L_i$. We note that both the upper and lower bounds of $L$ in terms of $L_{\max}$ in \eqref{eq:LiL} are tight. If $\nabla^2 f$ is a diagonal matrix, then $L_{\max}=L$, both being the largest diagonal element of $\nabla^2 f$. 
(This is the case in which all coordinates are independent of each other, for example, $f = \sum_i\lambda_ix_i^2$.)
On the other hand, if $\nabla^2f = \mathsf{e}\cdot\mathsf{e}^\top$, where $\mathsf{e}\in\mathbb{R}^d$ satisfies $\mathsf{e}_i=1$ for all $i$, then $L = dL_{\max}$. 
This is a situation where $f$ is highly anisotropic, e.g.,  $f = (\sum_i x_i)^2/2$, where $L=d$ and $L_{\max}=1$. 
In addition, when $L_{avg}=L$, we see that $L_{avg}=L_{\max}=L_i$ for all $i$.

\medskip

Next, it is important to note that the estimation of the gradients in quantum computing can be susceptible to noise, which stems from the inherent nature of quantum measurements. Consequently, in our analysis and comparative studies of different optimization methods, we will take into account the presence of noise. To facilitate such analysis, we make the following assumption:

\begin{assumption}[Bounds of the noise with respect to the $2$-norm]\label{as:noise_bound}
Given any $\bs{\theta}\in\mathbb{R}^d$, we assume that we can find an unbiased random estimate $\bs{g}(\bs{\theta})$ for the gradient $\nabla f(\bs{\theta})$, meaning that 
\[
\mathbb{E}\left[\bs{g}(\bs{\theta})\right]=\nabla f(\bs{\theta})\,.
\]
Furthermore, we assume that there exists a constant $\sigma^2_\infty>0$ such that
\begin{subequations}
    \begin{eqnarray}  \sigma^2_\infty>\sup_{\bs{\theta}\in\mathbb{R}^d}\max_{1\leq i \leq d}\mathbb{E}\left[|\partial_i f(\bs{\theta}) - g_i(\bs{\theta})|^2\right]\,.\label{eqn:sigma_infinite}
    \end{eqnarray}
\end{subequations}
Here, we also assume $g(\bs{\theta})$ is independent for different $\bs{\theta}$.
\end{assumption}
The constraint on the noise bounds is a standard assumption in the analysis of noisy gradient-based algorithms. A straightforward case that satisfies this assumption is $\nabla f(\bs{\theta})\in L^\infty(\mathbb{R}^d)$, a condition frequently observed in optimizations involving PQCs. Moreover, in cases where $|\nabla f(\bs{\theta})|$ lacks an upper bound, we can adjust the sample sizes in various $\bs{\theta}$ values to effectively reduce the variance to a bounded range.

Additionally, we assume the existence of a basin encompassing the global minimum, within which $f$ satisfies the Pl condition~\cite{Polyak_1963} and, equivalently, the local PŁ condition~\cite{liu2022loss}.

\begin{assumption}[Local PL condition]\label{as: local PL} Define $\mathcal{X}$ as the set of global minima and $f_{\min}$ as the global minimum value evaluated over $\mathcal{X}$. Then there exists a $\delta_f,\mu>0$ such that for any $\bs\theta\in\mathcal{N}(\mathcal{X}):=f^{-1}([f_{\min},\delta_f))$,    
\begin{equation*}
        \|\nabla f(\bs{\theta})\|^2\geq 2\mu \left(f(\bs{\theta})-f_{\min}\right).
    \end{equation*}
\end{assumption}
While a similar condition was theoretically justified in classical machine learning (e.g., over-parameterized regimes \cite{liu2022loss}), analysis on the case of parameterized quantum circuits \cref{eq: pqc} is still lacking and challenging. The technical tools in \cite{liu2022loss} might be applicable for the case of PQCs. We leave this for future work. Nevertheless, it is worthwhile to highlight that the local PL condition is defined not on the entire space $\mathbb{R}^d$ but the neighborhood of global minima $\mathcal{N}(\mathcal{X})$, which is reasonable in the setting of variational quantum algorithm. We support this argument with the following remark.

%\re{We acknowledge the technical challenges associated with theoretically justifying the local PL condition, as it often varies on a case-by-case basis. However, it is worthwhile to highlight that the local PL condition is defined not on the entire space $\mathbb{R}^d$ but the neighborhood of global minima $\mathcal{N}(\mathcal{X})$, which is reasonable in the setting of variational quantum algorithm.} We support this argument with the following remark.

\begin{remark}[PQCs are neither global PL nor globally convex]\label{remark: no global}
    Let $f(\bs\theta)$ be a cost function defined by some parameterized quantum circuit \eqref{eq:vqe}. Note that $f$ is  periodic and smooth, due to its specialized form. By the extreme value theorem, we see that there exist global maximum and minimum of $f$, denoted by $\bs\theta_{\max}$ and $\bs\theta_{\min}$. In general, $f$ is not constant, which means that $f_{\max}>f_{\min}$. Had $f$ satisfied the global PL condition, it would have followed that at the global maximum $\bs\theta_{\max}$, 
    \begin{equation}
        0=\|\nabla f(\bs\theta_{\max})\|^2\geq 2\mu \left(f_{\max}-f_{\min}\right)\geq 0,
    \end{equation}which gives a contradiction to the general case that $f_{\max}>f_{\min}$. As another case study, if $f$ is assumed to be convex, namely,
    \begin{equation}
        f(\bs \theta')\geq f(\bs \theta) + (\nabla f(\bs\theta),\bs\theta'-\bs\theta)\text{ for all }\bs\theta,\bs\theta'\in\mathbb{R}^d,
    \end{equation}then setting $\bs\theta = \bs\theta_{\max}$ and $\bs\theta'=\bs\theta_{\min}$ results in a contradiction. Therefore, the cost function $f$ that is constructed from an ansatz similar to \eqref{eq:vqe}, will not satisfy global PL or convex conditions in general.
\end{remark}

\subsection{Main result: complexity comparison of GD and RCD}\label{sec:main_result}
In this study, our main focus is to compare the complexity of noisy GD and RCD under the assumptions of a local PL condition~\ref{as: local PL}. For the sake of simplicity, in the remaining part of this paper, we will refer to ``noisy gradient descent" and ``noisy randomized coordinate descent" as ``GD" and ``RCD", respectively, without explicitly mentioning the term ``noisy".

The main theoretical results are summarized in the following two theorems:
\begin{theorem}[Complexity of GD~\eqref{eq:def_noisy_GD} under the local PL condition]\label{thm:GD_informal}
    Assume $f$ is a $L$-smooth function that satisfies assumption \ref{as: local PL} and $g$ satisfies assumption \ref{as:noise_bound}. Given $\epsilon>0$ small enough, if $f(\bs{\theta}_1)\leq \delta_f$ and $a_n=\Theta(\min\left\{\mu \epsilon/(L\sigma^2_\infty d),1/L\right\})$ in GD \eqref{thm:GD_noisy}, then with probability $1-f(\bs{\theta}_1)/\delta_f-o(1)$, there  exists  at least one
        \begin{equation}
            n<N=\tilde{\Theta}\left(\max\left\{L\sigma^2_\infty d/(\mu^2\epsilon),L/\mu\right\}\right) 
        \end{equation}such that $f(\bs{\theta}_n)\leq f_{\min}+\epsilon$.

\end{theorem}
\begin{theorem}[Complexity of RCD~\eqref{eq:def_noisy_RGD} under the local PL condition]\label{thm:RCD_informal}
    Assume $f$ is a $L$-smooth function that satisfies assumption \ref{as: local PL} and $g$ satisfies assumption \ref{as:noise_bound}. Given $\epsilon>0$ small enough, if $f(\bs{\theta}_1)\leq \delta_f$ and $a_n=\Theta(\max\left\{\mu \epsilon/(L_{\mathrm{avg}}\sigma^2_\infty d),1/L_{\max}\right\})$ in RCD \eqref{thm:RCD_noisy}, then with probability $1-f(\bs{\theta}_1)/\delta_f-o(1)$, there  exists  at least one 
        \begin{equation}
        n<N=\tilde{\Theta}(\max\left\{L_{\mathrm{avg}}\sigma^2_\infty d^2/(\mu^2\epsilon),L_{\max}d/\mu\right\})   
        \end{equation}
        such that $f(\bs{\theta}_n)\leq f_{\min}+\epsilon$.
\end{theorem}

Based on the theorem mentioned above, to achieve $f(\boldsymbol{\theta}_n)-f_{\min} \leq \epsilon$, we can select the learning rate $a_n = \frac{\mu \epsilon}{L\sigma^2_\infty d}$ for GD and $a_n = \frac{\mu \epsilon}{L_{\mathrm{avg}}\sigma^2_\infty d}$ for RCD. Recalling Eq. \eqref{eq:LiL}, we observe that $L_{\mathrm{avg}} \leq L$, which means that we could use a larger learning rate for RCD. This choice aligns with the learning rates utilized in the numerical experiments presented in Section \ref{sec:numerical_1} as well as those in Section \ref{sec:numerical_2}.

We compare the complexity of the noisy GD and RCD methods with the estimates of the number of iterations. First, according to the above result, we conclude that the number of iterations required for GD is $N=\tilde{\Theta}\left(\frac{L\sigma^2_\infty d}{\mu^2\epsilon}\right)$\footnote{This complexity aligns with the classical theoretical results for gradient descent (GD), which typically assume the presence of strong convexity or a local PL condition for the function $f$.}, while for RCD,  we have $N=\tilde{\mathcal{O}}\left(\frac{L_{\mathrm{avg}}\sigma^2_\infty d^2}{\mu^2\epsilon}\right)$. Notably, in RCD, there is an additional factor of $d$, which can be understood in the expectation sense: During each iteration of the noisy RCD, the randomness arises from two sources: the random direction $i_n$ and the noisy partial derivative $g_{i_n}(\bm \theta_n)$. By taking the conditional expectation with respect to $\bs{\theta}_n$, we obtain:
\begin{equation}\label{eqn:expect_rcd}
\mathbb{E}(\bs{\theta}_{n+1}|\bs{\theta_n})=\bs{\theta}_n - \frac{a_n}{d}\nabla f(\bs{\theta}_n)\,.
\end{equation}
Compared with \eqref{eqn:expect_gd}, there is an extra $1/d$ factor in the expectation of RCD. Consequently, in each iteration, the rate of decay of the cost function is smaller in RCD compared to GD. 
Consequently, we anticipate that RCD would necessitate more iteration steps to achieve convergence. On the other hand, it is also important to note that in certain scenarios where $L_{\mathrm{avg}}d$ is comparable to $L$, the number of iterations required for RCD is comparable  to that of GD.

Meanwhile, it is important to point out that a more practical criterion for comparing the two methods is the cumulative cost of each method, which is represented by the number of partial derivative calculations from the quantum circuits. This is because quantum algorithms for estimating the gradient have a cost proportional to $d$.  Since each iteration of GD needs to calculate the full gradient ($d$ partial derivatives), the total number of partial derivative estimations in GD is
\[
N_{\text{partial,GD}}=\tilde{\Theta}\left(\frac{L\sigma^2_\infty d^2}{\mu^2\epsilon}\right)\,.
\]
In contrast, the number of partial derivative estimations in RCD is:
\[
N_{\text{partial,RCD}}=\tilde{\mathcal{O}}\left(\frac{L_{\mathrm{avg}}\sigma^2_\infty d^2}{\mu^2\epsilon}\right)\,.
\]
From Eq. \eqref{eq:LiL}, we can deduce that:
\[
\tilde{\Omega}(N_{\text{partial,RCD}})=N_{\text{partial,GD}}=d\tilde{\mathcal{O}}(N_{\text{partial,RCD}})\,.
\]
This suggests that the computational cost of RCD is $L/L_{\mathrm{avg}}$ times cheaper than that of GD. In an extreme case where $f$ is highly skewed, i.e., $L/L_{\mathrm{avg}} \approx d$,  RCD can reduce the computational cost by a factor of the dimension $d$, which will be a significant reduction for large quantum circuits.

In addition to the complexity result, it is worth noting that the two methods exhibit similar success probability, which is approximately $1-f(\bs{\theta}_1)/\delta_f$, as indicated by the two aforementioned theorems. This observation is quite surprising, as each iteration of RCD appears noisier due to the random selection of the updating direction $i_n$.  Intuitively, this suggests that we might need to choose a smaller learning rate $a_n$ to ensure stability in RCD, which would consequently increase its complexity. However, our theory unveils that choosing a similar learning rate $a_n$ is adequate to stabilize RCD. To elucidate this point, it's important to recognize that, on average, RCD behaves equivalently to GD. By conducting more iterations, RCD can approximate its average behavior (expectation), effectively mitigating the extra randomness introduced by $i_n$. This compensation mechanism ensures that the success probabilities remain consistent between the two methods.

\section{Proof of main results}\label{sec: main_results}

In this section, we provide the proofs for Theorems \ref{thm:GD_informal} and \ref{thm:RCD_informal}. We will start by showing the stochastic stability of the two methods in Section \ref{sec: stochastic stability}. This will guarantee that the parameter $\bm\theta$ is likely to stay close to the global minimum until attaining a small loss. Following that, in Section \ref{sec:convergence analysis}, we utilize the local PL condition around the global minimum to establish the convergence of $f(\bs{\theta}_n)$. In all of the following theoretical results and the corresponding proofs in the Appendices, we assume $f_{\min}=0$ without loss of generality by modifying the original function as
\begin{eqnarray}\label{eq:shiftfmin}
    f(\bs\theta)\leftarrow f(\bs\theta) - f_{\min}.
\end{eqnarray}Thus, all results in this section can be reformulated for the original cost function by the substitution \eqref{eq:shiftfmin}, which will yield  \cref{thm:GD_informal,thm:RCD_informal}.

\subsection{Stochastic stability}\label{sec: stochastic stability}

In the context of optimization, stability and convergence are not separate properties. In a deterministic algorithm, convergence immediately guarantees stability. However, this connection does not hold for stochastic processes in general. For instance, when optimization methods such as noisy GD, SGD, or noisy RCD are applied, discrete-time stochastic processes are generated. In such cases, a convergence theory must be developed for a collection of random paths, which can exhibit different convergence behaviors among themselves. 

In our specific case, we anticipate that when $\bs{\theta}_n$ remains within the basin $\mathcal{N}(\mathcal{X})$ and the learning rate is correctly chosen, both the GD and the RCD methods, when the gradient is exactly calculated, converge to a global minimum due to the local PL condition stated in assumption \ref{as: local PL}. However, in the presence of noise in the gradient and the use of a constant learning rate, it is generally impossible to ensure that $\bs{\theta}_n\in \mathcal{N}(\mathcal{X})$ almost surely, unless a different strategy is adopted such as the decreasing learning rates~\cite{patel2022global,fehrman2020convergence,ko2022local}. On the other hand, the purpose of the optimization algorithm is to minimize the loss function, which means that it suffices to ensure stability until a small loss is achieved. To quantify such a likelihood, in this section, we demonstrate that when $\bs{\theta}_0\in \mathcal{N}(\mathcal{X})$, there exists a finite probability that $\bs{\theta}_n$ obtained from GD and RCD remain within $\mathcal{N}(\mathcal{X})$ until achieving a small loss. This provides a high probability of convergence for the two methods.

We summarize the result for noisy GD in the following lemma.

\begin{lemma}\label{lemma: stability of GD}
Assume that $f$ is a $L$-smooth function that satisfies the assumption \ref{as: local PL} and $g$ satisfies the assumption \ref{as:noise_bound}. If $f(\bs{\theta}_1)\leq \delta_f$ and the learning rate is chosen as follows,
\[
a_n=a<\min\left\{\frac{1}{L}\;,\frac{2\mu \delta_f}{L\sigma^2_\infty d}\right\}\,,
\]
then,  with high probability, iterations of noisy GD \eqref{thm:GD_noisy} remain in $f^{-1}([0,\delta_f))$ until a small loss is achieved. Specifically, 
\begin{equation}\label{eqn:p_gd}
\mathbb{P}\left\{\exists N>0\ \text{such that}\ f(\bs\theta_N)\notin\mathcal{N}\ \text{and}\  f(\bs\theta_n)>\frac{La\sigma^2_\infty d}{\mu},\ \forall n<N\right\}\leq \frac{f(\bs{\theta}_{1})}{\delta_f}.
\end{equation}
\end{lemma}
In light of Eq. \eqref{eqn:p_gd}, if we select the learning rate $a_n$ to be sufficiently small, then with a probability of $1-\frac{f(\bs{\theta}_{1})}{\delta_f}$, the parameters are guaranteed to achieve a small loss before escaping the basin.
\medskip

Despite infrequent updates of the gradient components, RCD still demonstrates a similar level of stochastic stability. This key observation is summarized in the following lemma:
\begin{lemma}\label{lemma: stability of RCD}
Assume that $f$ is a $L$-smooth function that satisfies assumption \ref{as: local PL} and $g$ satisfies assumption \ref{as:noise_bound}. Given any $f(\bs{\theta}_1)<\delta_f$, if one chooses the learning rate 
\[
a_n=a<\min\left\{\frac{1}{L_{\max}},\; \frac{d}{\mu}, \;\frac{2\mu \delta_f}{L_{\mathrm{avg}}\sigma^2_\infty d}\right\}\,,
\]
then,  with high probability, iterations from the noisy RCD \eqref{thm:RCD_noisy} stay at $f^{-1}([0,\delta_f))$ until achieving a small loss. Specifically, 
\[
\mathbb{P}\left\{\exists N>0\ \text{such that}\ f(\bs\theta_N)\notin\mathcal{N}\ \text{and}\  f(\bs\theta_n)>\frac{L_{\mathrm{avg}}a\sigma^2_\infty d}{\mu},\ \forall n<N\right\}\leq \frac{f(\bs{\theta}_{1})}{\delta_f}.
\]
\end{lemma}

The proofs of Lemma \ref{lemma: stability of GD} and \ref{lemma: stability of RCD} are provided in Appendices \ref{sec: appendix B} and \ref{sec: appendix C}, respectively. The core concept of these proofs is based on the construction of a specialized  supermartingale and the utilization of Markov's inequality. For example, to prove  Lemma \ref{lemma: stability of GD}, we define a stochastic process
\[
V_n=\left\{
\begin{aligned}
&f(\bs\theta_n)\mathbb{I}_n,\quad n<\tau\\
&f(\bs\theta_\tau)\mathbb{I}_\tau,\quad n\geq\tau
\end{aligned}\right.\,.
\]
where the indicator random variable is given by,
\[
I_n=\left\{
\begin{aligned}
&1,\quad \text{if}\quad \{\bs{\theta}_k\}^{n-1}_{k=1}\subset f^{-1}([0,\delta_f))\\
&0,\quad \text{otherwise}.
\end{aligned}\right.\,, 
\]
and the stopping time \[
\tau=\inf \left\{k:f(\bs\theta_k)\leq \frac{La\sigma^2_\infty d}{\mu}\right\}\,.
\]
We observe that $V_n$ is a meticulously crafted supermartingale, allowing us to distinguish between stable and unstable events. In particular, we demonstrate that if $\bs\theta_n$ exits the basin before it reaches $f(\bs\theta_n)=\frac{La\sigma^2_\infty d}{\mu}$ (an unstable event), then $\sup_n V_n\geq \delta_f$. 
Therefore, we can employ $V_n$ as a categorizer and the probability of failure of GD can be characterized by the value of $V_n$. More specifically, 
\[
    \mathbb{P}\left\{\exists N>0\ \text{such that}\ f(\bs\theta_N)\notin\mathcal{N}\ \text{and}\  f(\bs\theta_n)>\frac{La\sigma^2_\infty d}{\mu},\ \forall n<N\right\}\leq \mathbb{P}\left\{\sup_{n}V_n\geq \delta_f\right\}\,.
\]
Except for its use as a categorizer, we have designed $V_n$ in such a way that it is a supermartingale, meaning $\mathbb{E}(V_{n+1}|\bs{\theta}_{k\leq n})\leq V_{n}$. Therefore, we can use Markov's inequality for supermartingales to bound the supremum of $V_n$ and achieve the desired result.

\subsection{Convergence analysis}\label{sec:convergence analysis}
 
In this section, we present the convergence properties of noisy GD and RCD methods. 
It is important to note that Theorems \ref{thm:GD_informal} and \ref{thm:RCD_informal} directly follow from Theorems \ref{thm:GD_noisy} and \ref{thm:RCD_noisy}, respectively.

Our first theorem shows the convergence performance of the noisy GD method,
\begin{theorem}\label{thm:GD_noisy} Assume $f$ is a $L$-smooth function that satisfies Assumption \ref{as: local PL} and $g$ satisfies 
 \cref{as:noise_bound}. Given any precision $0<\epsilon<\delta_f$, the initial guess $f(\bs{\theta}_1)<\delta_f$, and the probability of failure $\eta\in\left(\frac{f(\bs\theta_1)}{\delta^{-1}_f},1\right)$, we choose the learning rate in \eqref{eq:def_GD}
\[
a_n=a=\mathcal{O}\left(\min\left\{\frac{1}{L},\frac{\mu \epsilon}{L\sigma^2_\infty d}\right\}\right),
\]
and the total number of iterations
\[
N=\Omega\left(\frac{1}{\mu a\eta}\log\left(\frac{f(\bs{\theta_1})}{\left(\eta-\frac{f(\bs\theta_1)}{\delta_f}\right)\epsilon}\right)\right)\,.
\]
Then, with probability $1-\eta$, we can find at least one $\bs\theta_m$ with $1\leq m\leq N$ such that $f(\bs\theta_m)\leq \epsilon$. In particular,
\[
\mathbb{P}\left\{\exists m\leq N,\ f(\bs{\theta}_m)\leq\epsilon \right\}\geq 1-\eta,
\]

\end{theorem}

Next, we state the convergence property of the noisy RCD method in the following theorem,
\begin{theorem}\label{thm:RCD_noisy}
 Assume $f$ is a $L$-smooth function that satisfies Assumption \ref{as: local PL} and $g$ satisfies Assumption \ref{as:noise_bound}. Given any precision $0<\epsilon<\delta_f$, the initial guess $f(\bs{\theta}_1)<\delta_f$, and the probability of failure $\eta\in\left(\frac{f(\bs\theta_1)}{\delta^{-1}_f},1\right)$, we choose the learning rate in \eqref{eq:def_noisy_RGD}
\[
a_n=a=\mathcal{O}\left(\min\left\{\frac{1}{L_{\max}},\frac{d}{\mu},\frac{\mu \epsilon}{L_{\mathrm{avg}}\sigma^2_\infty d}\right\}\right),
\]
and the total number of iterations
\[
N=\Omega\left(\frac{d}{\mu a\eta}\log\left(\frac{f(\bs{\theta_1})}{\left(\eta-\frac{f(\bs\theta_1)}{\delta_f}\right)\epsilon}\right)\right)\,.
\]
Then, with probability $1-\eta$, we can find at least one $\bs\theta_m$ with $1\leq m\leq N$ such that $f(\bs\theta_m)\leq \epsilon$. In particular,
\[
\mathbb{P}\left\{\exists m\leq N,\ f(\bs{\theta}_m)\leq\epsilon \right\}\geq 1-\eta,
\]

\end{theorem}

The proofs of these theorems can be found in the Appendix \ref{sec: appendix D}. 

\begin{remark}
We emphasize that \cref{thm:GD_noisy} and \cref{thm:RCD_noisy} are general convergence results that require only mild conditions. Specifically, \cref{thm:GD_noisy} can be used to demonstrate the stability and convergence of the traditional SGD algorithm when the right assumptions are in place. A convergence result analogous to the one previously discussed has been investigated in~\cite[Theorem 7]{liu2022loss}, where the authors impose a more stringent requirement on the cost function~\cite{bassily2018exponential}. In our work, we demonstrate the convergence of noisy GD using more sophisticated techniques in probability theory and adopt a weak version of probabilistic convergence~\cite{kleinberg2018alternative}. In addition, our approach can be directly extended to show the convergence of noisy RCD as in \cref{thm:RCD_noisy}, which to the best of our knowledge, has not been established before. These two theorems suggest that with a high probability, the loss function can achieve a small loss during the training process. In other words, it is likely that the parameter $\bs\theta$ remains in the basin $\mathcal{N}$ until the precision $\epsilon$ is attained at some point. After that, the optimization algorithm could diverge unless a certain strategy is applied, for example, a schedule of decreasing learning rates or an early stopping criterion.

\end{remark}

\begin{remark}
    Our theoretical result clarifies a relation between the learning rate and the desired precision in optimization. For example, the precision $\epsilon$ is manifested in the upper bounds of the learning rates in \cref{thm:GD_noisy} and \cref{thm:RCD_noisy}. Thus, to reach precision $\epsilon$, it is suggested to use an $\mathcal{O}(\epsilon)$ learning rate. Otherwise, due to the stability issue, the trajectory is no longer guaranteed to converge to the precision with positive probability. 
\end{remark}
We present the roadmap for proving Theorem \ref{thm:GD_noisy} as follows: Define the stopping time 
\[
\tau=\inf \left\{k:f(\bs\theta_k)\leq \epsilon\right\}\,.
\]
To prove Theorem \ref{thm:GD_noisy}, it suffices to demonstrate that the probability of failure $\mathbb{P}(\tau>N)$ is small. Since the learning rate $a_n$ is selected to be sufficiently small and, according to the lemma \ref{lemma: stability of GD}, it is likely that $\bs\theta_n$ will remain within the basin until the $\epsilon$ loss is achieved\footnote{Rigorously, we must also take into account the possibility that the optimization algorithm does not reach $\epsilon$ loss in a finite number of iterations.}. Thus, informally, it suffices for us to assume $\bs\theta_n\in\mathcal{N}$. The next step is to find an upper bound for the probability of failure $p_{\mathrm{fail}}=\mathbb{P}(\tau>N)$. Using the local PL condition, we can show that when $\epsilon<f(\bs\theta_n)<\delta_f$, 
\[
\mathbb{E}(f(\bs{\theta}_{n+1})|\bs{\theta}_n)\leq \left(1-\frac{\mu a}{2}\right)f(\bs{\theta}_n)\,,
\]
meaning that the conditional expectation of $f(\bs{\theta}_{n+1})$ decays to zero with rate $\left(1-\frac{\mu a}{2}\right)$. Inspired by this observation, we can construct a supermartingale to show that, if $\tau>N$, then, with high probability, we have $\inf_{1\leq n\leq N}f(\bs\theta_n)\leq \epsilon$. We note that this event is complementary to the failure event $\{\tau>N\}$. Consequently, we obtain an upper bound for $p_{\mathrm{fail}}$.

\section{Numerical results}\label{sec:numerical_2}

In \cref{sec:numerical_1}, depicted in \cref{fig:vqe}, we have demonstrated that the noisy RCD leads to faster convergence than the noisy GD for VQE problems. In this section, we extend our investigation to gauge the efficiency of noisy RCD applied to various other variational quantum algorithms, especially those involving nonconvex optimization problems. The implementation of these algorithms is executed on classical computers. To emulate quantum measurement noise, the partial derivatives undergo perturbation through additive Gaussian noise as outlined in \cref{sec:problem_setup}~\footnote{
The derivative with noise is computed by adding Gaussian noise to the original derivative: $\partial_i f(x) \leftarrow \partial_i f(x) + \epsilon$, where $\epsilon$ follows a Gaussian distribution, denoted as $\mathcal N(0, \sigma)$. In this notation, $\sigma$ signifies the standard deviation, defining the intensity of the Gaussian noise.}. Subsequently, we substantiate this approximation through a numerical experiment on a quantum simulator. This experiment further also proposes suitable values for the strength of the Gaussian noise that we will introduce in the upcoming numerical tests to appropriately mimic the measurement noise. 

In the experiment presented in \cref{sec:alternativevqe}, we utilize Qiskit-0.44.1~\cite{Qiskit}.

The algorithms for subsequent examples are implemented using Numpy~\cite{harris2020array} and Jax~\cite{bradbury2018jax}. We conducted each experiment 10 times, employing different random initializations for each run. All tests are executed on an Intel Xeon CPU @ 2.20GHz, complemented by a T4 GPU.

\subsection{Analyzing the noise distribution}\label{sec:noise_ana}

Building on the numerical experiment detailed in \cref{sec:numerical_1} and executed in Qiskit, we analyze the statistics of the partial derivatives derived from the quantum circuit.  
\cref{fig:vqe_noise} showcases the histograms representing 10000 estimates of partial derivatives with respect to the initial 12 directions, while the histograms for the remaining directions are presented in \cref{fig:noise1} and \cref{fig:noise2} in \cref{sec:noise}. Each estimate of the partial derivatives is averaged over 1000 shots.
From all histograms, we can clearly see that the distribution is closely approximated by a Gaussian distribution. 
In addition, the magnitude of the standard deviation of partial derivative estimates across all directions is comparable. These observations support assumptions of the noise model in \cref{problem_1}. For simplicity, we will employ the Gaussian noise model in our subsequent investigations to compare the performance of the noisy GD and RCD methods.

\begin{figure}[htbp]
   \centering
   \begin{subfigure}[b]{0.3\textwidth}
    \centering
    \includegraphics[width=\textwidth]{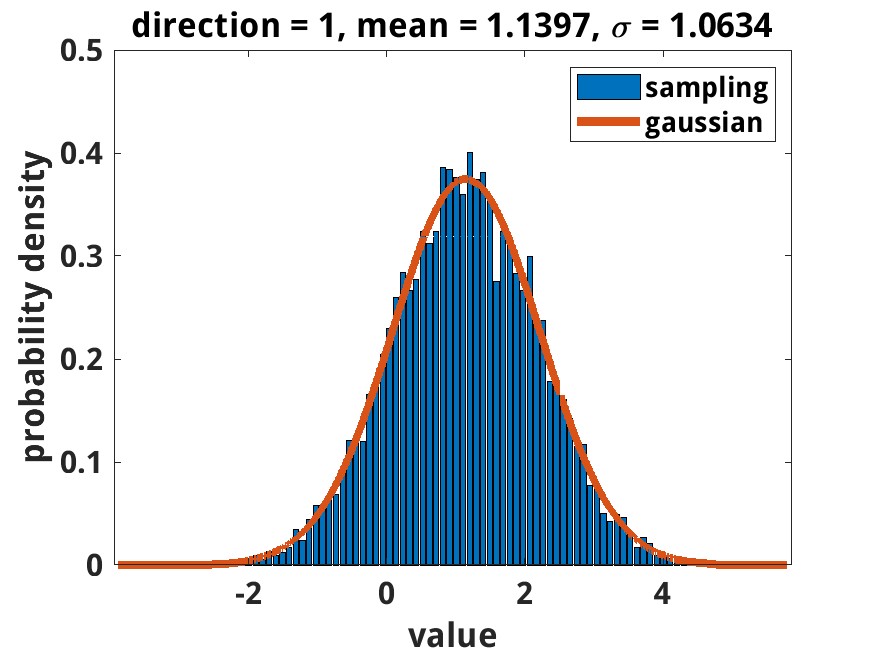}
    \end{subfigure}
   \begin{subfigure}[b]{0.3\textwidth}
    \centering
    \includegraphics[width=\textwidth]{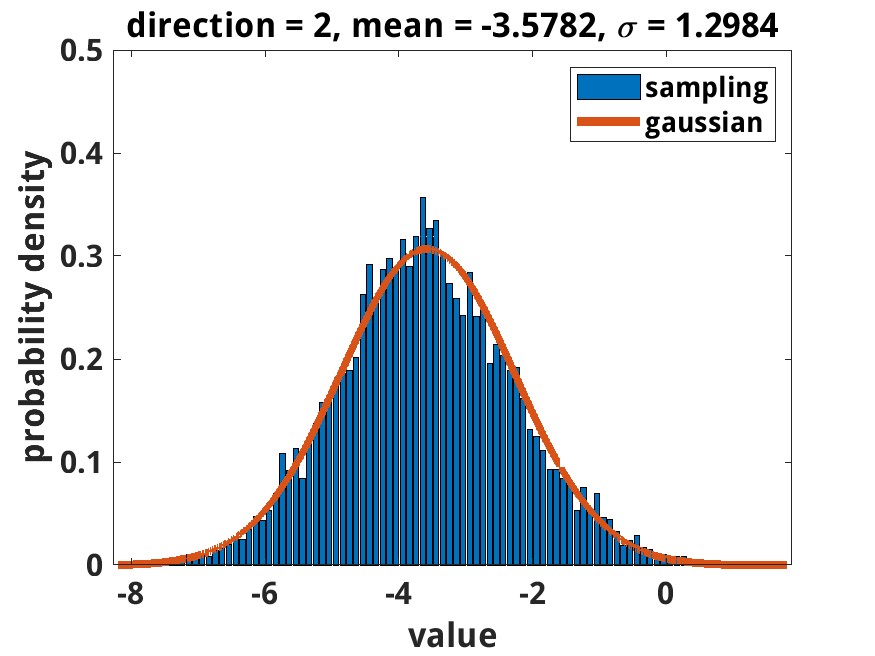}
    \end{subfigure}
    \centering
   \begin{subfigure}[b]{0.3\textwidth}
    \centering
    \includegraphics[width=\textwidth]{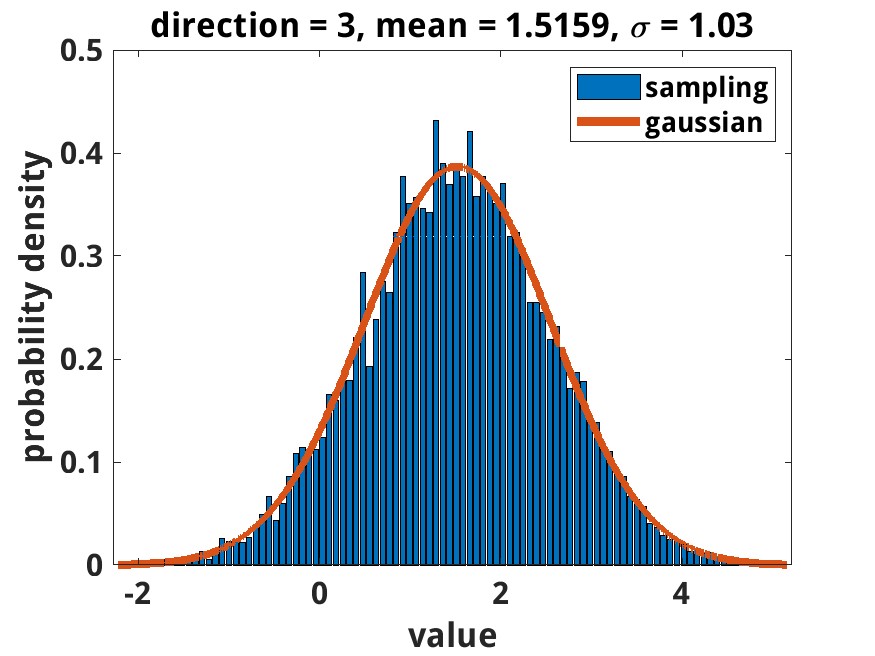}
    \end{subfigure}
    
    \begin{subfigure}[b]{0.3\textwidth}
    \centering
    \includegraphics[width=\textwidth]{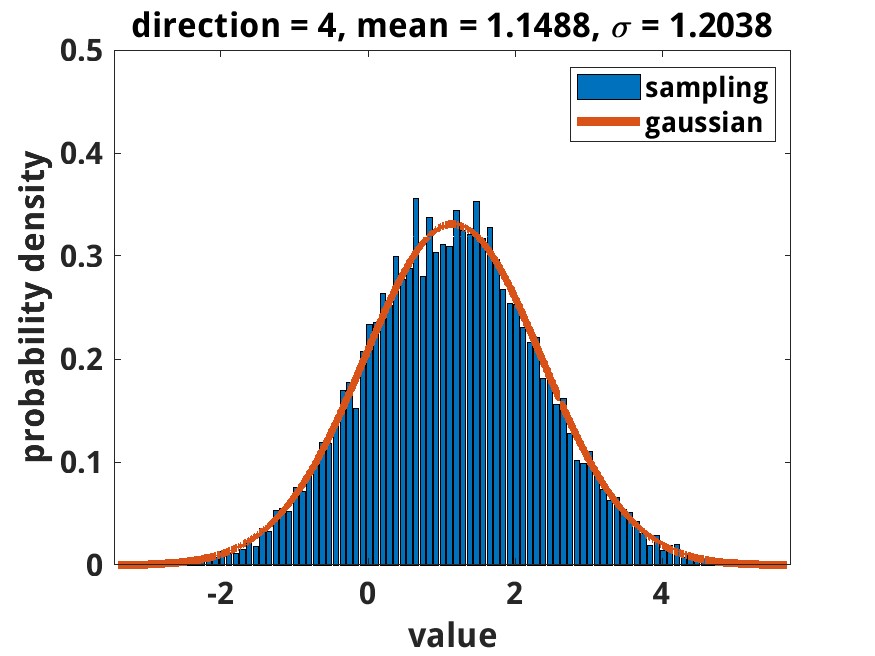}
    \end{subfigure}
   \begin{subfigure}[b]{0.3\textwidth}
    \centering
    \includegraphics[width=\textwidth]{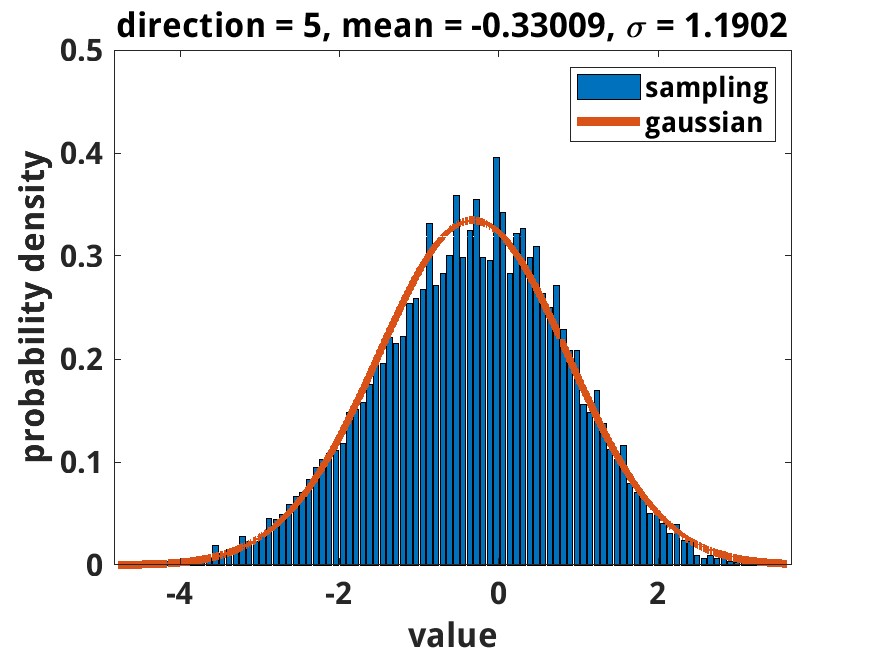}
    \end{subfigure}
    \centering
   \begin{subfigure}[b]{0.3\textwidth}
    \centering
    \includegraphics[width=\textwidth]{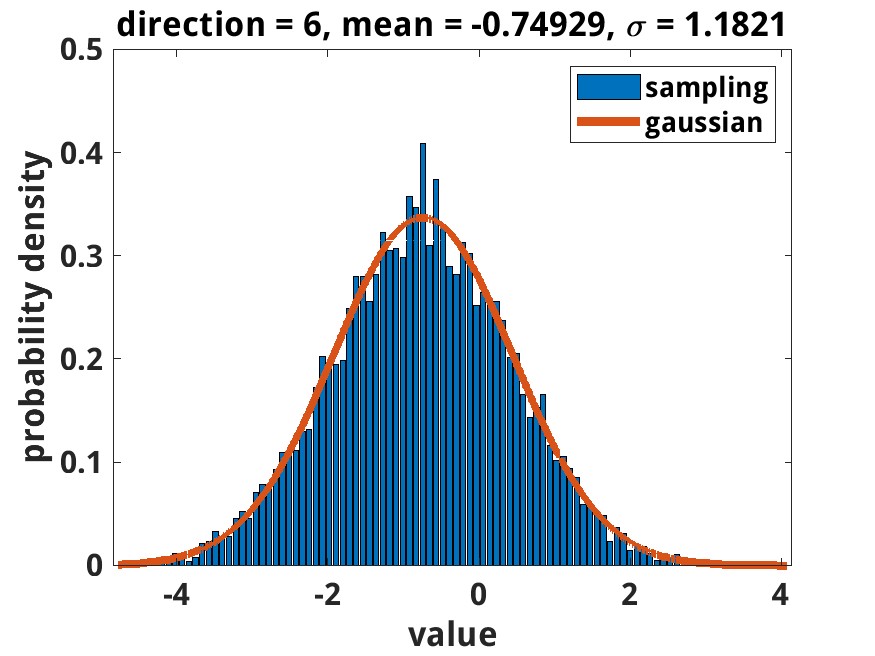}
    \end{subfigure}
    
    \begin{subfigure}[b]{0.3\textwidth}
    \centering
    \includegraphics[width=\textwidth]{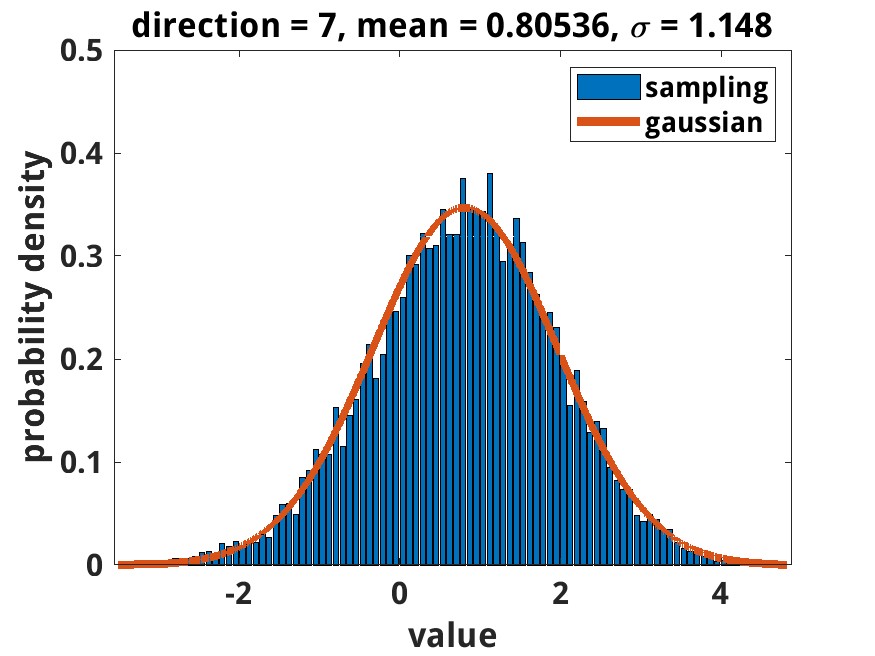}
    \end{subfigure}
   \begin{subfigure}[b]{0.3\textwidth}
    \centering
    \includegraphics[width=\textwidth]{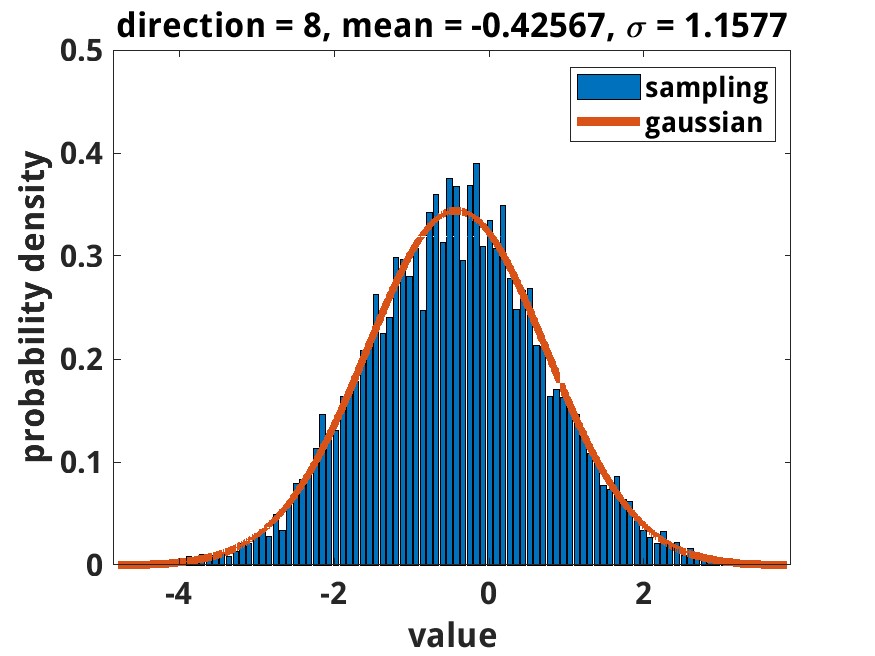}
    \end{subfigure}
    \centering
   \begin{subfigure}[b]{0.3\textwidth}
    \centering
    \includegraphics[width=\textwidth]{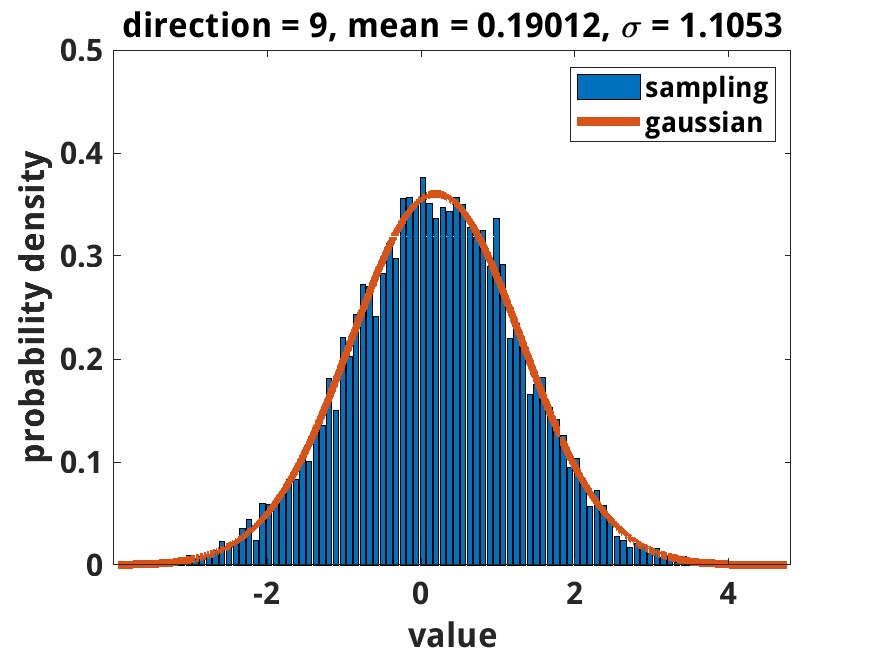}
    \end{subfigure}
    
    \begin{subfigure}[b]{0.3\textwidth}
    \centering
    \includegraphics[width=\textwidth]{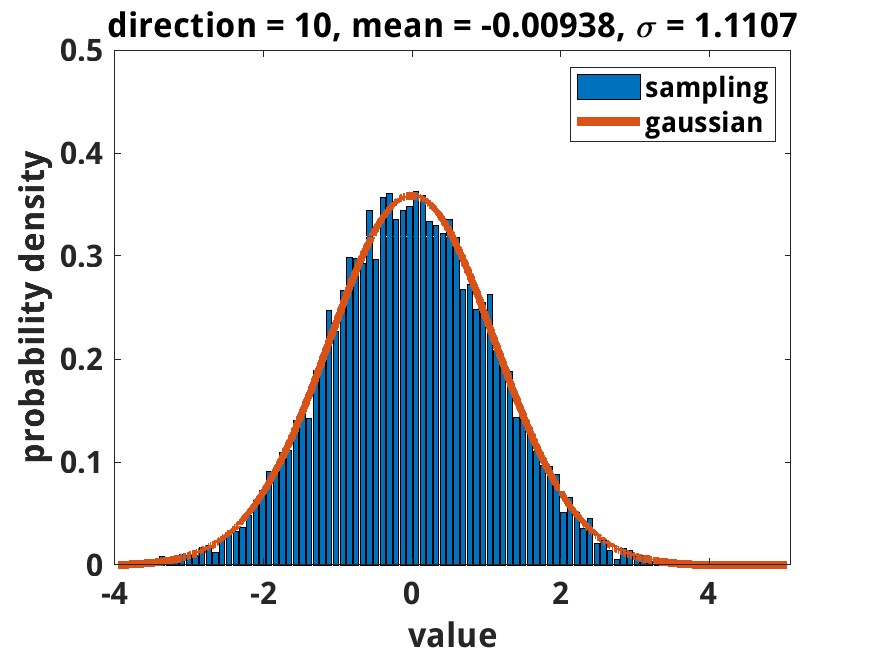}
    \end{subfigure}
   \begin{subfigure}[b]{0.3\textwidth}
    \centering
    \includegraphics[width=\textwidth]{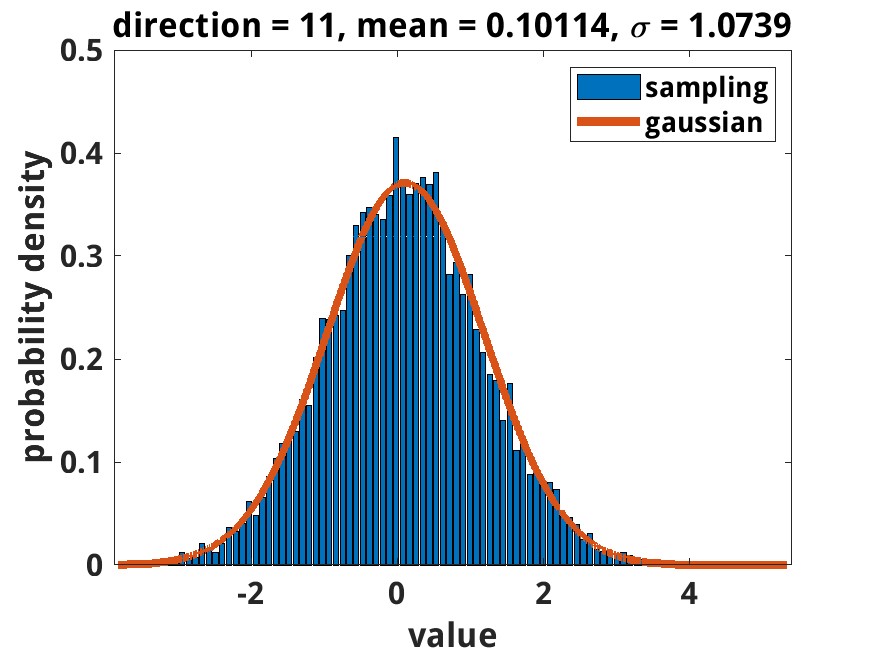}
    \end{subfigure}
    \centering
   \begin{subfigure}[b]{0.3\textwidth}
    \centering
    \includegraphics[width=\textwidth]{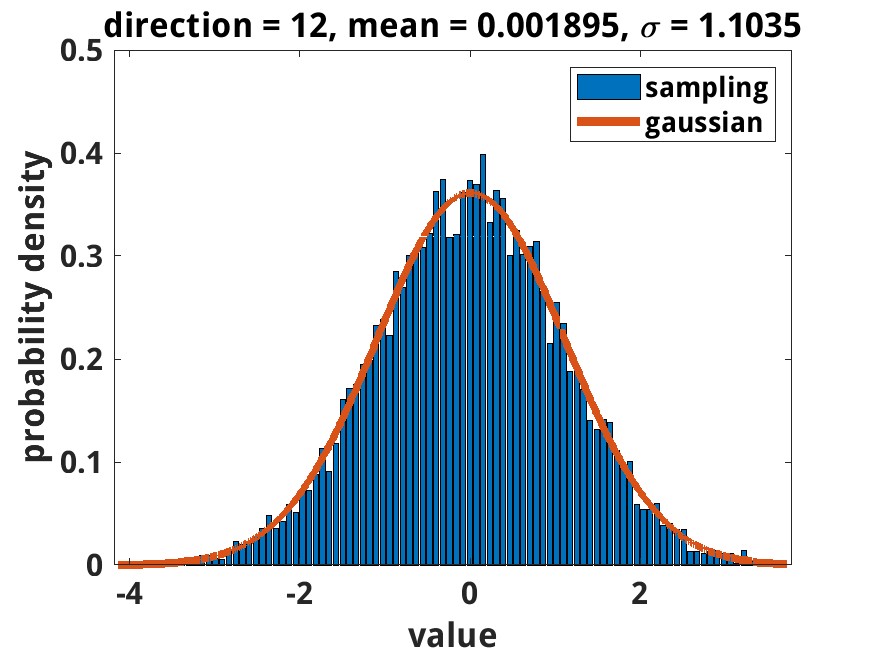}
    \end{subfigure}

    \caption{Histograms of the estimated partial derivatives:  Each panel displays the histogram of 10000 partial derivative estimates in one of the first 12 directions, which are obtained by applying the parameter-shift rule for the ansatz in \cref{fig:circuit}. The sampling of the partial derivatives is carried out at a suboptimal point chosen from one simulation used in \cref{fig:vqe}, where the fidelity is about 0.883.}
    \label{fig:vqe_noise}
\end{figure}

In the next two sections, we conduct a comprehensive comparison between noisy RCD and GD across a broad spectrum of variational quantum algorithms and applications.

\subsection{VQE with a varied circuit structure}\label{sec:alternativevqe}

In \cref{sec:numerical_1}, we utilize the VQE for the TFIM \eqref{eq:isingmodel} employing both the noisy GD and the noisy RCD. In this section, we tackle the same optimization task but with a modified setup.
Specifically,  \cref{fig:circuit_alt} depicts the PQC~\cite{Kandala_2017} utilized in the experiments showcased in \cref{fig:twolocal}, distinct from that presented in \cref{fig:circuit}.

In the experiments illustrated in \cref{fig:twolocal}, each optimization outcome derives from 10 identical simulations with the same initial condition. We set the learning rates for the RCD and GD at $0.3$ and $0.05$, respectively. Each experiment utilizes 10000 shots, with $18$ trainable parameters. Results shown in \cref{fig:twolocal} demonstrate that, compared to GD, RCD requires nearly three times fewer partial derivative evaluations to converge.

\begin{figure}
\centerline{

\scalebox{0.8}{

\Qcircuit @C=1.0em @R=0.2em @!R { \\
                \nghost{{q}_{0} :  } & \lstick{{q}_{0} :  } & \gate{\mathrm{R_Y}\,(\mathrm{{\ensuremath{\theta}}[0]})} & \gate{\mathrm{R_Z}\,(\mathrm{{\ensuremath{\theta}}[3]})} & \ctrl{1} & \ctrl{2} & \gate{\mathrm{R_Y}\,(\mathrm{{\ensuremath{\theta}}[6]})} & \gate{\mathrm{R_Z}\,(\mathrm{{\ensuremath{\theta}}[9]})} & \qw & \ctrl{1} & \ctrl{2} & \gate{\mathrm{R_Y}\,(\mathrm{{\ensuremath{\theta}}[12]})} & \gate{\mathrm{R_Z}\,(\mathrm{{\ensuremath{\theta}}[15]})} & \qw & \qw & \qw\\
                \nghost{{q}_{1} :  } & \lstick{{q}_{1} :  } & \gate{\mathrm{R_Y}\,(\mathrm{{\ensuremath{\theta}}[1]})} & \gate{\mathrm{R_Z}\,(\mathrm{{\ensuremath{\theta}}[4]})} & \targ & \qw & \ctrl{1} & \gate{\mathrm{R_Y}\,(\mathrm{{\ensuremath{\theta}}[7]})} & \gate{\mathrm{R_Z}\,(\mathrm{{\ensuremath{\theta}}[10]})} & \targ & \qw & \ctrl{1} & \gate{\mathrm{R_Y}\,(\mathrm{{\ensuremath{\theta}}[13]})} & \gate{\mathrm{R_Z}\,(\mathrm{{\ensuremath{\theta}}[16]})} & \qw & \qw\\
                \nghost{{q}_{2} :  } & \lstick{{q}_{2} :  } & \gate{\mathrm{R_Y}\,(\mathrm{{\ensuremath{\theta}}[2]})} & \gate{\mathrm{R_Z}\,(\mathrm{{\ensuremath{\theta}}[5]})} & \qw & \targ & \targ & \gate{\mathrm{R_Y}\,(\mathrm{{\ensuremath{\theta}}[8]})} & \gate{\mathrm{R_Z}\,(\mathrm{{\ensuremath{\theta}}[11]})} & \qw & \targ & \targ & \gate{\mathrm{R_Y}\,(\mathrm{{\ensuremath{\theta}}[14]})} & \gate{\mathrm{R_Z}\,(\mathrm{{\ensuremath{\theta}}[17]})} & \qw & \qw\\
\\ }}
}
\caption{A variational circuit ansatz is employed for the Transverse-field ising Model expressed in Eq. \eqref{eq:isingmodel}, utilizing three qubits. This circuit is a parameterized construct comprised of alternating rotation and entanglement layers. Each rotation layer involves the application of single qubit gates, specifically Rotation-y and Rotation-z gates, to all qubits. In contrast, the entanglement layer employs two-qubit gates, namely the controlled-X gate, to facilitate entanglement among the qubits. The ansatz is designated with $18$ parameters.}
\label{fig:circuit_alt}
\end{figure}
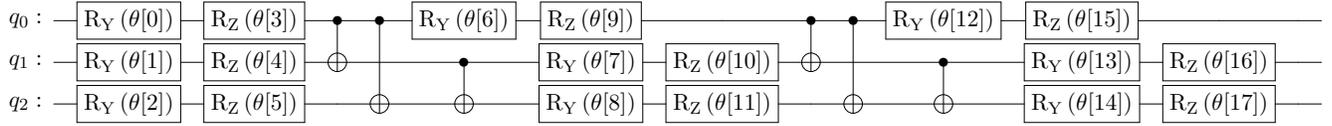

\begin{figure}[htbp]
   \centering
   \begin{subfigure}[b]{0.32\textwidth}
    \centering
    \includegraphics[width=\textwidth]{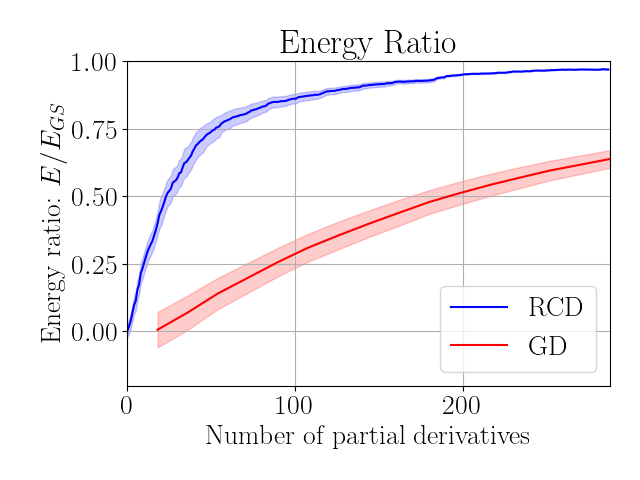}
    \end{subfigure}
    \begin{subfigure}[b]{0.32\textwidth}
    \centering
    \includegraphics[width=\textwidth]{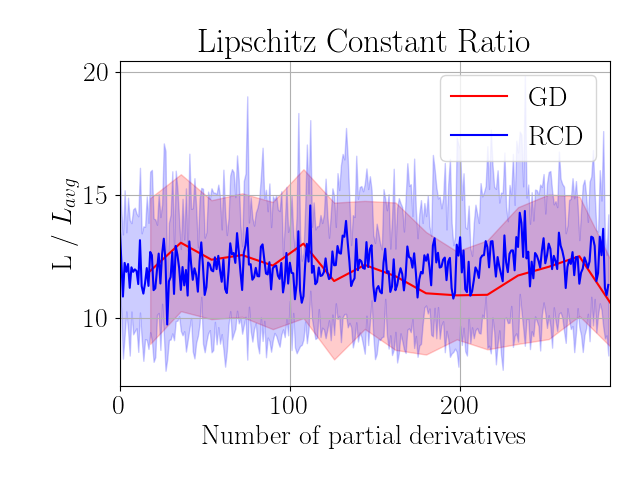}
    \end{subfigure}
    \begin{subfigure}[b]{0.32\textwidth}
    \centering
    \includegraphics[width=\textwidth]{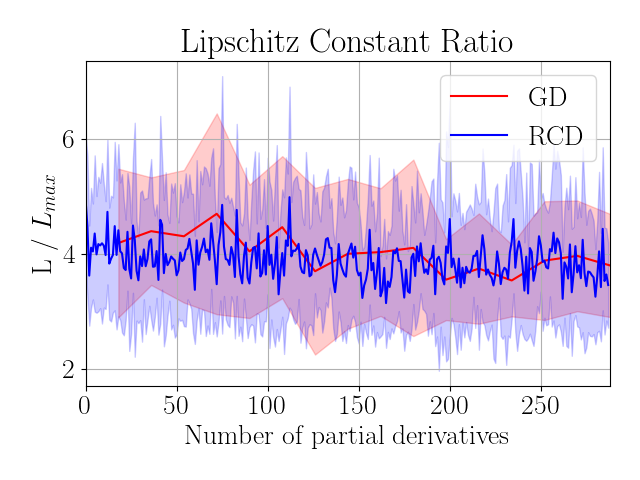}
    \end{subfigure}
    \caption{ \label{fig:twolocal} Performance comparison between GD (red) and RCD (blue) in terms of energy ratio and Lipschitz constant ratios for optimizing the Hamiltonian  \eqref{eq:isingmodel}. The energy ratio $ E / E_{GS}$ is presented in the left panel, while the Lipschitz constant ratios, denoted as $\frac{L}{L_{\text{avg}}}$ and $\frac{L}{L_{\max}}$, are shown in the middle and right panels respectively. The shaded areas in each panel represent variations observed across multiple trials. }
\end{figure}

\color{black}
\color{black}

\subsection{Quantum Approximate Optimization Algorithm (QAOA) for quantum Hamiltonians}
The quantum approximate optimization algorithm (QAOA) 
\cite{farhi2014quantum}, originally devised for solving combinatorial problems, is a leading example for demonstrating quantum advantage
on near-term quantum computers. As introduced in \cite{farhi2014quantum}, the QAOA utilizes a PQC, which naturally enables optimization through the variational quantum algorithm. 

In a generalized QAOA model, we begin with an initial quantum state $\ket{\psi_i}$, which can be easily prepared in experiments, and let it evolve by a parameterized unitary transformation,
\begin{equation}
    \ket{\psi(\bs\alpha,\bs\beta)} = U(\{\alpha_j, \beta_j\}_{j=1}^p) \ket{\psi_i} = e^{-i H_2 \beta_p} e^{-i H_1 \alpha_p} \cdots e^{-i H_2 \beta_1} e^{-i H_1 \alpha_1} \ket{\psi_i}
    \label{eqn:qaoawf},
\end{equation}where the vector $\bs\alpha$ (or $\bs\beta$) enumerates the parameters $\alpha_j$ (or $\beta_j$), and thus the total number of parameters is $2p$ and the unitary transformation alternates between two kinds of parameterized unitary transformations.
With this ansatz, the optimization is performed with the parameters $\{\alpha_j,\beta_j\}$  associated with the application-dependent Hamiltonian matrices $H_1$ and $H_2$, respectively. 

In the subsequent sections, we will consider optimization problems based on the QAOA \eqref{eqn:qaoawf}. We will conduct a comparative analysis of the noisy GD and RCD for various QAOA models that will span a range of systems, including the Ising model (refer to Section~\ref{sec:ising}), the Heisenberg model (refer to Section~\ref{sec:heisenberg}), and variational quantum factoring (refer to Section~\ref{sec:factor}).

\subsubsection{QAOA -- Ising Model} \label{sec:ising}

In this section, we parametrize the transverse-field Ising model by a  Hamiltonian 
\begin{equation}\label{eq: ising qaoa}
   H[h]= \sum_{j=1}^{N-1}Z_{j+1}Z_j + \sum_{j=1}^{N}(Z_j +  h X_j), 
\end{equation} 
where $N$ denotes the total number of qubits. The global control field $h\in\{\pm 4\}$ takes two discrete values, corresponding to the two alternating QAOA generators $H_1=H[-4]$ and $H_2=H[+4]$~\cite{PhysRevX.8.031086, yao2020policy}. The initial state \( \ket{\psi_i} \) corresponds to the ground state of \( H[-2] \), while the desired target state \( \ket{\psi_\ast} \) is selected as the ground state of \( H[+2] \). The variational problem aims to optimize the fidelity~\footnote{Fidelity serves as a metric for optimization. However, one caveat of utilizing fidelity is its reliance on the ground state. In this context, we assume the presence of an oracle capable of producing the fidelity value. Subsequently, we also employ energy as an observable metric for optimization purposes.}, 
\begin{eqnarray}
    \max_{\{\alpha_i, \beta_i\}_{i=1}^p} \mathcal{F}(\{\alpha_i, \beta_i\}_{i=1}^p) = \max_{\{\alpha_i, \beta_i\}_{i=1}^p}\abs{ \mel{\psi_{*}}{ U(\{\alpha_i, \beta_i\}_{i=1}^p)}{ \psi_i} }^2,     
\label{eqn:qaoa_HM}
\end{eqnarray} 
where,
\begin{equation}
 U(\{\alpha_i, \beta_i\}_{i=1}^p) \ket{\psi_i}= e^{-i H_2 \beta_p} e^{-i H_1 \alpha_p} \cdots e^{-i H_2 \beta_1} e^{-i H_1 \alpha_1}\ket{\psi_i}.
\label{eqn:qaoa_HM_psi}
\end{equation}We note that the fidelity optimization \eqref{eqn:qaoa_HM} is equivalent to the optimization of the form \eqref{eq:vqe} by letting the Hamiltonian be $\ketbra{\psi_*}$.  

In the numerical test, we choose a system from \eqref{eq: ising qaoa} with three qubits ($N=3$), and  then apply both GD and RCD methods in the optimization. 
Figure~\ref{fig:qaoa} shows the optimization results obtained from the noisy GD and RCD with the respective learning rates of 0.0045 and 0.015 by using an ansatz defined with 20 parameters. By adjusting the learning rate and tracking the stability, We observe that RCD permits a larger learning rate in comparison to GD, while maintaining the stability. Similarly to the results presented in \cref{fig:vqe}, we compare the performance of the two methods in terms of the number of partial derivative evaluations. From \cref{fig:qaoa}, We observe that noisy RCD converges much faster than noisy GD. While RCD achieves a fidelity near 1 with 500 partial derivative evaluations, GD only attains a fidelity below 0.25 with an equivalent number of evaluations.
This computational effectiveness of RCD can be attributed to the large ratios of Lipschitz constants shown in \cref{fig:qaoa}, which are obtained along the optimization trajectories.

\begin{figure}[htbp]
   \centering
   \begin{subfigure}[b]{0.32\textwidth}
    \centering
    \includegraphics[width=\textwidth]{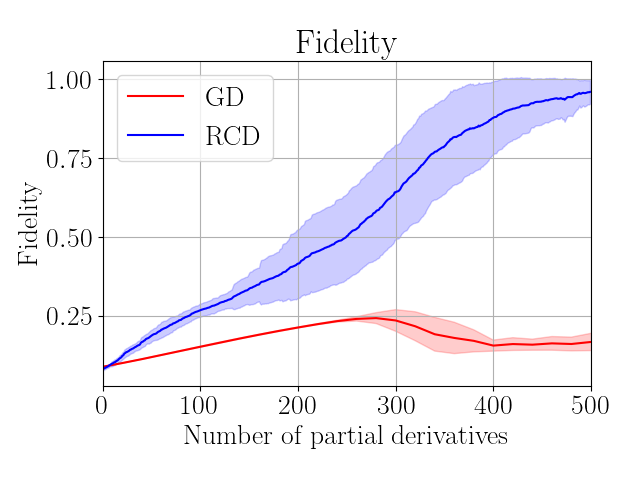}
    \end{subfigure}
    \begin{subfigure}[b]{0.32\textwidth}
    \centering
    \includegraphics[width=\textwidth]{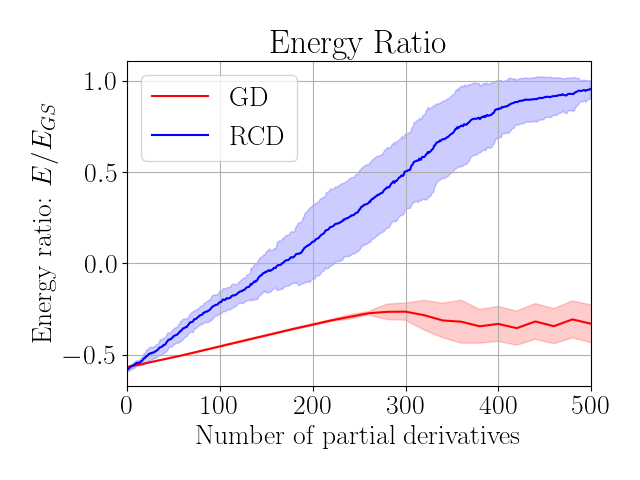}
    \end{subfigure}
    
    \begin{subfigure}[b]{0.32\textwidth}
    \centering
    \includegraphics[width=\textwidth]{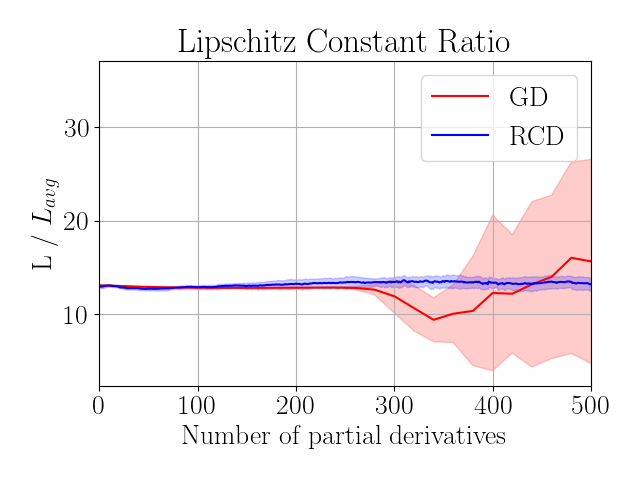}
    \end{subfigure}
    \begin{subfigure}[b]{0.32\textwidth}
    \centering
    \includegraphics[width=\textwidth]{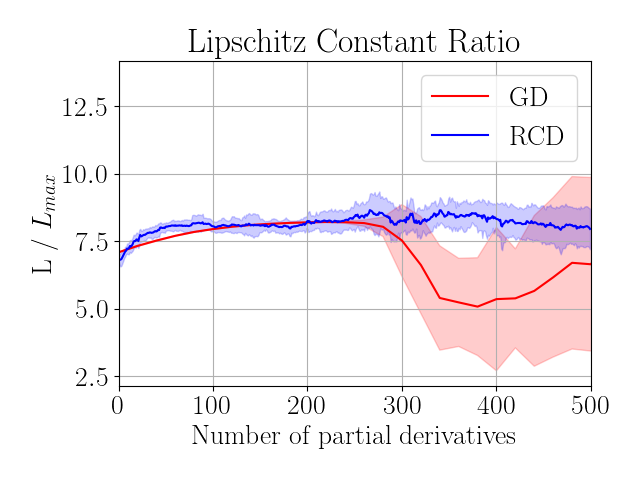}
    \end{subfigure}
    \caption{ \label{fig:qaoa} Performance comparison between the noisy GD and RCD for the Ising model \eqref{eq: ising qaoa}. The corresponding Lipschitz constant ratios, denoted as $\frac{L}{L_{\text{avg}}}$ and $\frac{L}{L_{\max}}$, are presented in the bottom figures. The shaded areas within the figures represent variations that have been observed across 10 random realizations. The optimization is performed for parameters with dimension equal to 20. 
}
\end{figure}

\subsubsection{QAOA -- Heisenberg Model} \label{sec:heisenberg}

Our second test problem with QAOA is the (anisotropic) spin-$1$ Heisenberg model, $H \!=\!H_1\!+\!H_2$, with the alternating Hamiltonians given by, \[H_1\!\!=\!\! J\sum_{j=1}^N (X_{j+1}X_j \!+\! Y_{j+1}Y_j), \quad H_2= \Delta \sum_{j=1}^N Z_{j+1}Z_j, \nonumber\] with anisotropic parameter $\Delta/J=0.5$ (topological/Haldane~\cite{chen2003ground,pollmann2010entanglement,langari2013ground, PhysRevX.11.031070}). For the Heisenberg model, we consider a system consisting of eight qubits ($N=8$) and choose the fidelity as a measure for optimization, similarly to the setup for the results in \cref{fig:qaoa}. We set the antiferromagnetic initial state to \(\ket{\psi_i} =  \ket{10101010} \). The target state  is the ground state of the Hamiltonian \(H \!=\! H_1 + H_2\). We employ the QAOA ansatz represented by Eq.~\eqref{eqn:qaoa_HM_psi} and carry out the fidelity optimization detailed in Eq.~\eqref{eqn:qaoa_HM}.

Figure~\ref{fig:heisenberg} showcases the performance outcomes  from noisy GD and RCD simulations with learning rates set to 0.01 and 0.1, respectively. This QAOA model involves  28 parameters. 
The fidelity result shows that RCD converges to the target state much faster than GD. This phenomenon can be elucidated by noting that the ratios of Lipschitz constants derived from both noisy methods, $\frac{L}{L_{\text{avg}}}$ and $\frac{L}{L_{\max}}$, average around 10 and 6 along the trajectories, respectively. Especially, the magnitude of the ratio $\frac{L}{L_{\max}}$ is similar to that of the ratio of the numbers of partial derivative evaluations to reach a high fidelity $>0.8$ from both noisy methods, as shown in \cref{fig:heisenberg}. Based on the observed numerical results, a high ratio of $\frac{L}{L_{\max}}$ is responsible for the  efficiency of RCD in this optimization problem.

\begin{figure}[htbp]
   \centering
   \begin{subfigure}[b]{0.32\textwidth}
    \centering
    \includegraphics[width=\textwidth]{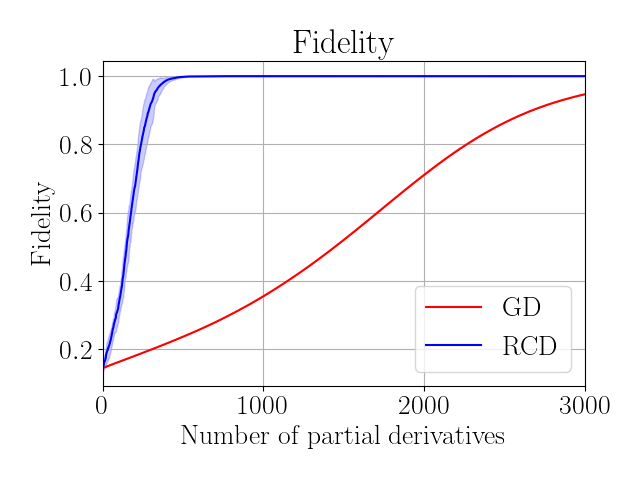}
    \end{subfigure}
    \begin{subfigure}[b]{0.32\textwidth}
    \centering
    \includegraphics[width=\textwidth]{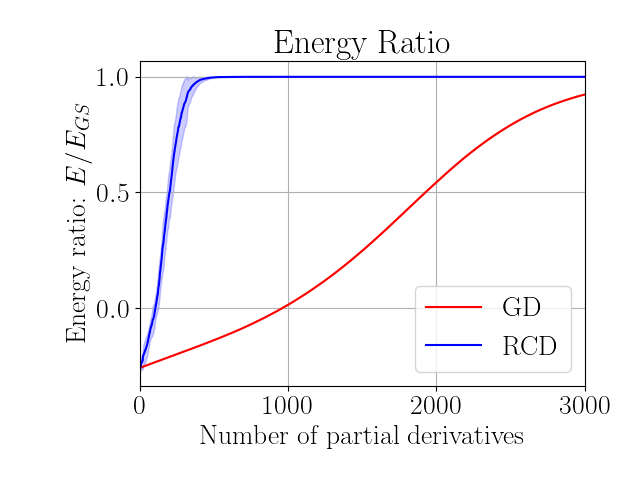}
    \end{subfigure}
    
       \centering
   \begin{subfigure}[b]{0.32\textwidth}
    \centering
    \includegraphics[width=\textwidth]{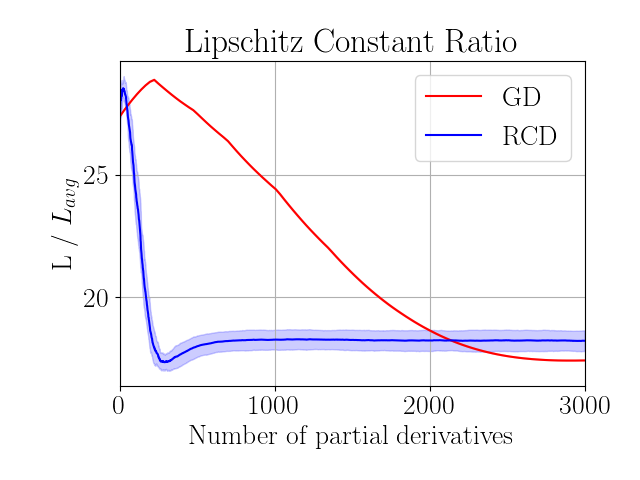}
    \end{subfigure}
   \begin{subfigure}[b]{0.32\textwidth}
    \centering
    \includegraphics[width=\textwidth]{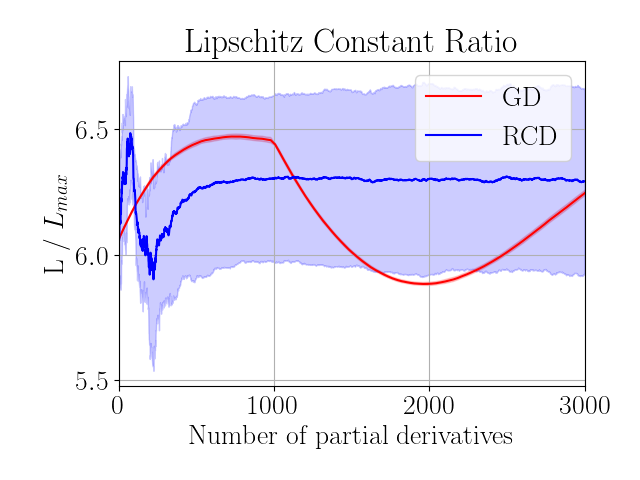}
    \end{subfigure}
    \caption{ \label{fig:heisenberg} Performance comparison between noisy GD and RCD for the Heisenberg model. The corresponding Lipschitz constant ratios, denoted as $\frac{L}{L_{\text{avg}}}$ and $\frac{L}{L_{\max}}$, are presented in the middle and right. The shaded areas within the figure represent variations that have been observed across 10 random realizations. The optimization is performed in dimensions of 28. }
\end{figure}

\subsection{QAOA for classical combinatorial optimization problems} 

Quadratic unconstrained binary optimization (QUBO) problems have significant applications in fields such as finance, logistics, and machine learning, etc. Recognized as one prominent optimization model in quantum computing, QUBO consolidates a wide range of combinatorial optimization problems~\cite{PhysRevA.104.062426,Kochenberger2014, anthony2017quadratic, glover2019quantum} and translates them into identifying the ground state of classical Ising models~\cite{lucas2014ising}.

The goal of QUBO is to identify a sequence of binary variables ($0$ or $1$) that minimize a quadratic function. Specifically, a cost function $f_{Q}$ is constructed over the set of binary vectors, $\mathbb{B}^{n}$:
\begin{equation}\label{eq:qubof}
    f_{Q}(x)=x^{\top} Q x=\sum_{i,j=1}^n Q_{i j} x_{i} x_{j}.
\end{equation}
In this context, $\mathbb{B}=\{0,1\}$ signifies the set of binary values (or bits), and $\mathbb{B}^{n}$ represents the collection of binary vectors with length $n>0$. A symmetric, real-valued matrix $Q \in \mathbb{R}^{n \times n}$ is introduced, with each element $Q_{i j}$ determining the weight for the corresponding pair of indices $i, j \in {1, \ldots, n}$. For example, if $i=j$, then the term $Q_{ii} x_{i}^2$ contributes $Q_{ii}$ to the function value when $x_i=1$. On the other hand, if $i\neq j$, then the term $Q_{ij} x_{i} x_{j}$ contributes $Q_{ij}$ to the function value when both $x_i=1$ and $x_j=1$. 

Overall, QUBO seeks to minimize the function $f_{Q}$ over the set of binary vectors by determining an optimal minimizer $x^*$,
\begin{equation}
x^{*}=\underset{x \in \mathbb{B}^{n}}{\arg \min } f_{Q}(x).    
\end{equation}
Incorporating the variational quantum algorithm into QUBO, we reformulate the cost function using the following substitution:
\begin{equation}
    x_i = \frac{1-Z_i}{2}\text{ or }\frac{1+Z_i}{2},
\end{equation}where the variable $x_i$ is supplanted by the Pauli Z matrix operating on the $i$-th qubit. This replacement facilitates the formulation of a model Hamiltonian. Its ground state can be approximated by minimizing the expected energy via the variational quantum algorithm, as elaborated in \cref{sec:maxcut}.

In the following sections, we evaluate the performance of the noisy GD and RCD across various QUBO applications, focusing on the ground state energy estimation. These applications encompass Max-Cut in \cref{sec:maxcut}, the traveling salesman problem in \cref{sec:tsp}, and variational quantum factoring in \cref{sec:factor}.

\subsubsection{Max-Cut} \label{sec:maxcut}

For the Max-Cut problem, the graph employed in our numerical experiments is presented as follows:

\begin{center}
\begin{tikzpicture}
[scale=.8,auto=left,every node/.style={circle,fill=blue!20}]

\node (n0) at (2,4) {0};
\node (n1) at (0,2) {1};
\node (n2) at (2,0) {2};
\node (n3) at (4,2) {3};

\foreach \from/\to in {n0/n1,n0/n2,n0/n3,n1/n2,n2/n3}
\draw (\from) -- (\to);
\end{tikzpicture}
\end{center}
The global cost function is designed to maximize $C = \sum_{(i,j) \in E} x_i (1 - x_j)$, where $E$ represents the edges in the graph. For the given graph, the QUBO problem can be formulated as:
\[
\min_{x_i\in\{0,1\}}-3x_{0}^{2} + 2x_{0}x_{1} + 2x_{0}x_{2} + 2x_{0}x_{3} - 2x_{1}^{2} + 2x_{1}x_{2} - 3x_{2}^{2} + 2x_{2}x_{3} - 2x_{3}^{2}\,.
\]
In order to construct the corresponding Hamiltonian, we associate the binary variables \( x_i \) with the Pauli \( Z \) matrices, denoted as \( Z_i \), which act on individual qubits. Taking into account the relationship between the binary variables \( x_i \) and the Pauli matrices \( Z_i \), defined by the equation \( x_i = \frac{{1 - Z_i}}{2} \), the cost function is articulated by the Hamiltonian:
\begin{equation}
H = \frac{1}{2} I - 3 Z_0 + \frac{1}{2} Z_0 Z_1 + \frac{1}{2} Z_0 Z_2 + \frac{1}{2} Z_0 Z_3 + \frac{1}{2} Z_1 Z_2 + \frac{1}{2} Z_2 Z_3.
\end{equation}

Using this Hamiltonian, we construct a parameterized quantum circuit with four qubits ($N=4$) and 20 parameters. The circuit consists of alternating single-gate rotations, denoted as $U_{\text {single }}(\bs\theta)=\prod_{i=1}^{n} \operatorname{RY}\left(\theta_{i}\right)$~\footnote{Each layer of rotation gates includes a rotation-Y gate applied to every qubit.} and entangler gate $U_{\text {entangler }}$~\footnote{The entanglement layer incorporates two-qubit gates for qubit entanglement without tunable parameters. In this experiment, the entangler gate employs controlled-Z gates. For a comprehensive explanation, refer to the circuit architecture in Appendix~\ref{sec:circuit_qubo}}. The configuration of the parametrized quantum circuit is illustrated in Figure~\ref{fig:circuit_qubo}. This structure resembles the variational quantum circuit of the QAOA, with the ansatz given by $|\psi(\bs\theta)\rangle=\left[U_{\text {single }}(\bs\theta) U_{\text {entangler }}\right]^{m}|+\rangle$. For the optimization process, we assign a learning rate of 0.1 for GD and 3.0 for RCD and select energy as the optimization metric.

As illustrated in \cref{fig:maxcut}, the RCD also outperforms GD in this case, as it converges to an energy ratio of 1 with roughly 200 partial derivative evaluations. In contrast, the GD achieves only an average of 0.75 with 1000 derivative evaluations. The superior performance of RCD in \cref{fig:maxcut} can again be attributed to the significant values of $\frac{L}{L_{\text{avg}}}$ and $\frac{L}{L_{\max}}$,  both exceeding an order of magnitude of 3. As observed from the optimization result, a high ratio of $\frac{L}{L_{\text{avg}}}$ is indicative of the rapid convergence of RCD in this application.

\begin{figure}[htbp]
   \centering
   \begin{subfigure}[b]{0.32\textwidth}
    \centering
    \includegraphics[width=\textwidth]{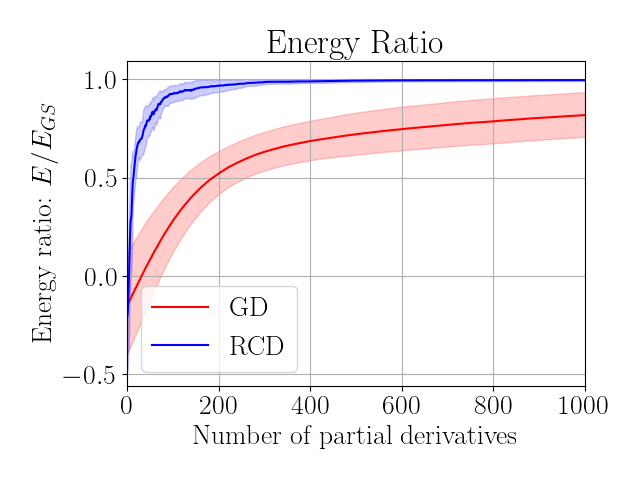}
    \end{subfigure}
   \begin{subfigure}[b]{0.32\textwidth}
    \centering
    \includegraphics[width=\textwidth]{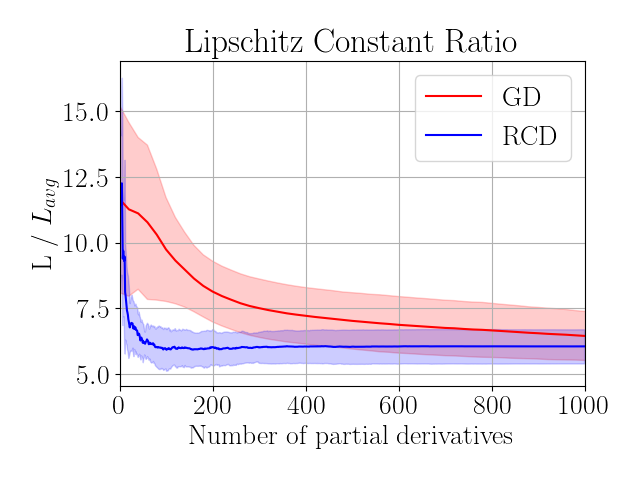}
    \end{subfigure}
   \begin{subfigure}[b]{0.32\textwidth}
    \centering
    \includegraphics[width=\textwidth]{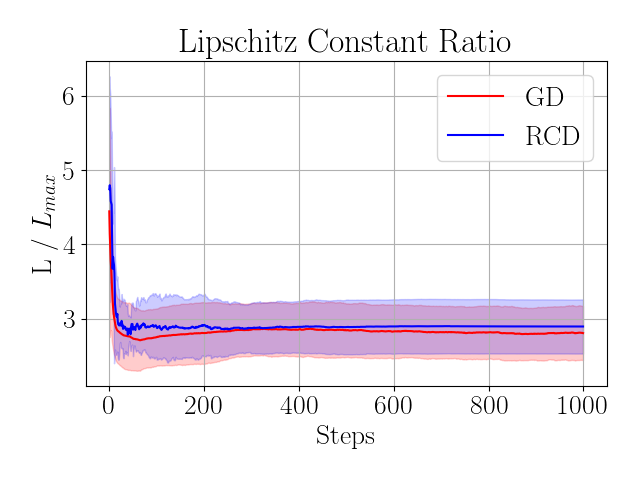}
    \end{subfigure}
    \caption{ \label{fig:maxcut} Performance comparison between noisy GD and RCD for the Max-Cut problem. The corresponding Lipschitz constant ratios, denoted as $\frac{L}{L_{\text{avg}}}$ and $\frac{L}{L_{\max}}$, are presented in the middle and right panels. The shaded areas within the figure represent variations that have been observed across 10 random realizations. The optimization process has been performed in 20 dimensions.}

\end{figure}

\subsubsection{Traveling Salesman Problem (TSP)} \label{sec:tsp}

We have designed a numerical test for the traveling salesman problem (TSP) using three cities as an example. The intercity costs for these cities are 48, 63, and 91 respectively.  The cost of TSP is defined as
\[
C(\mathbf{x})=\sum_{i, j} w_{i j} \sum_{p} x_{i, p} x_{j, p+1}+A \sum_{p}\left(1-\sum_{i} x_{i, p}\right)^{2}+A \sum_{i}\left(1-\sum_{p} x_{i, p}\right)^{2}\,,
\]
where $i$ labels the node, $p$ indicates its order, and $x_{i, p}$ is in the set $\{0, 1\}$ and the penalty parameter $A$ is set sufficiently large to effectively enforce constraints. More details regarding the expansion of $C(\textbf{x})$ can be found in \cref{sec:cost}.

Utilizing the defined cost function, we establish a model Hamiltonian in the same manner as presented in \cref{sec:maxcut}. We aim to prepare its ground state to address the QUBO problem. A detailed representation of the Hamiltonian is available in \cref{sec:cost}. We construct a parameterized quantum circuit comprising alternating single-gate rotations, represented by $U_{\text {single }}(\bs\theta)=\prod_{i=1}^{n} \operatorname{R_Y}\left(\theta_{i}\right)$ and entangler gate $U_{\text {entangler}}$. This circuit resembles the one depicted in Figure~\ref{fig:circuit_qubo}, albeit with a greater number of qubits. The total number of trainable parameters is 90, which requires nine qubits ($N=9$) and 10 alternating layers. We employ energy as the measure for the optimization cost function.

In the left panel in \cref{fig:tsp}, the optimization results obtained from the noisy RCD and GD are plotted. Notably, GD exhibits slower convergence compared to RCD in achieving an energy ratio of 1. The employment of 90 parameters in the optimization, a number markedly greater than those in prior applications, might account for this disparity. This increased parameter count likely requires additional iterations and partial derivative evaluations when applying GD. Similarly to previous results, the two types of  Lipschitz constant ratios are obtained and shown along with the iterations in \cref{fig:tsp}. Again, we can see that the values of the ratios are considerably large, especially during the initial stage of the optimization, underlining the efficiency of RCD in the optimization process.

\begin{figure}[htbp]
   \centering
   \begin{subfigure}[b]{0.32\textwidth}
    \centering
    \includegraphics[width=\textwidth]{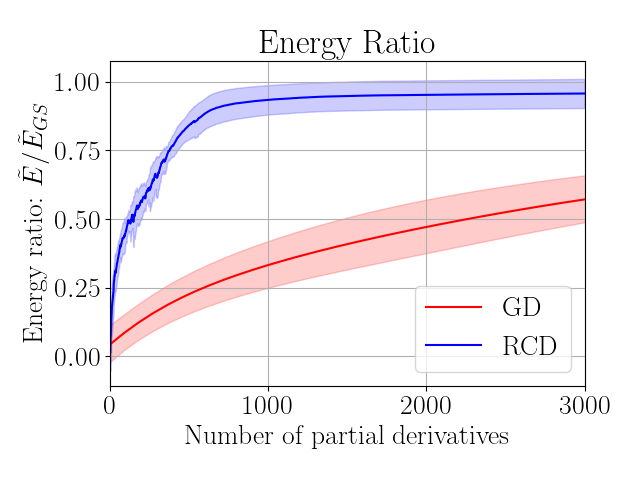}
    \end{subfigure}
   \begin{subfigure}[b]{0.32\textwidth}
    \centering
    \includegraphics[width=\textwidth]{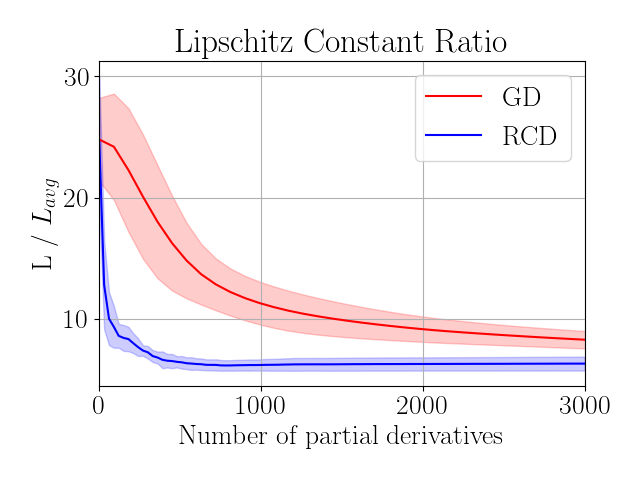}
    \end{subfigure}
   \begin{subfigure}[b]{0.32\textwidth}
    \centering
    \includegraphics[width=\textwidth]{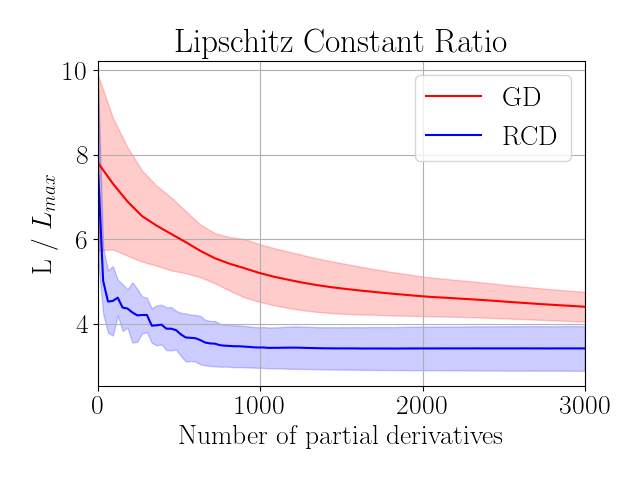}
    \end{subfigure}
    \caption{ \label{fig:tsp} Performance comparison between noisy GD and RCD for the TSP problem. The corresponding Lipschitz constant ratios, denoted as $\frac{L}{L_{\text{avg}}}$ and $\frac{L}{L_{\max}}$, are presented in the middle and right panels. The shaded areas within the figure represent variations that have been observed across 10 random realizations. The optimization process has been performed in 90 dimensions. In the first panel, $\tilde E / \tilde E_{GS}$ is defined as $(E - c) / (E_{GS} - c)$, where $c/E_{GS} = 3000$. For clarity in the presentation, we adjust the energy by a constant. }
\end{figure}

\subsubsection{Variational Quantum Factoring}\label{sec:factor}

Our next QUBO problem is designed as a variational quantum factoring task. For this task, we formulated the optimization problem within the framework of quantum adiabatic computation~\cite{dattani2014quantum, anschuetz2019variational}. 
For example, to factorize $143$ into the product of two prime numbers, let $143 = pq$, where 
\begin{align*}
p &= 8 + 4p_2 + 2p_1 + 1, \\
q &= 8 + 4q_2 + 2q_1 + 1.
\end{align*}
On direct computation, the relations are simplified to 
\begin{align}
p_1 + q_1 - 1 &= 0, \\
p_2 + q_2 - 1 &= 0, \\
p_2q_1 + p_1q_2 - 1&=0.
\end{align}
To solve this system of equations, we introduce a cost function
\begin{equation}
c(p_1, q_1, p_2, q_2) = (p_1+q_1-1)^2+(p_2+q_2-1)^2+(p_2 q_1+p_1 q_2-1)^2.
\label{eqn:ss_transform}
\end{equation}
By borrowing techniques (see Appendix.~\ref{sec:tech_qfac} for more details) from~\cite{xu2012quantum, schaller2007role}, the cost function can be reduced to 
\begin{equation}
c(p_1, q_1, p_2, q_2) = 5-3 p_1-p_2-q_1+2 p_1 q_1-3 p_2 q_1+2 p_1 p_2 q_1-3 q_2+p_1 q_2+2 p_2 q_2+2 p_2 q_1 q_2.
\end{equation}
Following the methods detailed in the QUBO, we treat $(p_1, q_1, p_2, q_2)$ as Boolean functions and substitute each Boolean with $\frac12 (1 - Z_i)$ as we did in previous sections. Then, the problem can be reformulated into the Ising Hamiltonian,
\begin{equation}
\begin{aligned}
     H =& -3 \mathbb{I} + \frac{1}{2} Z_0 + \frac{1}{4} Z_1 + \frac{3}{4} Z_0Z_2 + \frac{1}{4} Z_2 - \frac{1}{4} Z_1Z_2 + \frac{1}{4} Z_0Z_1 \\ 
      &- \frac{1}{4} Z_0Z_1Z_2 + \frac{1}{2} Z_3 + \frac{1}{4} Z_0Z_3 + \frac{3}{4} Z_1Z_3 + \frac{1}{4} Z_2Z_3 - \frac{1}{4} Z_1Z_2Z_3.
\end{aligned}
    \label{eqn:factor}
\end{equation}
The ground states of this Hamiltonian are $|0110\rangle$ and $|1001\rangle$, which respectively correspond to the solutions for the factorization of the number $143$. We summarize it as follows, 
\begin{eqnarray}
    & (p_{1}, p_{2}, q_{1}, q_{2})=(0,1,1,0) \longleftrightarrow (p,q) = (13,11)\\
    & (p_{1}, p_{2}, q_{1}, q_{2})=(1,0,0,1)  \longleftrightarrow (p,q) = (11,13)\\
    & p = 8 + 4p_2 + 2p_1 + 1 \text{ and }q = 8 + 4q_2 + 2q_1 + 1\quad \text{Boolean functions}. 
\end{eqnarray}

In our numerical experiment, we select the mixer Hamiltonian $H_2 = \sum X_i$ and set up a 20-layer QAOA, which corresponds to 40 parameters~\footnote{The QAOA ansatz builds the variational circuit by alternating between the parametrized unitary evolution associated with the problem Hamiltonian $H$ and the mixer Hamiltonian $H_2$.}. We set the learning rates to $0.0001$ for GD and $0.005$ for RCD and choose the energy as a measure for optimization. Even with a small step size, the variance of GD is notably large. Employing a larger step size for GD further exacerbates the results.

In \cref{fig:factor}, the optimization results of the Hamiltonian \eqref{eqn:factor} are depicted, showing that the number of partial derivative evaluations for the RCD to reach an energy ratio of 1 is about 400 whereas the GD seems to require more than 1000 to the same tolerance. As discussed previously, this observation aligns with prior discussions, particularly given the pronounced magnitude of the Lipschitz constant ratios evident in  \cref{fig:factor}.

\begin{figure}[htbp]
   \centering
   \begin{subfigure}[b]{0.32\textwidth}
    \centering
    \includegraphics[width=\textwidth]{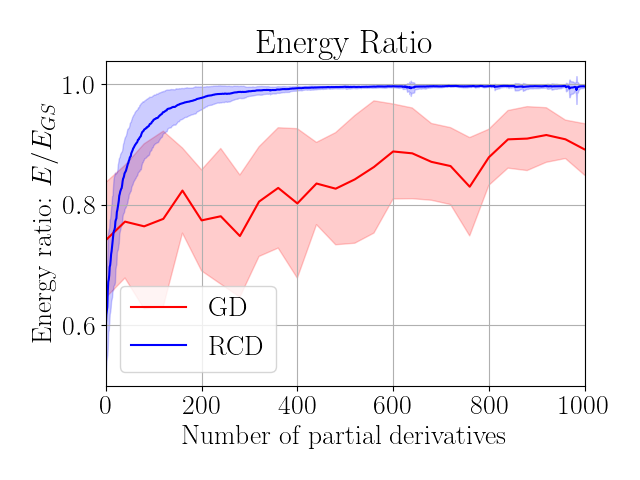}
    \end{subfigure}
   \begin{subfigure}[b]{0.32\textwidth}
    \centering
    \includegraphics[width=\textwidth]{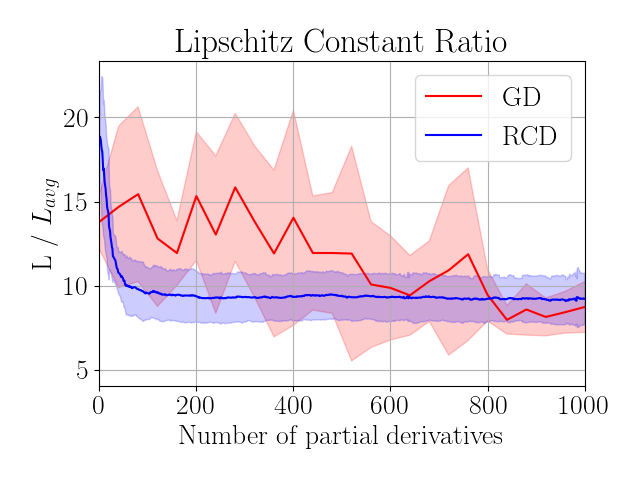}
    \end{subfigure}
   \begin{subfigure}[b]{0.32\textwidth}
    \centering
    \includegraphics[width=\textwidth]{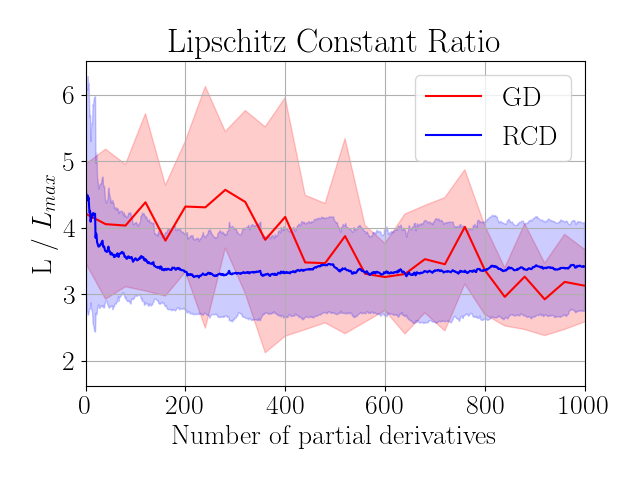}
    \end{subfigure}
    \caption{ \label{fig:factor}  Performance comparison between noisy GD and RCD for the quantum factoring problem. The corresponding Lipschitz constant ratios, denoted as $\frac{L}{L_{\text{avg}}}$ and $\frac{L}{L_{\max}}$, are presented in the middle and right panels. The shaded areas within the figure represent variations that have been observed across 10 random realizations. The optimization process has been performed in 40 dimensions.  }

\end{figure}

\section{Conclusion}

In this paper, we propose that the RCD method is a simple and effective algorithm for optimizing parameterized quantum circuits. The effectiveness of RCD is demonstrated using extensive numerical experiments  for a variety of quantum optimization problems (all are nonconvex problems). In a noisy environment, we find that RCD's computational cost (measured by the number of partial derivative calculations) can be significantly lower than that of the gradient descent method. This suggests that in optimizing parameterized quantum circuits, the benefits of developing an asymptotically efficient and yet complex quantum circuit for evaluating all partial derivatives, as discussed in~\cite{jordan2005fast,gilyen2019optimizing,abbas2023quantum}, may be limited and demand re-evaluation.

From a theoretical perspective, most previous works on randomized coordinate descent algorithms studied the case of convex cost functions or functions satisfying the global PL condition, which do not fit into most variational quantum applications that involve nonconvex cost functions. In this work, we generalized the conventional convergence analysis of randomized coordinate descent to local convergence analysis under a local-PL condition that can capture a large class of nonconvex optimization. In particular, we proved that noisy randomized coordinate descent can converge faster than noisy gradient descent in terms of the total cost, measured in terms of the total number of partial derivative estimations. Relaxing the global PL condition to a local one leads to additional theoretical difficulties, namely, the inherent noise can potentially displace parameters from the local basin. In our study, we introduce, to the best of our knowledge, a novel supermartingale approach to show that, with a proper initial condition, the parameters can stay within the basin and converge to a low-loss point. We believe that this analysis technique may also benefit the analysis of other quantum optimization algorithms. We also emphasize that our analysis does not guarantee the global convergence of the algorithm and, specifically, does not address its performance in the presence of barren plateaus \cite{McClean2018}.

Our analysis noise comes from quantum measurements. However, in real quantum devices, the presence of generic quantum noise may introduce additional noise to the problem. For example, circuit noise will have two effects on our optimization methods. First, it will increase the variance $\sigma^2$ associated with the gradient estimators and, as a result, increase the sampling complexity~\cite{PhysRevApplied.15.034026} or the number of iteration steps (\cref{thm:GD_informal} and~\cref{thm:RCD_informal}). Second, noise may induce a bias that is proportional to the noise strength ($\lambda$)~\cite{PhysRevApplied.15.034026}. Although the bias can be improved by error mitigation techniques (see~\cite{endo2018practical} and~\cite{wang2024can}), e.g., the zero noise extrapolation method, the variance overhead will persist, and thus increase the optimization complexity, as indicated by \cref{thm:GD_informal,thm:RCD_informal}. Meanwhile incorporating small bias into the convergence analysis seems possible, as suggested by previous works (refer, for instance, to~\cite{ajalloeian2020convergence}).

From an optimization standpoint, variational quantum optimization as outlined in \cref{problem_1} also raises many interesting questions. For instance, in our experiment, we employ a fixed number of measurement shots for each step. However, in practical scenarios, implementing an adaptive approach to determine the number of shots could potentially reduce the overall shot count and increase the accuracy. When the estimation for $f(\theta)$ and $\nabla f(\theta)$ reveals significant magnitudes, increasing the number of shots can effectively reduce variance, ensuring high accuracy in the estimation process. Another interesting direction is to consider non-gradient based optimization algorithm. For example, can second-order, or zero th-order optimization methods (i.e., methods using only function evaluation) be more efficient compared to the current gradient-based algorithms? In a technical viewpoint, another question is whether the stability result \cref{lemma: stability of GD} can be generalized so that the event covers the case that the iteration diverges at some time instances, but it remains in the entire basin until then, $f^{-1}\left[0,\delta_f\right)$, not necessarily in the region above the set of global minima, $f^{-1}\left(\frac{La\sigma_{\infty}^2d}{\mu},\delta_f\right)$. If this is possible to show, then it will provide a stronger result such as the stability of the noisy GD and RCD within the entire basin as the stability of Markov Chain in \cite{kushner1997applications}. 
Furthermore, the barren plateau~\cite{McClean2018} (i.e., vanishing gradients) poses a significant challenge for training algorithms for PQCs. Moreover, there remains a gap in the theoretical understanding of how the observed vanishing gradients relate to the assumptions underlying the analysis of optimization algorithms. A simple speculation is that the presence of barren plateau will make some of the constants (e.g., the local PL constant) in our analysis vanishingly small as the number of qubits and/or the circuit depth increases.  It may be possible to combine  RCD with some existing novel optimization algorithms, e.g.,  layerwise training \cite{skolik2021layerwise}, to partially mitigate the difficulty caused by the barren plateau in specific applications. % overall, RCD is not designed as a direct solution to the QBP problem. 
%We anticipate that Rcd alone may not sufficiently overcome the obstacles, given its reliance solely on local derivative information.

\section*{Acknowledgements.   } TK and XL's research is supported by the National Science Foundation Grants No. DMS-2111221 and No. CCF-2312456. TK is supported by a KIAS Individual Grant CG096001 at Korea Institute for Advanced Study. This material is based on work supported by the U.S. Department of Energy, Office of Science, National Quantum Information Science Research Centers, Quantum Systems Accelerator (ZD).  Additional funding is provided by the SciAI Center funded by the Office of Naval Research (ONR), under Grant Number N00014-23-1-2729 (LL). LL is a Simons investigator in Mathematics.

\appendix

\section{Stochastic stability of noisy GD}\label{sec: appendix B}

In this section, we prove Lemma \ref{lemma: stability of GD}.

\begin{proof}[Proof of Lemma \ref{lemma: stability of GD}]
Define the probability filtration: $\mathcal{F}_n=\sigma\left(\bs{\theta}_k\middle|k\leq n\right)$ and the stopping time\footnote{It is straightforward to see 
\[
\{\tau\leq n\}\in \mathcal{F}_n,\quad \{\tau>n\}\in \mathcal{F}_n\,.
\]}
\[
\tau=\inf \left\{k:f(\bs\theta_k)\leq \frac{La\sigma^2_\infty d}{\mu}\right\}\,,
\]
which is the smallest timepoint where the noisy GD achieves $f(\bs\theta_k)\leq\frac{La\sigma^2_\infty d}{\mu}$. 

Define the indicator function $\mathbb{I}_n$:
\begin{equation}\label{eq: characteristic ft2}
    \mathbb{I}_n=\left\{
\begin{aligned}
&1,\quad \text{if}\quad \{\bs{\theta}_k\}^{n-1}_{k=1}\subset f^{-1}\left(\left[0,\delta_f\right)\right)\\
&0,\quad \text{otherwise}
\end{aligned}\right.\,,    
\end{equation}
and the stochastic process
\[
V_n=\left\{
\begin{aligned}
&f(\bs\theta_n)\mathbb{I}_n,\quad n<\tau\\
&f(\bs\theta_\tau)\mathbb{I}_\tau,\quad n\geq\tau
\end{aligned}\right.\,.
\]
 According to the definition of $V_n$, there are complementary and exclusive events (cases):
\begin{itemize}
    \item Case 1: If there exists $0<n<\infty$ such that: 1. $\bs\theta_n\notin\mathcal{N}$; 2. For any $m<n$, $\bs\theta_m\in\mathcal{N}$ and $f(\bs\theta_m)>\frac{La\sigma^2_\infty d}{\mu}$. Then
    \[
    V_n\geq\delta_f\Rightarrow \sup_n V_n\geq\delta_f\,.
    \]
    \item Case 2: For any $n<\tau$, $f(\bs\theta_n)\in\mathcal{N}$.

\end{itemize}
We observe that Case 2 is the stable situation, indicating that $f(\bs\theta_n)$ remains in the basin of the global minimum until it achieves a small loss\footnote{We emphasize that Case 2 also includes the situation where $f(\bs\theta_n)$ remains in the basin and never achieves the small loss.}. 
To prove \eqref{eqn:p_gd}, it suffices to show that
\begin{equation}\label{eqn:P_omega_1_bound}
\mathbb{P}(\Omega_1)\leq \frac{f(\bs{\theta}_{1})}{\delta_f},
\end{equation}where $\Omega_1$ denotes the event associated with Case 1.

Now, we show that $V_n$ is a supermartingale to bound $\sup_n V_n$. Taking the conditional expectation, we obtain
\[
\mathbb{E}(V_{n+1}|\mathcal{F}_{n})=\mathbb{E}(V_{n+1}|\mathcal{F}_{n},\mathbb{I}_n=1,\tau\leq n)\mathbb{P}(\tau\leq n)+\mathbb{E}(V_{n+1}|\mathcal{F}_{n},\mathbb{I}_n=1,\tau>n)\mathbb{P}(\tau>n)\,,
\]
where we use $\mathbb{I}_{n+1}\leq \mathbb{I}_n$. There are two terms in the above equation:
\begin{itemize}
    \item For the first term, when $\tau\leq n$, we obtain $V_{n+1}=V_\tau=V_n$. This implies 
    \begin{equation}\label{eqn:first_term_estimation}
    \mathbb{E}(V_{n+1}|\mathcal{F}_{n},\mathbb{I}_n=1,\tau\leq n)=V_n\,.
    \end{equation}
    \item For the second term, when $\tau>n$, we have $f(\bs\theta_{n})> \frac{La\sigma^2_\infty d}{\mu}$. Then, taking the conditional expectation  yields
\begin{equation}\label{eqn:second_term_estimation}
\begin{split}
    &\mathbb{E}[V_{n+1}|\mathbb{I}_n=1,\bs{\theta}_1,\tau>n]\\
    =&\mathbb{E}[f(\bs{\theta}_{n+1})\mathbb{I}_{n+1}|\mathbb{I}_n=1,\bs{\theta}_1,\tau>n]\\
    \leq &f(\bs{\theta}_n)-a\|\nabla f(\bs{\theta}_n)\|^2+\frac{La^2}{2}(\|\nabla f(\bs{\theta}_n)\|^2+\sigma^2_\infty d)\\
    \leq &(1-\mu a)f(\bs{\theta}_n)+\frac{La^2\sigma^2_\infty d}{2}\\
    < &(1-\mu a)f(\bs{\theta}_n) +\frac{\mu a}{2}f(\bs{\theta}_n)\\ 
    \leq &\left(1-\frac{\mu a}{2}\right)f(\bs{\theta}_n)\mathbb{I}_n=\left(1-\frac{\mu a}{2}\right)V_n,
\end{split}
\end{equation}
where we use assumption \ref{as: local PL} and $a<\frac{1}{L}$ in the second inequality, $\tau>n$ in the third inequality.
\end{itemize}
Combining \eqref{eqn:first_term_estimation} and \eqref{eqn:second_term_estimation}, we obtain
\begin{equation}\label{eqn:iteration}
\mathbb{E}(V_{n+1}|\mathcal{F}_{n})=V_n\mathbb{P}(\tau\leq n)+\left(1-\frac{\mu a}{2}\right)V_n\mathbb{P}(\tau>n)\leq V_n\,.
\end{equation}
Thus, $V_{n}$ is a supermartingale. 

Now, we consider the Case 1 event:
\[
\Omega_1=\left\{\exists n>1, \bs\theta_n\notin\mathcal{N}\ \text{and}\ f(\bs\theta_m)>\epsilon\text{ with }\bs\theta_m\in\mathcal{N},\ \forall 1\leq m<n\right\}\subset\left\{\sup_n V_n\geq \delta_f\right\}\,.
\]
Because $V_{n}$ is a supermartingale, we obtain Case 1 happens with small probability:
\[
\mathbb{P}(\Omega_1)\leq \frac{V_1}{\delta_f}=\frac{f(\bs\theta_1)}{\delta_f}\,.
\]
This concludes the proof.
\end{proof}

\section{Stochastic stability of noisy RCD}\label{sec: appendix C}

In this section, we prove \cref{lemma: stability of RCD} with a slight modification of the proof in \cref{sec: appendix B}. From a theoretical viewpoint, the difference between the noisy GD and RCD methods is made by the construction of gradient estimate (e.g., see \eqref{eq:def_GD} and \eqref{eq: def_RCD}). Compared to GD,  the additional randomness of RCD comes with the random selection of a component as in \eqref{eq: def_RCD}. This difference affects the recursive inequality \eqref{eqn:second_term_estimation} mainly in the previous proof, where we considered the properties of the gradient estimator. From this observation, it suffices to derive a recursive inequality similar to \eqref{eqn:second_term_estimation} to prove \cref{lemma: stability of RCD}.

Note that the sampling of a component within RCD is performed before estimating a partial derivative. Thus, the first step is to take expectation on the partial derivative estimate,
\begin{equation}\label{eq: expectation on xi_RCD}
\mathbb{E}_{\xi_{i_n}}[f(\bs{\theta}_{n+1})]\leq f(\bs{\theta}_n)-a\mathbb{E}_{\xi_{i_n}}[\partial_{i_n}f(\bs{\theta}_n)g_{i_n}(\bs{\theta}_n)]+\frac{L_{i_n}a^2}{2}\mathbb{E}_{\xi_{i_n}}[|g_{i_n}(\bs{\theta}_n)|^2]\,,
\end{equation}where $i_n$ is uniformly sampled index and $g_{i_n}$ is the corresponding unbiased estimate for the partial derivative. Let $\mathcal{F}_n$, $\tau$, $\mathbb{I}_n$, and $V_n$ be as defined in the previous section. By considering the inequality \eqref{eq: expectation on xi_RCD} and the conditional expectation in \eqref{eqn:second_term_estimation}, we achieve the following result by taking expectations with respect to the random index $i_n$,
\begin{equation}\label{eqn:second_term_estimation_RCD}
\begin{split}
    &\mathbb{E}[V_{n+1}|\mathbb{I}_n=1,\bs{\theta}_1,\tau>n]\\
    =&\mathbb{E}[f(\bs{\theta}_{n+1})\mathbb{I}_{n+1}|\mathbb{I}_n=1,\bs{\theta}_1,\tau>n]\\
    \leq &\left(f(\bs{\theta}_n)-\frac{a}{d}\|\nabla f(\bs{\theta}_n)\|^2+\frac{L_{\max}a^2}{2d}\|\nabla f(\bs{\theta}_n)\|^2+\frac{L_{\mathrm{avg}}\sigma^2_\infty da^2}{2d}\right)\mathbb{I}_{n+1}\\
    \leq &\left(\left(1-\frac{\mu a}{d}\right)f(\bs{\theta}_n)+\frac{L_{\mathrm{avg}}a^2\sigma^2_\infty }{2}\right)\mathbb{I}_{n+1}\\
    < &\left(\left(1-\frac{\mu a}{d}\right)f(\bs{\theta}_n) +\frac{\mu a}{2d}f(\bs{\theta}_n)\right)\mathbb{I}_{n+1}\\ 
    = &\left(1-\frac{\mu a}{2d}\right)f(\bs{\theta}_n)\mathbb{I}_n=\left(1-\frac{\mu a}{2d}\right)V_n,
\end{split}
\end{equation}provided that $f(\bs\theta_n)>\frac{L_{\mathrm{avg}a\sigma_{\infty}^2d}}{\mu}$ and $a_n=a<\min\left\{\frac{1}{L_{\max}},\; \frac{d}{\mu}, \;\frac{2\mu \delta_f}{L_{\mathrm{avg}}\sigma^2_\infty d}\right\}$. 

Similarly to \eqref{eqn:iteration}, in the case of RCD, \eqref{eqn:second_term_estimation_RCD} implies
\begin{equation}\label{eqn:iteration_RCD}
\mathbb{E}(V_{n+1}|\mathcal{F}_{n})=V_n\mathbb{P}(\tau\leq n)+\left(1-\frac{\mu a}{2d}\right)V_n\mathbb{P}(\tau>n)\leq V_n\,,
\end{equation}
which implies $V_n$ forms a supermartingale. The remaining proof of \cref{lemma: stability of RCD} follows the same steps as the proof of \cref{lemma: stability of GD}, so we will not include them here.

\section{The proofs of 
\texorpdfstring{\cref{thm:GD_noisy} and \cref{thm:RCD_noisy}}{Lg}}\label{sec: appendix D}

We first show the convergence rate of the noisy GD method, followed by a similar analysis for the noisy RCD method. The following proofs are similar to those in \cref{sec: appendix B} and \cref{sec: appendix C} with minor differences.

\begin{proof}[Proof of Theorem \ref{thm:GD_noisy}]
Define the probability filtration: $\mathcal{F}_n=\sigma\left(\bs{\theta}_k\middle|k\leq n\right)$ and the stopping time
\[
\tau=\inf \left\{k:f(\bs\theta_k)\leq \epsilon\right\}\,,
\]
which is the smallest timepoint where the noisy GD achieves $f(\bs\theta_k)\leq\epsilon$. Then, our ultimate goal is to show that the event that $\inf_{1\leq n\leq N}f(\bs\theta_n)\leq \epsilon$ occurs with high probability, say, at least $1-\eta$.
Then, our goal is to show that for any $\eta\in\left(\frac{f(\bs\theta_1)}{\delta_f},1\right)$, there exists a sufficiently large $N$ such that
\begin{equation}\label{pfail}
    p_{\mathrm{fail}}:=\mathbb{P}(\tau>N)\leq \eta.
\end{equation}
Define the indicator function $\mathbb{I}_n$:
\[
    \mathbb{I}_n=\left\{
\begin{aligned}
&1,\quad \text{if}\quad \{\bs{\theta}_k\}^{n-1}_{k=1}\subset f^{-1}\left(\left[0,\delta_f\right)\right)\\
&0,\quad \text{otherwise}
\end{aligned}\right.\,,    
\]
and the stochastic process
\[
V_n=\left\{
\begin{aligned}
&f(\bs\theta_n)\mathbb{I}_n,\quad n<\tau\\
&f(\bs\theta_\tau)\mathbb{I}_\tau,\quad n\geq\tau
\end{aligned}\right.\,.
\]
Define the unstable event:
\[
\Omega=\left\{\exists n>1, \bs\theta_n\notin\mathcal{N}\ \text{and}\ f(\bs\theta_m)>\epsilon,\ \forall 1\leq m<n\right\}\subset\left\{\sup_n V_n\geq \delta_f\right\}\,.
\]
According to Lemma \ref{lemma: stability of GD} and the proof in \cref{sec: appendix B}, for learning rate $a$ with $\frac{La\sigma^2_\infty d}{\mu}<\epsilon$, we obtain $\Omega$ happens with small probability:
\begin{equation}\label{eqn:failure_1}
\mathbb{P}(\Omega)\leq \frac{V_1}{\delta_f}=\frac{f(\bs\theta_1)}{\delta_f}\,.
\end{equation}  
Recalling \eqref{pfail}, we note that, for any $n\leq N$, 
\[
\mathbb{P}(\tau>n)\geq p_{\mathrm{fail}}\,.
\]
Plugging this into \eqref{eqn:iteration}, we obtain that
\begin{equation}
\begin{split}
    \mathbb{E}(V_{n+1}|\mathcal{F}_n)&= \left(1-\mathbb{P}(\tau>n)+\left(1-\frac{\mu a}{2}\right)\mathbb{P}(\tau>n)\right)V_n\\
    & = \left(1-\frac{\mu a \mathbb{P}(\tau>n)}{2}\right)V_n\\
    &\leq \left(1-\frac{\mu a p_{\mathrm{fail}}}{2}\right)V_n.
\end{split}
\end{equation}By taking the total expectation on both sides and using a telescoping trick, we achieve that
\begin{equation}
    \mathbb{E}(V_{n+1})\leq \left(1-\frac{\mu a p_{\mathrm{fail}}}{2}\right)^n V_1=\left(1-\frac{\mu a p_{\mathrm{fail}}}{2}\right)^n f(\bs\theta_1).
\end{equation}
This means that if the probability of failure, 
 $p_{\mathrm{fail}}$, is large, the expectation of $V_{n+1}$ decreases quickly. By Markov's inequality, we have
\[
\mathbb{P}\left(V_N> \epsilon\right)\leq \frac{\left(1-\frac{\mu a p_{\mathrm{fail}}}{2}\right)^{N-1} f(\bs\theta_1)}{\epsilon}\,,
\]
equivalently,
\[
\mathbb{P}\left(V_N\leq \epsilon\right)\geq 1-\frac{\left(1-\frac{\mu a p_{\mathrm{fail}}}{2}\right)^{N-1} f(\bs\theta_1)}{\epsilon}\,.
\]
Now, if we consider the event $\{V_N\leq \epsilon\}$, then it is the union of the following two events (not necessarily exclusive and complementary), which are slightly different from the ones in \cref{sec: appendix B}:
\begin{itemize}
    \item $\Omega_1$: There exists $n\leq N$ such that $f(\bs\theta_n)\leq \epsilon$ and $\bs\theta_n\in\mathcal{N}$. This means \[\inf_{1\leq n\leq N}f(\bs\theta_n)\leq \epsilon\,.\]

    We want to show that $\Omega_1$ happens with high probability.
    
    \item $\Omega_2$: There exists $n< N$ such that $f(\bs \theta_n)>\delta_f$ and $f(\bs\theta_m)>\epsilon$ for any $m<n$. 
    
    We note that, when $\Omega_2$ happens, we have $V_{n+1}=0$ with $f(\bs\theta_n)>\delta_f$, which implies $\Omega_2\subset \Omega$. According to \eqref{eqn:failure_1}, we obtain
    \[
\mathbb{P}\left(\Omega_2\right)\leq \mathbb{P}\left(\Omega\right)\leq \frac{f(\bs\theta_1)}{\delta_f}\,.
    \]
\end{itemize}
Now, we give a lower bound for the event $\Omega_1$:
\begin{equation}\label{lowerbound}
    \mathbb{P}\left(\inf_{1\leq n\leq N}f(\bs\theta_n)\leq \epsilon\right)=\mathbb{P}\left(\Omega_1\right)\geq \mathbb{P}\left(V_N\leq \epsilon\right)-\mathbb{P}(\Omega_2)\geq 1-\frac{\left(1-\frac{\mu a p_{\mathrm{fail}}}{2}\right)^N f(\bs\theta_1)}{\epsilon}-\frac{f(\bs\theta_1)}{\delta_f}.
\end{equation}Notice
\[
\mathbb{P}\left(\inf_{1\leq n\leq N}f(\bs\theta_n)\leq \epsilon\right)\leq \mathbb{P}(\tau\leq N)=1-p_{\mathrm{fail}}\,.
\]
Combining the above two inequalities, we have
\begin{equation}\label{eqn:must_hold}
p_{\mathrm{fail}}\leq \frac{\left(1-\frac{\mu a p_{\mathrm{fail}}}{2}\right)^N f(\bs\theta_1)}{\epsilon}+\frac{f(\bs\theta_1)}{\delta_f}\,.
\end{equation}
Next, we show \eqref{pfail} using the proof by contradiction. Assume that the conclusion of the theorem is not true, meaning that for some $\eta\in\left(\frac{f(\bs\theta_1)}{\delta_f},1\right)$ and every $N$,
\[
p_{\mathrm{fail}}> \eta\,.
\]
When $p_{\mathrm{fail}}>\eta$ and $N=\frac{2}{\mu a\eta}\log\left(\frac{f(\bs{\theta_1})}{\left(\eta-\frac{f(\bs\theta_1)}{\delta_f}\right)\epsilon}\right)$, then 
\[
\begin{aligned}
    &\frac{\left(1-\frac{\mu a p_{\mathrm{fail}}}{2}\right)^N f(\bs\theta_1)}{\epsilon}+\frac{f(\bs\theta_1)}{\delta_f}\\
    <&\frac{\left(1-\frac{\mu a \eta}{2}\right)^N f(\bs\theta_1)}{\epsilon}+\frac{f(\bs\theta_1)}{\delta_f}\\
    \leq&\frac{\exp\left(-\frac{\mu a \eta N}{2}\right) f(\bs\theta_1)}{\epsilon}+\frac{f(\bs\theta_1)}{\delta_f}=\eta<p_{\mathrm{fail}}
\end{aligned}
\]
where we use $p_{\mathrm{fail}}>\eta$ in the first inequality and $(1-x)^N\leq \exp(-xN)$ in the second inequality. This contradicts to \eqref{eqn:must_hold}. Thus, \eqref{pfail} must be true and we conclude the proof.

\end{proof}

\begin{proof}[Proof of Theorem \ref{thm:RCD_noisy}]
 Denote the probability of failure
\[
    p_{\mathrm{fail}}=\mathbb{P}(\tau>N).
\]
Similarly to the calculation in the previous proof, from \eqref{eqn:iteration_RCD}, we have
\begin{equation}\label{lowerbound_RCD}
    \mathbb{P}\left(\inf_{1\leq n\leq N}f(\bs\theta_n)\leq \epsilon\right)\geq 1-\frac{\left(1-\frac{\mu a p_{\mathrm{fail}}}{2d}\right)^N f(\bs\theta_1)}{\epsilon}-\frac{f(\bs\theta_1)}{\delta_f}.
\end{equation}
With the same logic below \eqref{lowerbound}, we conclude the proof of \cref{thm:RCD_noisy}.
\end{proof}

\newpage

\section{Parameterized Circuit for the VQE}\label{sec:circuit}

The quantum circuit described below is utilized for the numerical result in \cref{fig:vqe}.

\begin{figure}[htbp]
    \centering
\Qcircuit @C=0.9em @R=.7em {
& \gate{R_Y(\frac{3\pi}{2})} & \multigate{1}{R_{ZZ}(\theta_1)} & \qw 
 & \gate{R_X(\theta_2)} & \qw & 
\cdots &   & \multigate{1}{R_{ZZ}(\theta_{39})} & \qw 
 & \gate{R_X(\theta_{40})} & \qw \\
& \gate{R_Y(\frac{3\pi}{2})} & \ghost{R_{ZZ}(\theta_1)} &  \multigate{1}{R_{ZZ}(\theta_1)} & \gate{R_X(\theta_2)} & \qw & \cdots &  & \ghost{R_{ZZ}(\theta_{39})} &  \multigate{1}{R_{ZZ}(\theta_{39})} & \gate{R_X(\theta_{40})} & \qw \\
& \gate{R_Y(\frac{3\pi}{2})} & \multigate{1}{R_{ZZ}(\theta_1)} & \ghost{R_{ZZ}(\theta_1)} & \gate{R_X(\theta_2)} & \qw & \cdots &  & \multigate{1}{R_{ZZ}(\theta_{39})} & \ghost{R_{ZZ}(\theta_{39})} & \gate{R_X(\theta_{40})} & \qw\\
& \gate{R_Y(\frac{3\pi}{2})} & \ghost{R_{ZZ}(\theta_1)} & \multigate{1}{R_{ZZ}(\theta_1)} & \gate{R_X(\theta_2)} & \qw & \cdots &  & \ghost{R_{ZZ}(\theta_{39})} & \multigate{1}{R_{ZZ}(\theta_{39})} & \gate{R_X(\theta_{40})} & \qw\\
& \gate{R_Y(\frac{3\pi}{2})} & \multigate{1}{R_{ZZ}(\theta_1)} & \ghost{R_{ZZ}(\theta_1)} & \gate{R_X(\theta_2)} & \qw & \cdots &  & \multigate{1}{R_{ZZ}(\theta_{39})} & \ghost{R_{ZZ}(\theta_{39})} & \gate{R_X(\theta_{40})} & \qw\\
& \gate{R_Y(\frac{3\pi}{2})} & \ghost{R_{ZZ}(\theta_1)} & \multigate{1}{R_{ZZ}(\theta_1)} & \gate{R_X(\theta_2)} & \qw & \cdots &  & \ghost{R_{ZZ}(\theta_{39})} & \multigate{1}{R_{ZZ}(\theta_{39})} & \gate{R_X(\theta_{40})} & \qw\\
& \gate{R_Y(\frac{3\pi}{2})} & \multigate{1}{R_{ZZ}(\theta_1)} & \ghost{R_{ZZ}(\theta_1)} & \gate{R_X(\theta_2)} & \qw & \cdots & & \multigate{1}{R_{ZZ}(\theta_{39})} & \ghost{R_{ZZ}(\theta_{39})} & \gate{R_X(\theta_{40})} & \qw\\
& \gate{R_Y(\frac{3\pi}{2})} & \ghost{R_{ZZ}(\theta_1)} & \multigate{1}{R_{ZZ}(\theta_1)} & \gate{R_X(\theta_2)} & \qw & \cdots &  & \ghost{R_{ZZ}(\theta_{39})} & \multigate{1}{R_{ZZ}(\theta_{39})} & \gate{R_X(\theta_{40})}
& \qw \\
& \gate{R_Y(\frac{3\pi}{2})} & \multigate{1}{R_{ZZ}(\theta_1)} & \ghost{R_{ZZ}(\theta_1)} & \gate{R_X(\theta_2)} & \qw & \cdots & & \multigate{1}{R_{ZZ}(\theta_{39})} & \ghost{R_{ZZ}(\theta_{39})} & \gate{R_X(\theta_{40})} & \qw \\
& \gate{R_Y(\frac{3\pi}{2})} & \ghost{R_{ZZ}(\theta_1)} & \qw & \gate{R_X(\theta_2)} & \qw & \cdots &  & \ghost{R_{ZZ}(\theta_{39})} & \qw & \gate{R_X(\theta_{40})}
& \qw } 
\caption{A QAOA-like ansatz motivated by \cite{wiersema2020exploring} is used for the TFIM model \eqref{eq:isingmodel} with $10$ qubits. 
    For the result in \cref{fig:vqe}, 36 parameters are assigned with 18 layers of alternating rotation ZZ gates and rotation X gates.}
    \label{fig:circuit}
\end{figure}
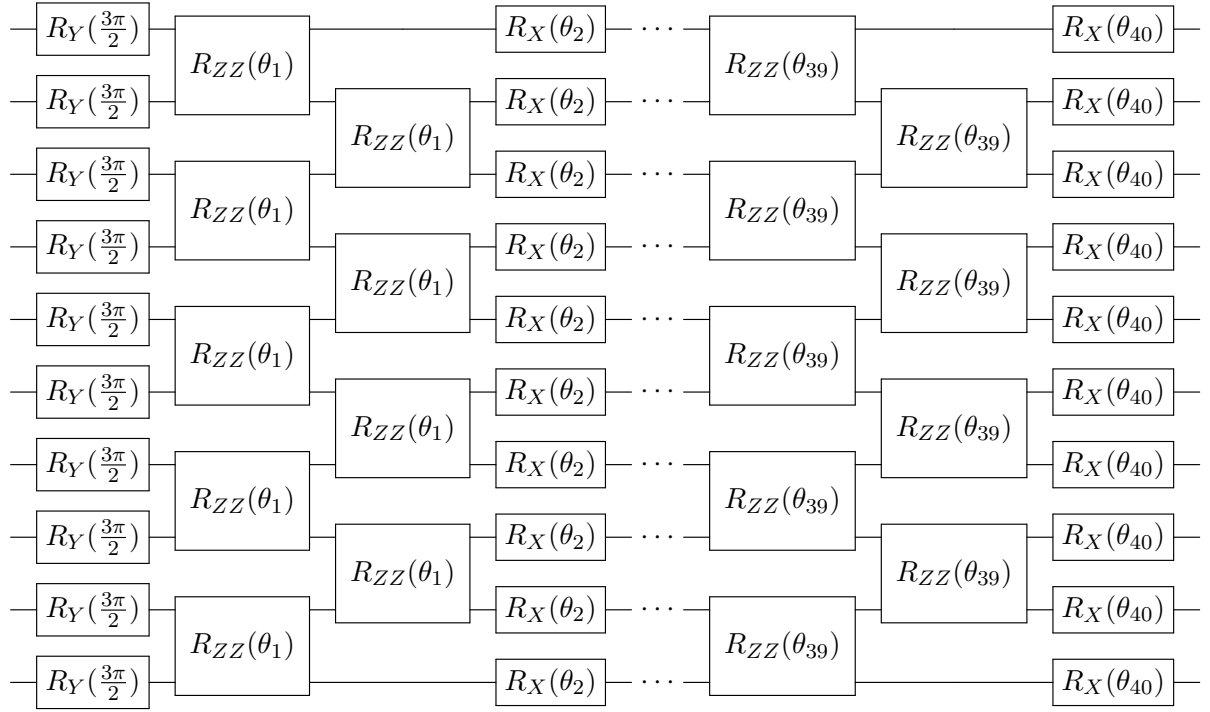

\newpage 

\section{Parameterized Circuit for the VQE in QUBO experiments}\label{sec:circuit_qubo}

The quantum circuit described below is utilized in the QUBO experiments.

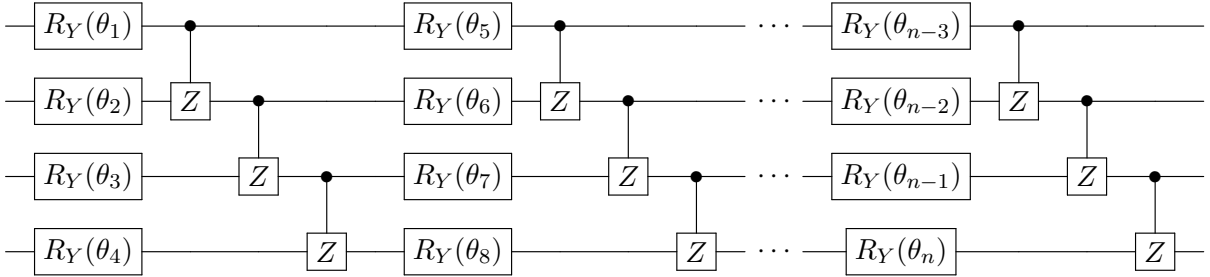
\begin{figure}[htbp]
    \centering
\Qcircuit @C=1em @R=1em {
& \gate{R_Y(\theta_1)} & \ctrl{1} & \qw & \qw & \qw  & \gate{R_Y(\theta_5)} & \ctrl{1} & \qw & \qw & \qw  & \cdots & & \gate{R_Y(\theta_{n-3})} & \ctrl{1} & \qw & \qw & \qw  \\
& \gate{R_Y(\theta_2)} & \gate{Z} & \ctrl{1} & \qw & \qw  & \gate{R_Y(\theta_6)} & \gate{Z} & \ctrl{1} & \qw & \qw  & \cdots & & \gate{R_Y(\theta_{n-2})} & \gate{Z} & \ctrl{1} & \qw & \qw  \\
& \gate{R_Y(\theta_3)} & \qw & \gate{Z} & \ctrl{1} & \qw  & \gate{R_Y(\theta_7)} & \qw & \gate{Z} & \ctrl{1} & \qw  & \cdots & & \gate{R_Y(\theta_{n-1})} & \qw & \gate{Z} & \ctrl{1} & \qw  \\
& \gate{R_Y(\theta_4)} & \qw & \qw & \gate{Z} & \qw  & \gate{R_Y(\theta_8)} & \qw & \qw & \gate{Z} & \qw  & \cdots & & \gate{R_Y(\theta_n)} & \qw & \qw & \gate{Z} & \qw  \\
}

\caption{A parametrized quantum circuit is employed in the QUBO experiments. This circuit features alternating layers of single rotation gates and entangling controlled-z gates. The adjustable parameters are exclusively found in the single rotation gates, and these parameters vary across different layers and qubits.}
    \label{fig:circuit_qubo}
\end{figure}

\section{Additional histograms of partial derivative estimates}\label{sec:noise}

\cref{fig:vqe_noise} plots the histograms with respect to the first 12 parameters among 36. The rest of 24 histograms are shown in the following figures. It is observed in all figures that the variances of partial derivative estimates in all directions are a similar magnitude of value. 

\begin{figure}[htbp]
   \centering
   \begin{subfigure}[b]{0.3\textwidth}
    \centering
    \includegraphics[width=\textwidth]{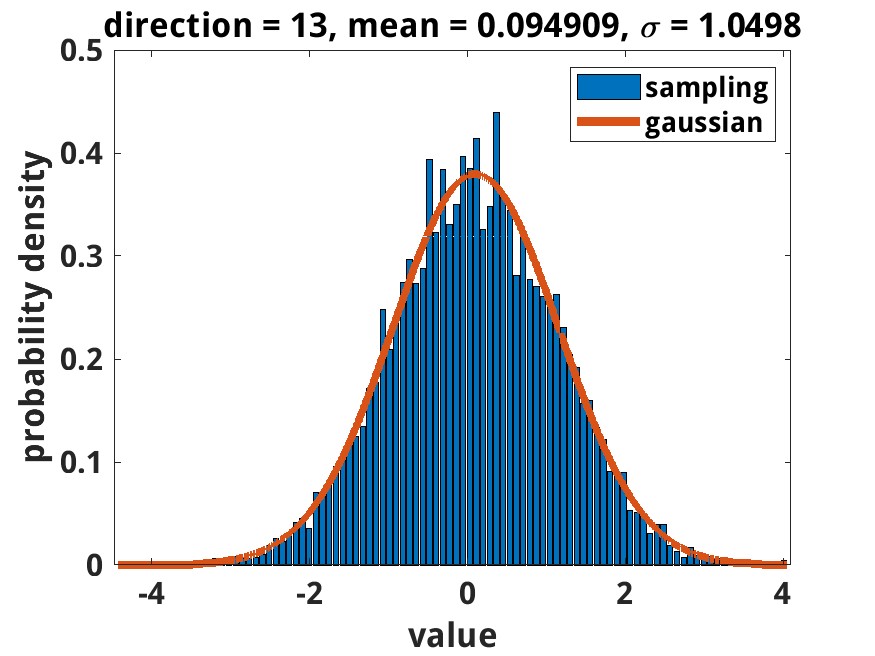}
    \end{subfigure}
   \begin{subfigure}[b]{0.3\textwidth}
    \centering
    \includegraphics[width=\textwidth]{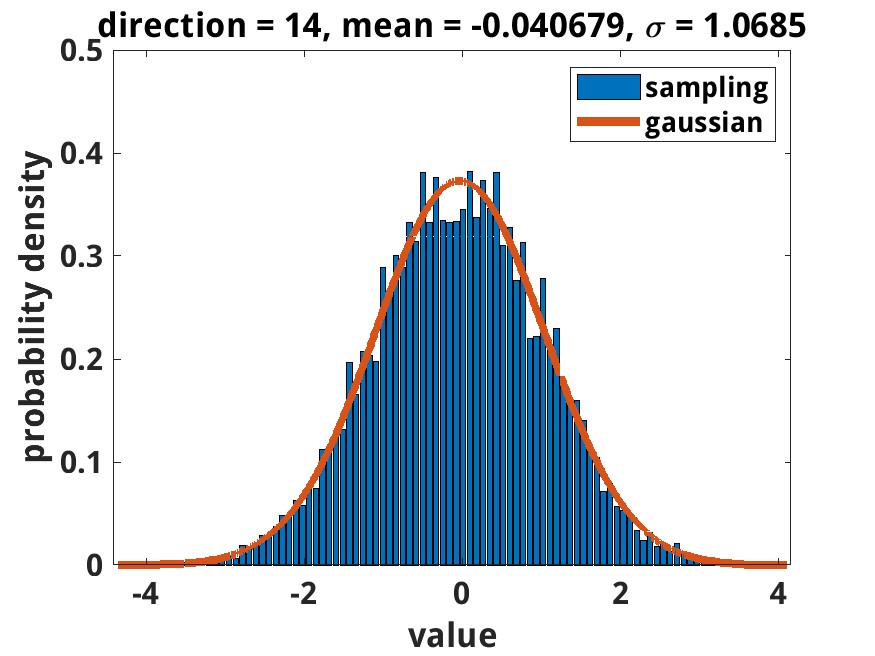}
    \end{subfigure}
    \centering
   \begin{subfigure}[b]{0.3\textwidth}
    \centering
    \includegraphics[width=\textwidth]{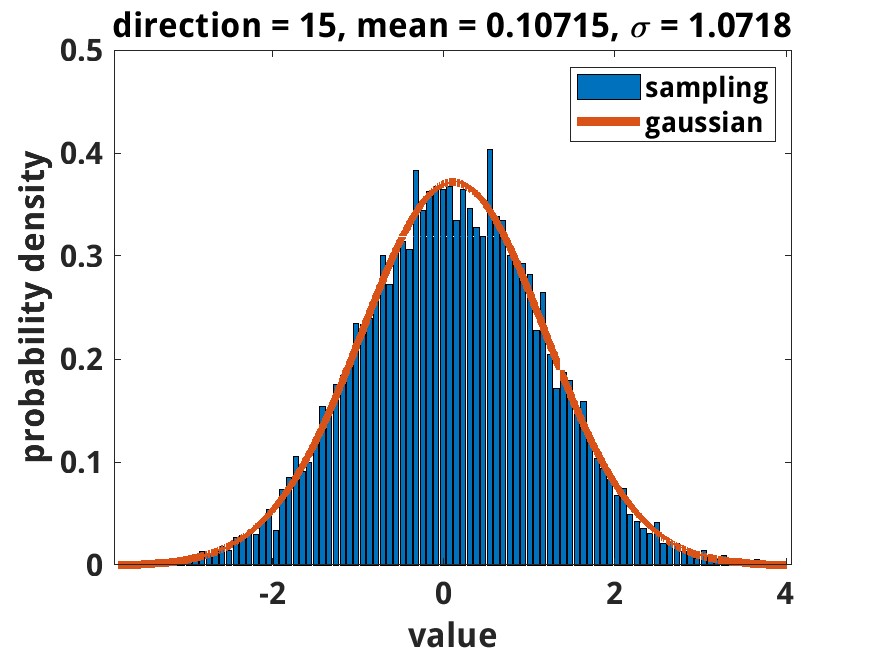}
    \end{subfigure}
    
    \begin{subfigure}[b]{0.3\textwidth}
    \centering
    \includegraphics[width=\textwidth]{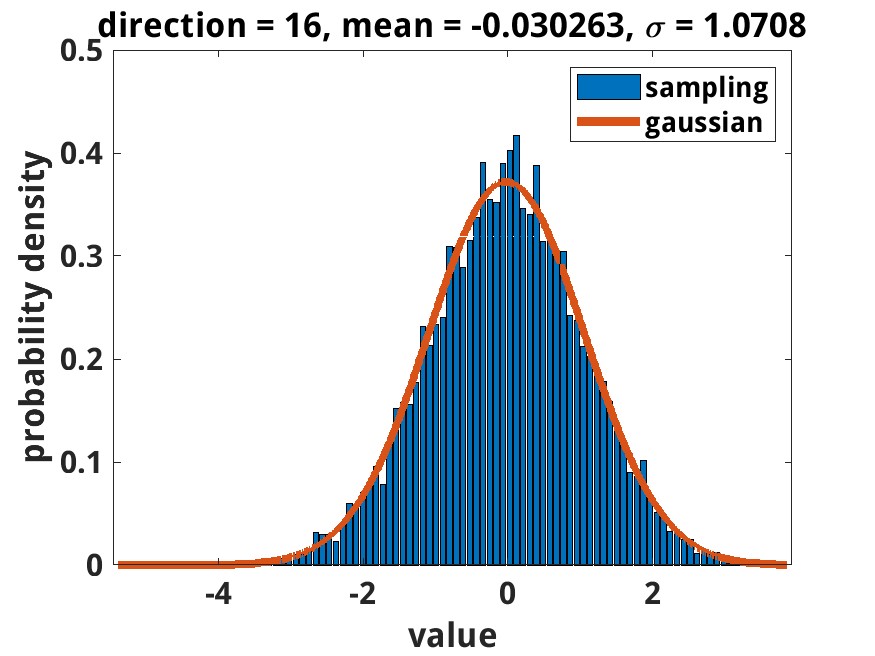}
    \end{subfigure}
   \begin{subfigure}[b]{0.3\textwidth}
    \centering
    \includegraphics[width=\textwidth]{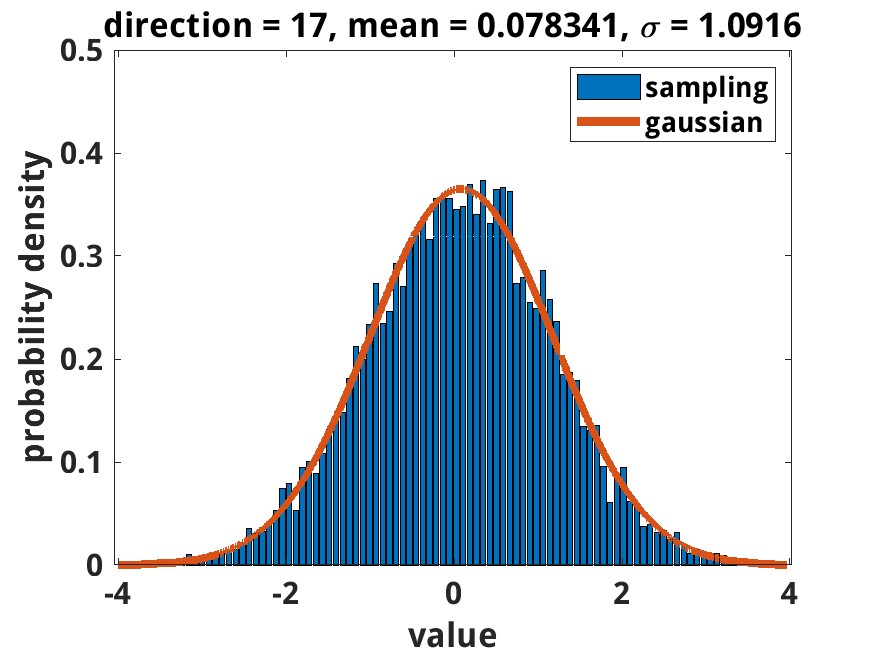}
    \end{subfigure}
    \centering
   \begin{subfigure}[b]{0.3\textwidth}
    \centering
    \includegraphics[width=\textwidth]{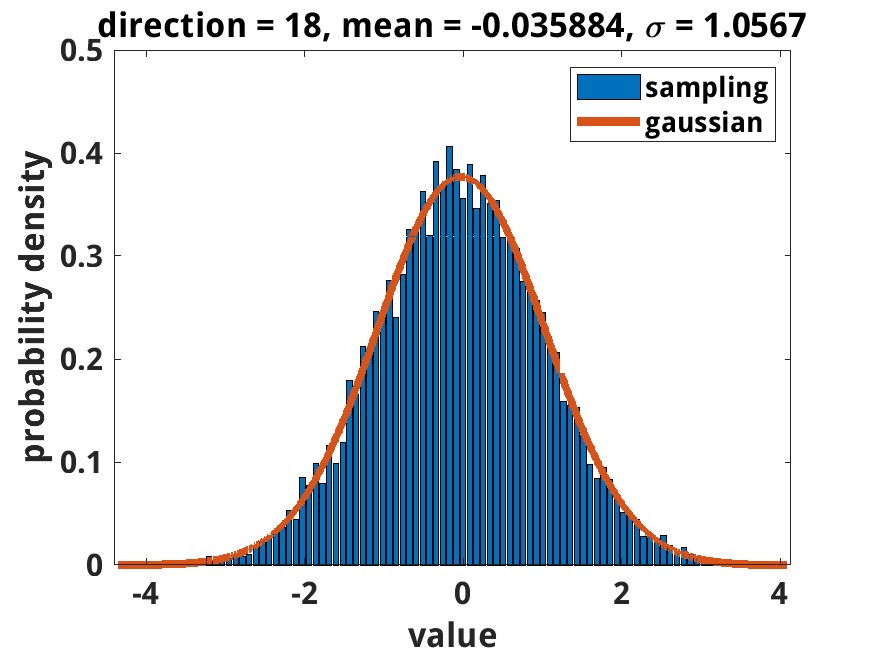}
    \end{subfigure}
    
    \begin{subfigure}[b]{0.3\textwidth}
    \centering
    \includegraphics[width=\textwidth]{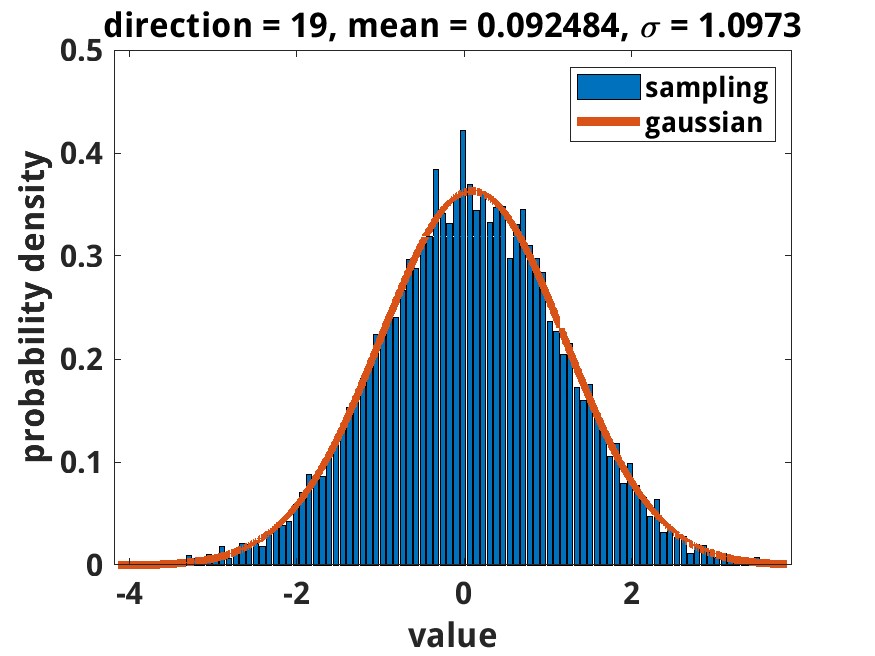}
    \end{subfigure}
   \begin{subfigure}[b]{0.3\textwidth}
    \centering
    \includegraphics[width=\textwidth]{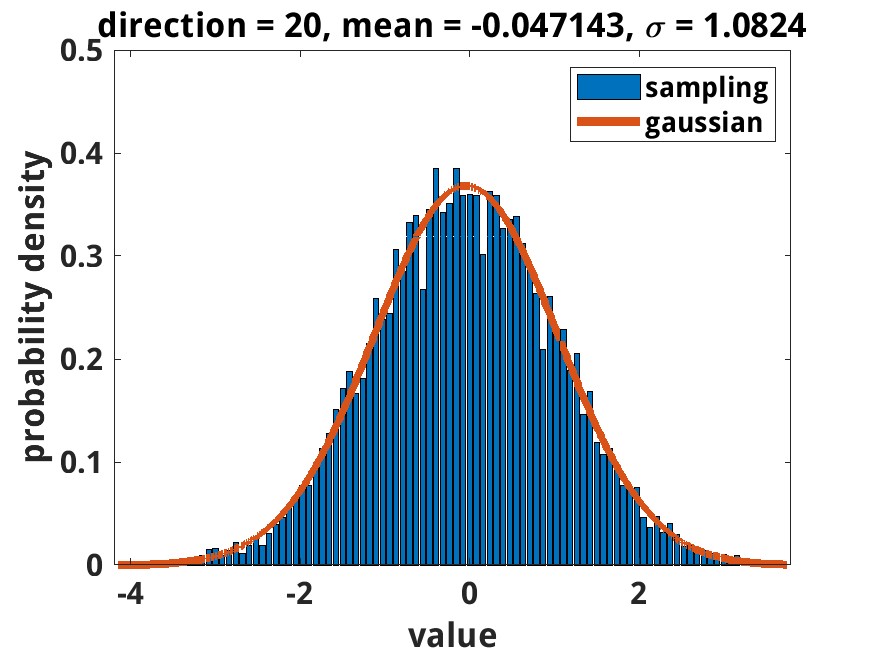}
    \end{subfigure}
    \centering
   \begin{subfigure}[b]{0.3\textwidth}
    \centering
    \includegraphics[width=\textwidth]{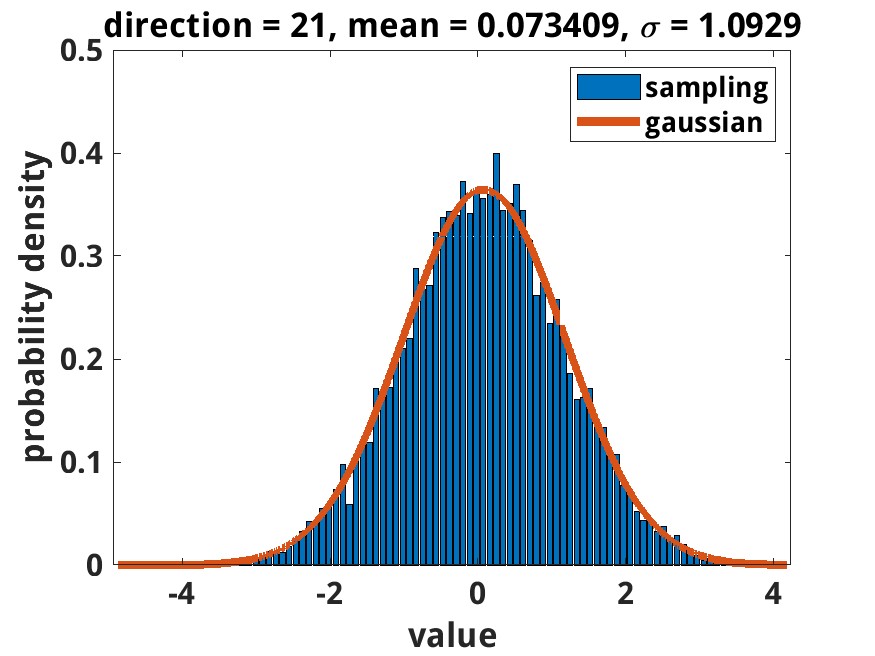}
    \end{subfigure}
    
    \begin{subfigure}[b]{0.3\textwidth}
    \centering
    \includegraphics[width=\textwidth]{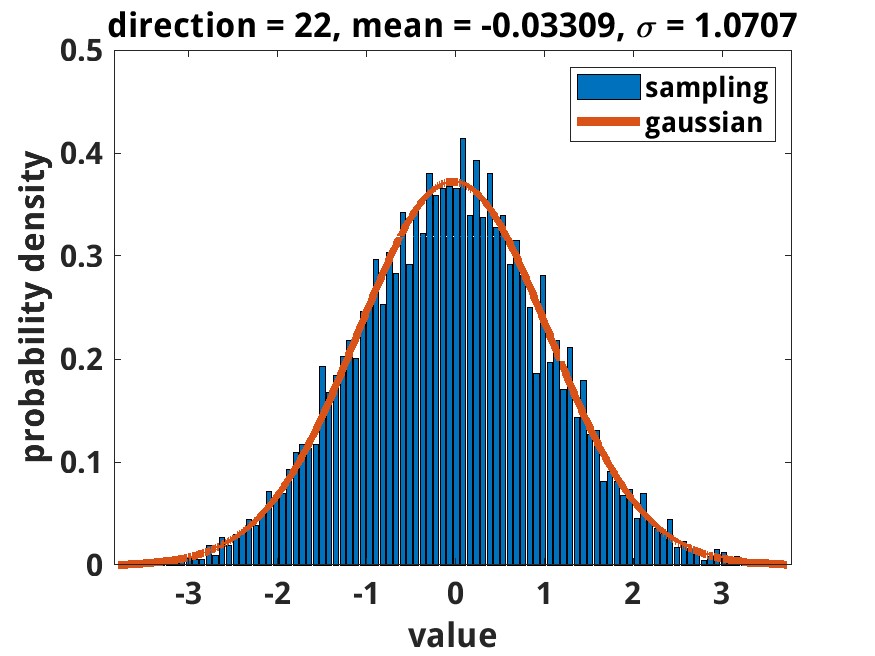}
    \end{subfigure}
   \begin{subfigure}[b]{0.3\textwidth}
    \centering
    \includegraphics[width=\textwidth]{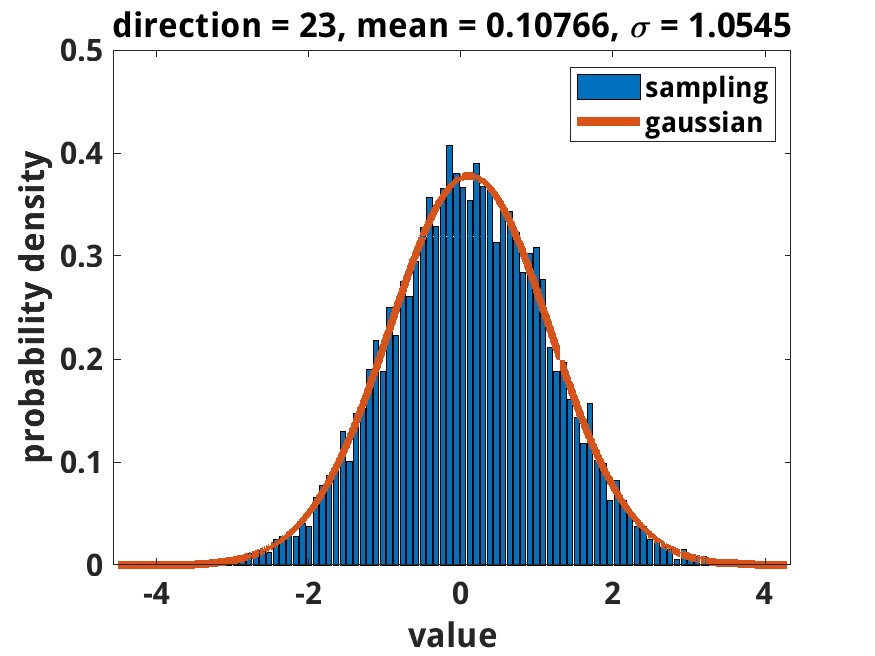}
    \end{subfigure}
    \centering
   \begin{subfigure}[b]{0.3\textwidth}
    \centering
    \includegraphics[width=\textwidth]{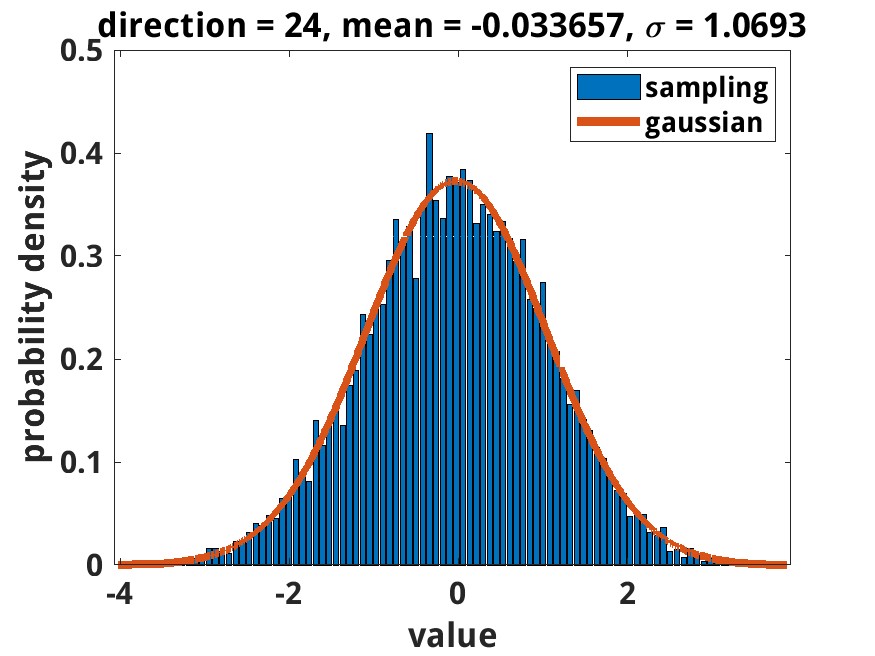}
    \end{subfigure}

    \begin{subfigure}[b]{0.3\textwidth}
    \centering
    \includegraphics[width=\textwidth]{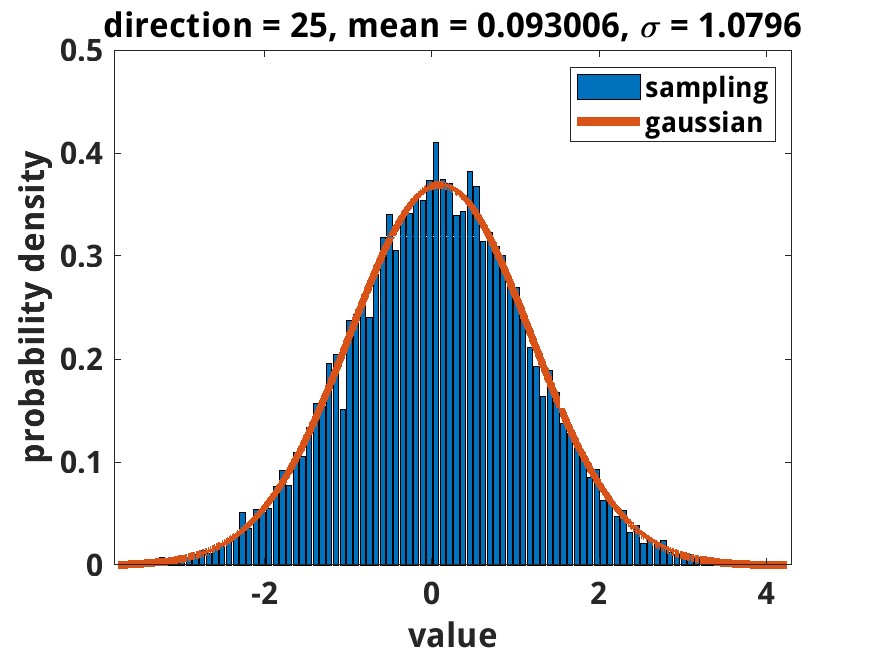}
    \end{subfigure}
    \centering
   \begin{subfigure}[b]{0.3\textwidth}
    \centering
    \includegraphics[width=\textwidth]{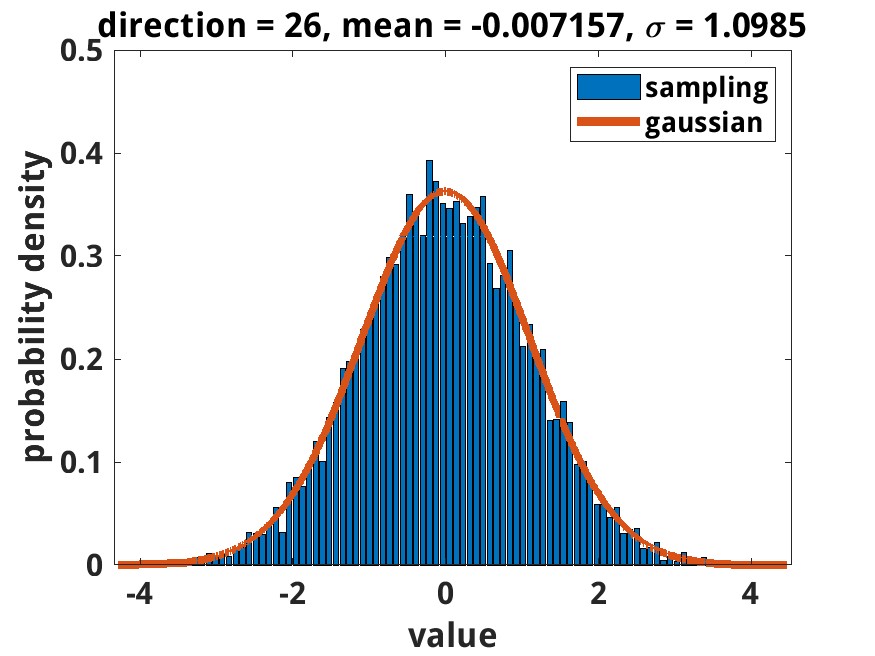}
    \end{subfigure}

    \caption{The histogram of partial derivative estimates with respect to the 13 th to the 26 th parameters are plotted with the same setup as in \cref{fig:vqe_noise}.}
    \label{fig:noise1}
\end{figure}

\begin{figure}[htbp]
   \centering
   \begin{subfigure}[b]{0.3\textwidth}
    \centering
    \includegraphics[width=\textwidth]{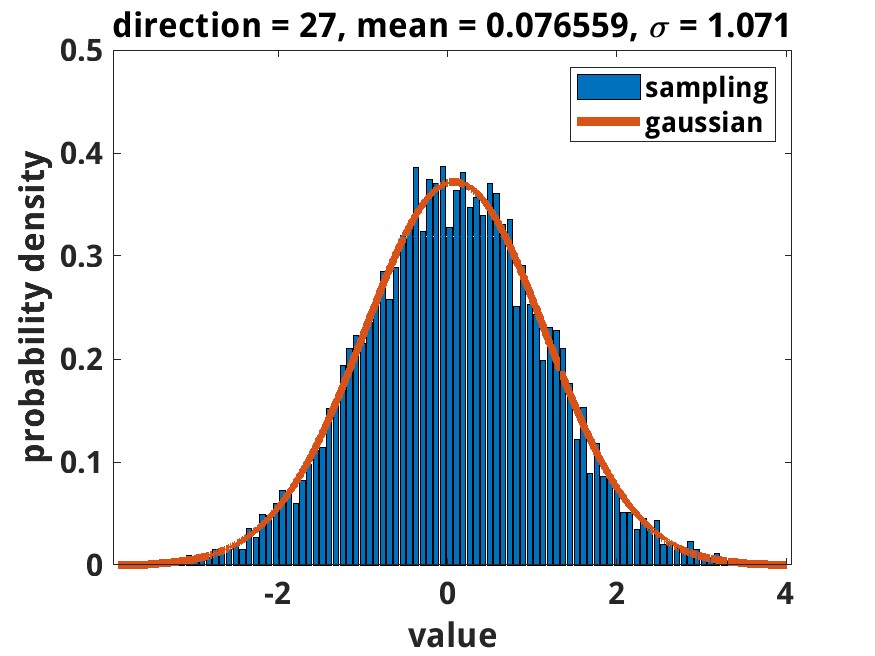}
    \end{subfigure}
   \begin{subfigure}[b]{0.3\textwidth}
    \centering
    \includegraphics[width=\textwidth]{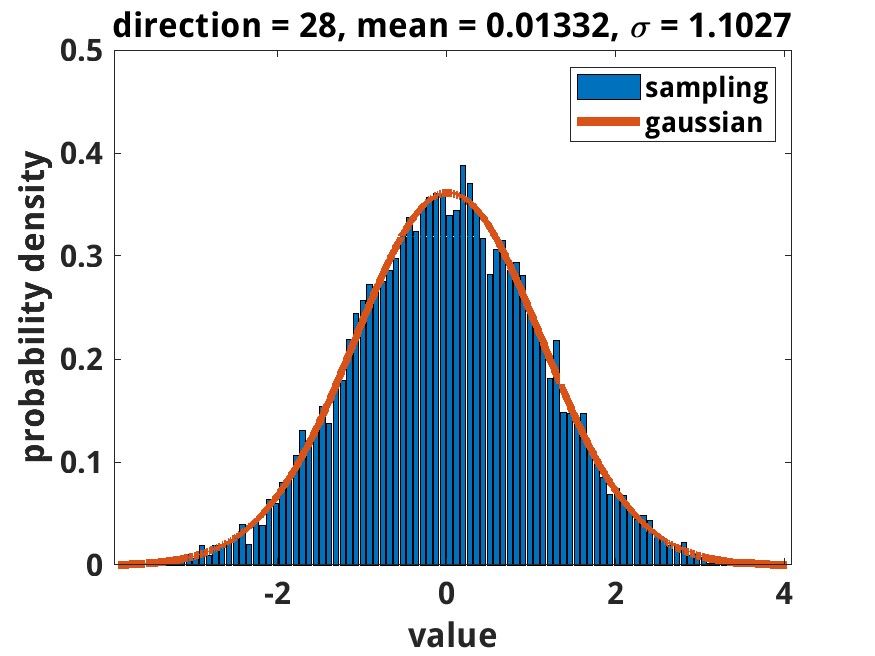}
    \end{subfigure}
    \centering
   \begin{subfigure}[b]{0.3\textwidth}
    \centering
    \includegraphics[width=\textwidth]{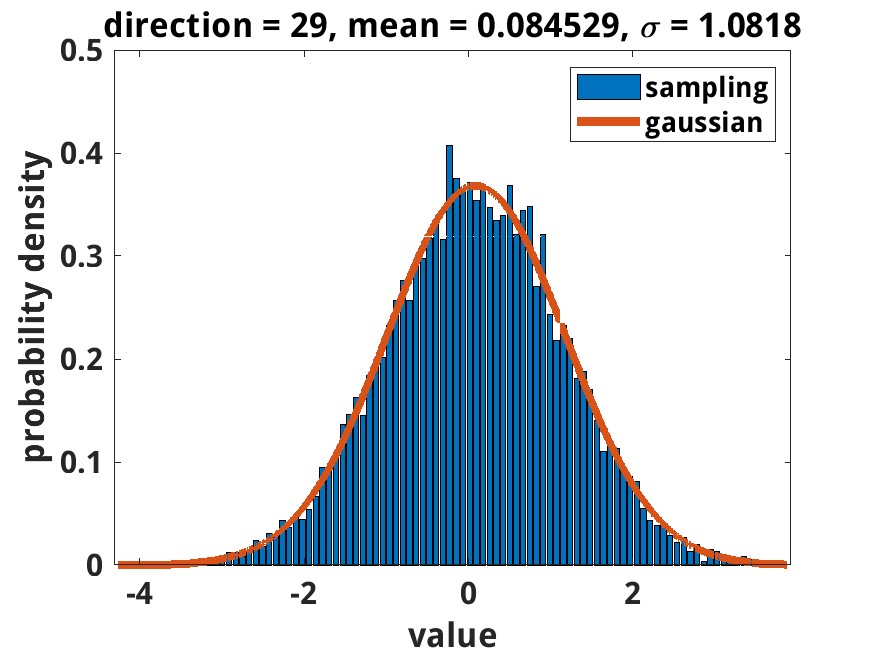}
    \end{subfigure}
    
    \begin{subfigure}[b]{0.3\textwidth}
    \centering
    \includegraphics[width=\textwidth]{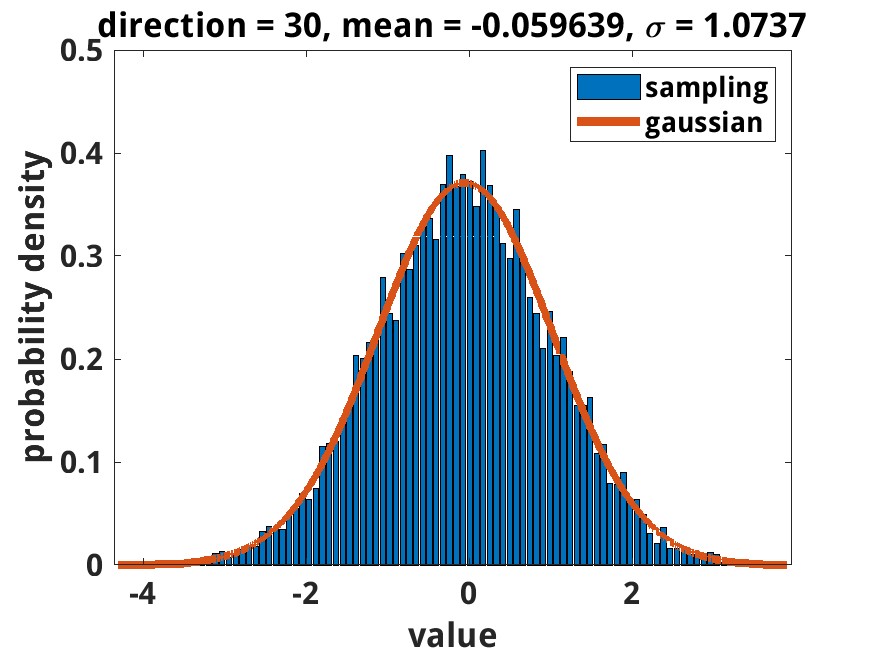}
    \end{subfigure}
   \begin{subfigure}[b]{0.3\textwidth}
    \centering
    \includegraphics[width=\textwidth]{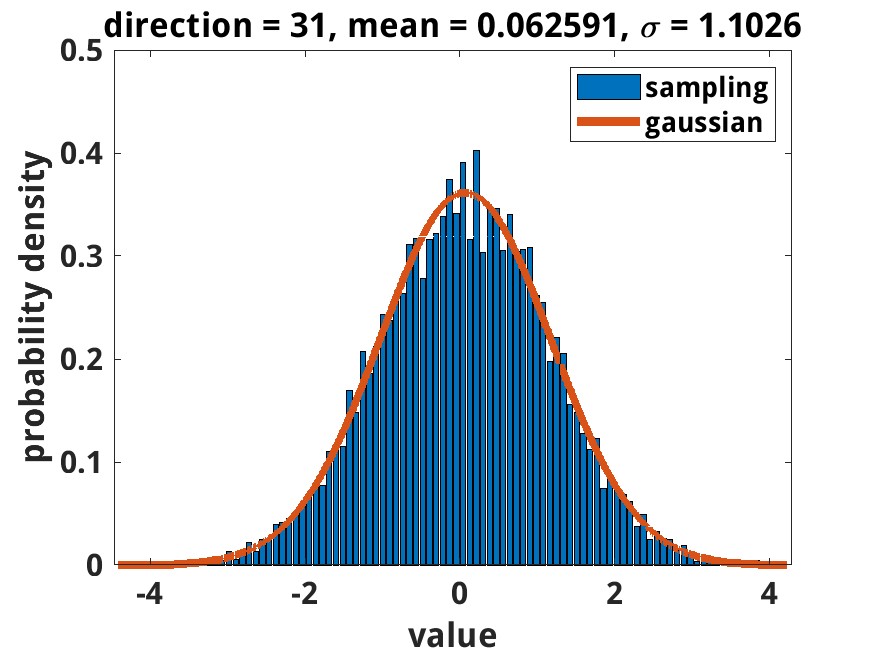}
    \end{subfigure}
    \centering
   \begin{subfigure}[b]{0.3\textwidth}
    \centering
    \includegraphics[width=\textwidth]{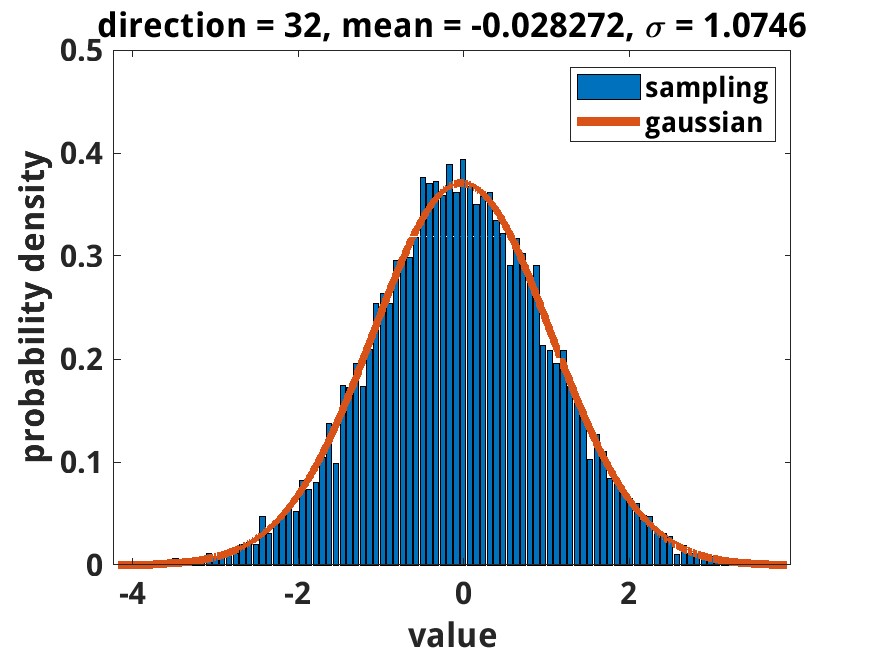}
    \end{subfigure}
    
    \begin{subfigure}[b]{0.3\textwidth}
    \centering
    \includegraphics[width=\textwidth]{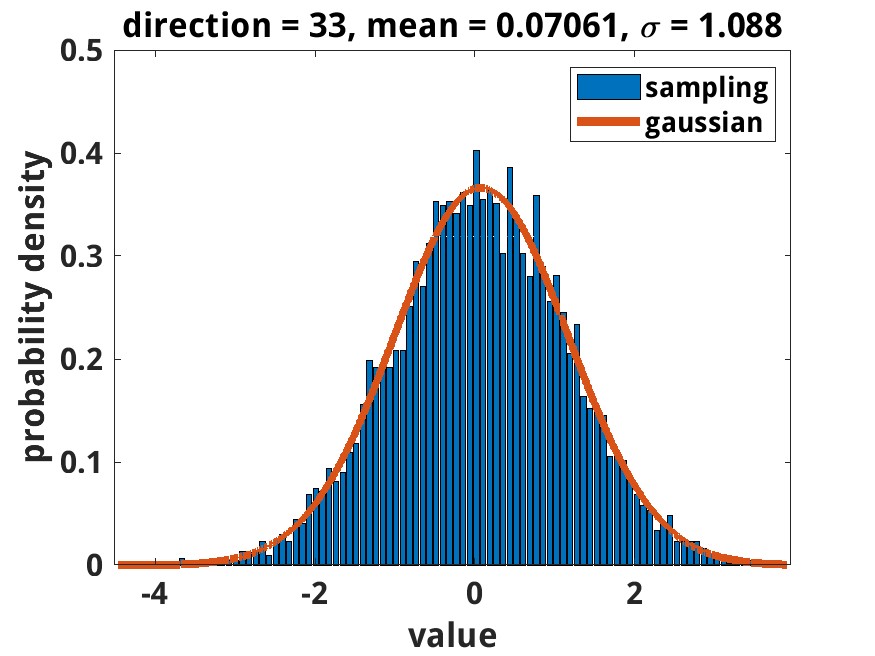}
    \end{subfigure}
   \begin{subfigure}[b]{0.3\textwidth}
    \centering
    \includegraphics[width=\textwidth]{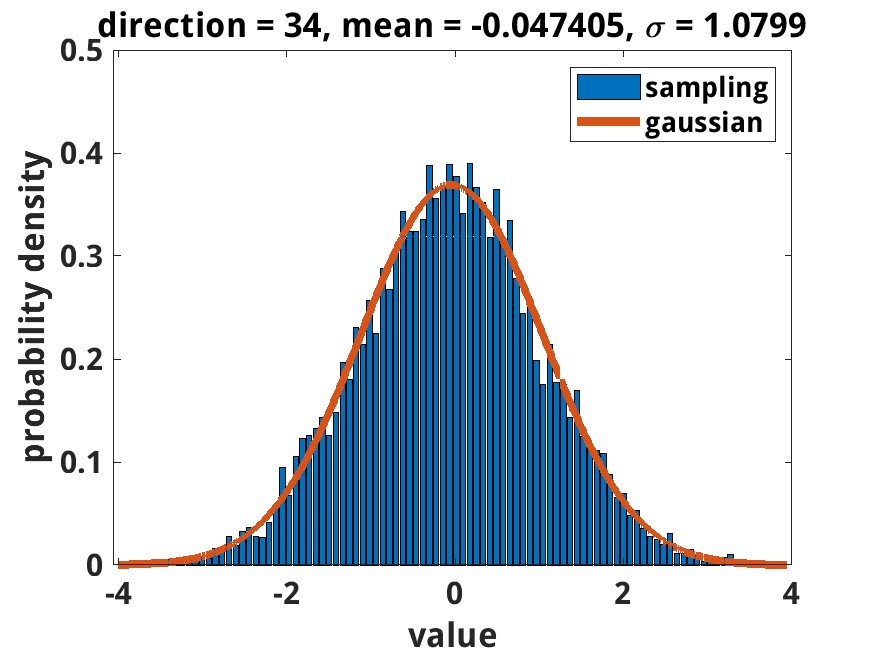}
    \end{subfigure}
    \centering
   \begin{subfigure}[b]{0.3\textwidth}
    \centering
    \includegraphics[width=\textwidth]{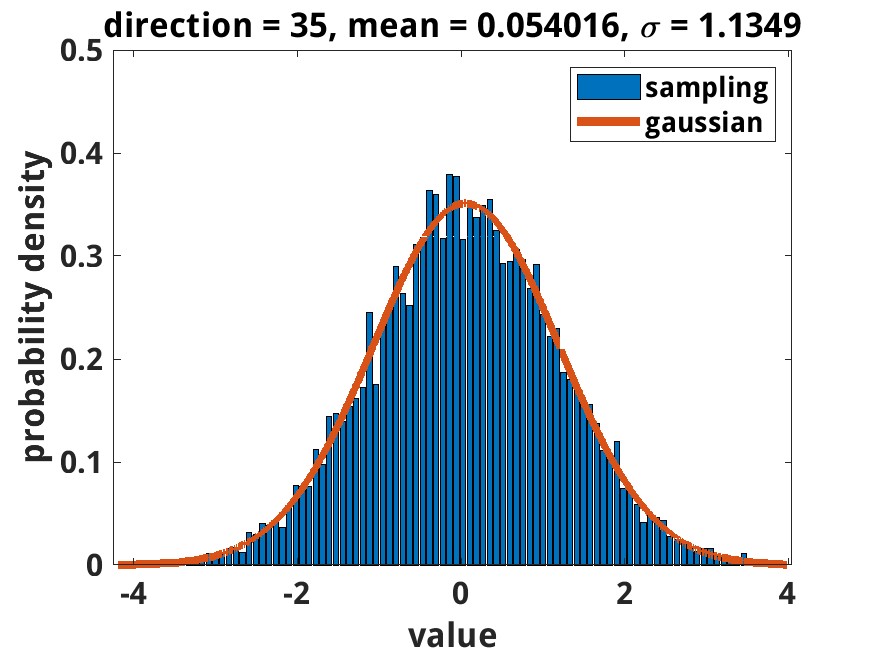}
    \end{subfigure}
    
    \begin{subfigure}[b]{0.3\textwidth}
    \centering
    \includegraphics[width=\textwidth]{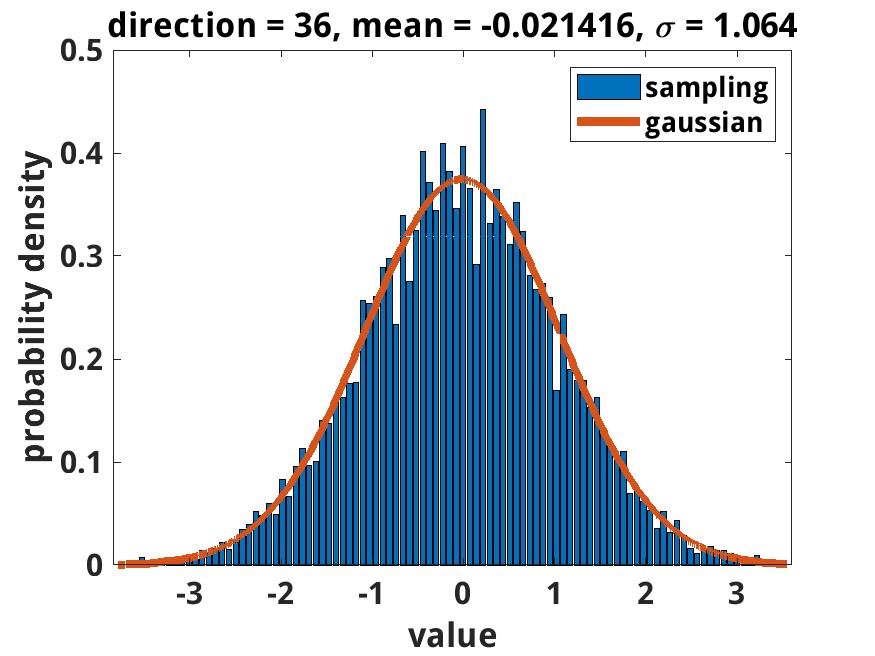}
    \end{subfigure}

    \caption{The histogram of partial derivative estimates with respect to the 27 th to the 36 th parameters are plotted with the same setup as in \cref{fig:vqe_noise}.}
    \label{fig:noise2}
\end{figure}

\section{Cost function for the TSP}\label{sec:cost} 

First, the cost function is defined as 

\[
C(\mathbf{x})=\sum_{i, j} w_{i j} \sum_{p} x_{i, p} x_{j, p+1}+A \sum_{p}\left(1-\sum_{i} x_{i, p}\right)^{2}+A \sum_{i}\left(1-\sum_{p} x_{i, p}\right)^{2}\,,
\]
where $A=10000$ , \( w_{12} = w_{21} = 48 \), \( w_{13} = w_{31} = 91 \), and \( w_{23} = w_{32} = 63 \), and \(w_{ii} = 0, i = 1, 2, 3\). 

We can introduce a new Boolean variable, denoted by \( \tilde{x}_{3i + j - 4} = x_{i,j} \), where \( i, j = 1, 2, 3 \). For simplicity, in the following formula, we will use \( x_0, \ldots, x_8 \) to represent \( \tilde{x}_0, \ldots, \tilde{x}_8 \). With this notation, the expanded form of the cost function can be expressed as:

\begin{equation*}
\begin{split}
C({\bm x}) = -200000 x_0 & - 200000 x_1 - 200000 x_2 - 200000 x_3 - 200000 x_4 - 200000 x_5 \\
& - 200000 x_6 - 200000 x_7 - 200000 x_8 \\
& + \left[ 200000 x_0x_1 + 200000 x_0x_2 + 200000 x_0x_3 + 48 x_0x_4 + 48 x_0x_5 \right.\\
& + 200000 x_0x_6 + 91 x_0x_7 + 91 x_0x_8 + 200000 x_1x_2 + 48 x_1x_3 \\
& + 200000 x_1x_4 + 48 x_1x_5 + 91 x_1x_6 + 200000 x_1x_7 + 91 x_1x_8 \\
& + 48 x_2x_3 + 48 x_2x_4 + 200000 x_2x_5 + 91 x_2x_6 + 91 x_2x_7 \\
& + 200000 x_2x_8 + 200000 x_3x_4 + 200000 x_3x_5 + 200000 x_3x_6 \\
& + 63 x_3x_7 + 63 x_3x_8 + 200000 x_4x_5 + 63 x_4x_6 + 200000 x_4x_7 \\
& + 63 x_4x_8 + 63 x_5x_6 + 63 x_5x_7 + 200000 x_5x_8 + 200000 x_6x_7 \\
& \left.+ 200000 x_6x_8 + 200000 x_7x_8 \right] + 600000
\end{split}
\end{equation*}

In order to build the corresponding Hamiltonian, we align the binary variables \( x_i \) with the Pauli \( Z \) matrices, which operate on individual qubits, and are represented by \( Z_i \). Taking into account the relationship between the binary variables \( x_i \) and the Pauli \( Z \) matrices, defined by the equation \( x_i = \frac{{1 - Z_i}}{2} \), we can express the Hamiltonian for QUBO as follows,
\begin{equation*}
\begin{split}
H_{\text{TSP}} = 600303.0 &-100069.5 Z_0 -100055.5 Z_4 + 12.0 Z_4 Z_0 -100069.5 Z_1\\
&-100055.5 Z_5 + 12.0 Z_5 Z_1 -100069.5 Z_2 -100055.5 Z_3 + 12.0 Z_3 Z_2\\
&-100077.0 Z_7 + 22.75 Z_7 Z_0 -100077.0 Z_8 + 22.75 Z_8 Z_1\\
&-100077.0 Z_6 + 22.75 Z_6 Z_2 + 12.0 Z_3 Z_1 + 12.0 Z_4 Z_2\\
&+ 12.0 Z_5 Z_0 + 15.75 Z_7 Z_3 + 15.75 Z_8 Z_4 + 15.75 Z_6 Z_5\\
&+ 22.75 Z_6 Z_1 + 22.75 Z_7 Z_2 + 22.75 Z_8 Z_0\\
&+ 15.75 Z_6 Z_4 + 15.75 Z_7 Z_5 + 15.75 Z_8 Z_3\\
&+ 50000.0 Z_3 Z_0 + 50000.0 Z_6 Z_0 + 50000.0 Z_6 Z_3\\
&+ 50000.0 Z_4 Z_1 + 50000.0 Z_7 Z_1 + 50000.0 Z_7 Z_4\\
&+ 50000.0 Z_5 Z_2 + 50000.0 Z_8 Z_2 + 50000.0 Z_8 Z_5\\
&+ 50000.0 Z_1 Z_0 + 50000.0 Z_2 Z_0 + 50000.0 Z_2 Z_1\\
&+ 50000.0 Z_4 Z_3 + 50000.0 Z_5 Z_3 + 50000.0 Z_5 Z_4\\
&+ 50000.0 Z_7 Z_6 + 50000.0 Z_8 Z_6 + 50000.0 Z_8 Z_7
\end{split}
\end{equation*}

\section{Technique used in quantum factoring}\label{sec:tech_qfac} 

The introduced technique proposes an alternative formulation for equations of the type \(A B + S = 0\). Here \(A\) and \(B\) represent Boolean variables, while \(S\) denotes integers with \(S \in \mathbb{Z}\). The optimization algorithm targets the minimization of the quadratic version of this equation.

Given the problem Hamiltonian, defined as \( H = (A B + S)^2 \), it can be restructured as:
\begin{equation}
H = 2 \left[ \frac{1}{2} \left( A + B - \frac{1}{2} \right) + S \right]^2 - \frac{1}{8}.
\end{equation}

While the two Hamiltonians are not generally equivalent, they do share the same minimizer due to their underlying Boolean function properties. For instance:

\begin{itemize}
    \item When \(AB=1\): The minimizer for the first Hamiltonian dictates \(S=-1\). In the reformulated version, we get
    \begin{align*}
    H &= 2 \left[ \frac{1}{2} \left( 1 + 1 - \frac{1}{2} \right) - 1 \right]^2 - \frac{1}{8} \\
    &= 0.
    \end{align*}
    
    \item When \(AB=0\): At least one of \(A\) or \(B\) is zero. Assuming \(A=0\) (without loss of generality) and due to the minimizer, we get \(S=0\). This also minimizes the reformulated Hamiltonian since, regardless of whether \(B\) is 0 or 1, the result remains 0.
\end{itemize}

Thus, the reformulated version can be employed interchangeably in certain scenarios. However, this updated representation leads to a significant reduction in the many-body interactions observed experimentally. Specifically, the quartic terms in the Ising Hamiltonian are eliminated, simplifying experimental realizations.
As a result, the third Hamiltonian term $(p_2q_1 + p_1q_2 - 1 ) ^ 2$ in Eq. (\ref{eqn:ss_transform}) is reformulated as:
\[ 
H^{\prime} = 2 \left[ \frac{1}{2} \left( p_{1} + q_{2} - \frac{1}{2} \right) + p_{2} q_{1} - 1 \right]^2 - \frac{1}{8}. 
\]

\section{Additional result for the TFIM with different ansatzes}

In addition to the QAOA-like ansatz depicted in \cref{fig:circuit}, we incorporate hardware-efficient ansatzes \cite{choy2023molecular,yoshioka2022generalized} to evaluate the GD and RCD methods within the context of TFIM \eqref{eq:isingmodel}. The ansatzes, named HEA and HEA2 following the works of \cite{choy2023molecular} and \cite{yoshioka2022generalized} respectively, were modified in our approach. Notably, we replaced the CNOT gates with rotation ZZ gates in the HEA configuration.  For our numerical analysis in \cref{fig:ansatzes}, HEA was configured with 36 parameters while HEA2 was equipped with 240 parameters, which correspond to 18 layers of HEA and 12 layers of HEA2. We also did numerical tests for the two ansatzes with more parameters, for example, 30 layers of HEA and 18 layers of HEA2. We found that performance with RCD method applied to the ansatzes with more layers was degraded in comparison to the case of the smaller numbers of layers.

%For our numerical analysis, HEA was configured with 36 parameters, while HEA2 was equipped with 240 parameters. 

Our analysis revealed a notable performance disparity between the two ansatzes, with the RCD method demonstrating superior efficiency over GD, as depicted in \cref{fig:ansatzes}. Specifically, the optimization of HEA2 required a significantly larger number of derivative computations in comparison to HEA. This discrepancy is largely attributed to the structural differences between HEA and HEA2. In HEA2, each parameter influences a singular qubit rotation Pauli gate, whereas in HEA, a single parameter adjustment can concurrently modify multiple gates. Thus, parameter alterations in HEA2 affect only one parameterized gate at a time, contrasting with HEA where multiple gates are simultaneously updated with each parameter change.

From the results shown in \cref{fig:vqe} and \cref{fig:ansatzes}, our findings indicate that achieving a fidelity threshold of 0.9 with HEA2 demands a substantially higher number of partial derivative calculations than with HEA and the QAOA-like ansatz illustrated in \cref{fig:circuit}. We also see that the QAOA-like ansatz performs slightly better than HEA. Precisely, the maximum and average fidelities at the last iterates of 10 simulations with the QAOA-like ansatz are 0.99 and 0.97 whereas those with HEA are 0.95 and 0.94, respectively. This comparison underscores the relative efficacy of  QAOA-like ansatz in approximating the ground state of the TFIM \eqref{eq:isingmodel} over the hardware-efficient designs HEA and HEA2, despite their intended hardware optimization benefits.

%Furthermore, our findings indicate that achieving a fidelity threshold of 0.9 with both HEA and HEA2 demands a substantially higher number of partial derivative calculations than what was recorded in \cref{fig:vqe}, featuring from the QAOA-like ansatz illustrated in \cref{fig:circuit}. This comparison underscores the QAOA-like ansatz's relative efficacy in approximating the ground state of the TFIM \eqref{eq:isingmodel} over the hardware-efficient designs HEA and HEA2, despite their intended hardware optimization benefits.

\begin{figure}[t]
   \centering
   \begin{subfigure}[b]{0.48\textwidth}
    \centering
    \includegraphics[width=\textwidth]{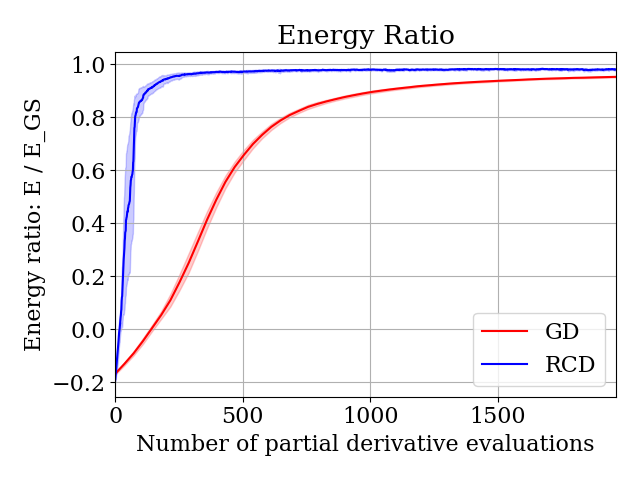}
    \end{subfigure}
   \begin{subfigure}[b]{0.48\textwidth}
    \centering
    \includegraphics[width=\textwidth]{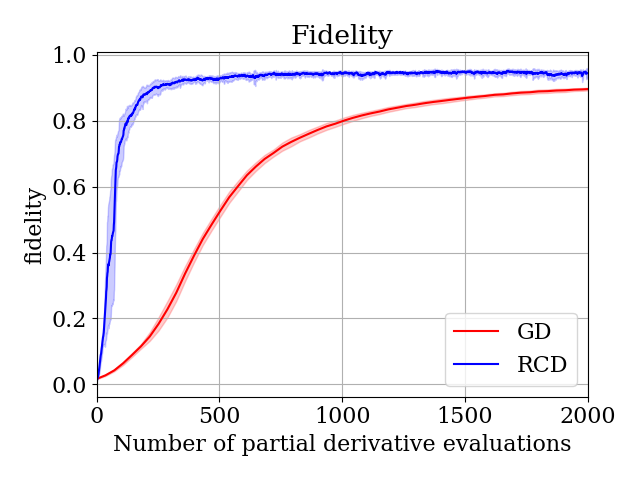}
    \end{subfigure}
    
    \begin{subfigure}[b]{0.48\textwidth}
    \centering
    \includegraphics[width=\textwidth]{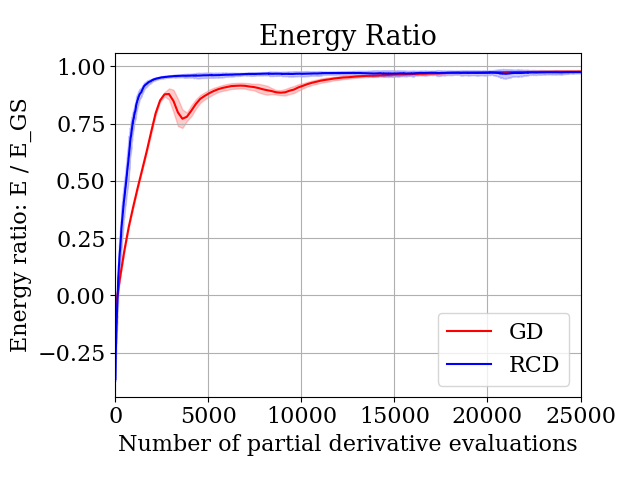}
    \end{subfigure}
   \begin{subfigure}[b]{0.48\textwidth}
    \centering
    \includegraphics[width=\textwidth]{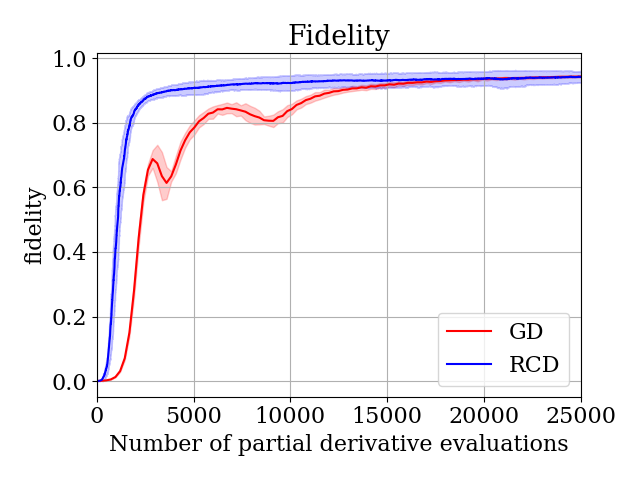}
    \end{subfigure}
    \caption{Comparison of the performance of GD (red) and RCD (blue) for optimizing the Hamiltonian  \eqref{eq:isingmodel}.  
    The result in the top panels is obtained from a slight modification of the ansatz in \cite{choy2023molecular}. For the result in the bottom panels, we used an ansatz \cite{yoshioka2022generalized}.}
    \label{fig:ansatzes}
\end{figure}

\section{Comparison of RCD with SPSA using the TFIM model}\label{sec:RCD_vs_SPSA}

We now present results from another numerical experiment where we compare the performance of RCD and SPSA. The numerical tests are based on the TFIM model \eqref{eq:isingmodel} using the ansatz shown in \cref{fig:circuit}. 
Specifically, we consider the problem \eqref{eq:isingmodel} with $n=12$, the number of qubits. In the implementation, both methods are subject to measurement noise, and we set the number of shots to $1000$ to estimate gradients within both optimization methods. We assigned 80 trainable parameters to the ansatz, that is, 40 layers within the circuit in \cref{fig:circuit}. \cref{fig:rcdvsspsa} shows results from five numerical experiments. All results are obtained from 10 independent simulations with the same fixed initial configuration, but different learning rates, specifically $a=0.008$ for RCD, $a=0.001,0.0005,0.00001$ for SPSA\_1, SPSA\_2, SPSA\_3, and a default setting in Qiskit for SPSA\_4.  Note that the two methods have the same complexity of derivative estimation, employing the finite-difference formula (SPSA) and the parameter shift rule (RCD).

%The results in \cref{fig:rcdvsspsa} suggest that RCD exhibits greater stability properties and efficiency than SPSA.  Furthermore,  hyperparameter tuning in RCD is easier than that in SPSA: RCD requires only a constant learning rate, while SPSA requires scheduling a decreasing learning rate and a diminishing perturbation parameter for the finite-difference scheme.  

From \cref{fig:rcdvsspsa}, it is evident that RCD outperforms SPSA significantly, even after fine-tuning the parameters of SPSA.  We observe that hyperparameter tuning in RCD is easier than that in SPSA: RCD requires only a constant learning rate, while SPSA requires scheduling a decreasing learning rate and a diminishing perturbation parameter for the finite-difference scheme. To obtain successful SPSA training,  a small learning rate must be selected. Compared to RCD, to achieve stable training for SPSA, the learning rate must be selected 80 times smaller, rendering the algorithm much less efficient than RCD. This can be clearly observed in \cref{fig:rcdvsspsa}, where RCD (with learning rate $a=8\times 10^{-3}$) attains a fidelity exceeding 0.95 at the end, already reaching over 0.9 between 200 and 400 parameter shift rule evaluations. In contrast, only SPSA\_3 (with learning rate $a=10^{-4}$) achieves a fidelity greater than 0.9, but with a much slower rate. 

% Among the four SPSA results, SPSA\_2 exhibits rapid convergence between 200 and 400 finite evaluations but still lags behind RCD in performance.

\begin{figure}[htbp]
    \caption{Comparison of the performance of RCD (red) and SPSA (other colors) for optimizing the Hamiltonian of the TFIM \eqref{eq:isingmodel} defined in Section 1.5. The $x$ axis labels the number of parameter shift rule (PSR) or finite-difference (FD) evaluations as an indication of the computational complexity, which corresponds to RCD and SPSA, respectively. Both evaluations involve the same measurement complexity.
    The $y$ axis indicates the fidelity with the ground state from the iterations.}
    \centering
    \includegraphics[width=0.5\textwidth]{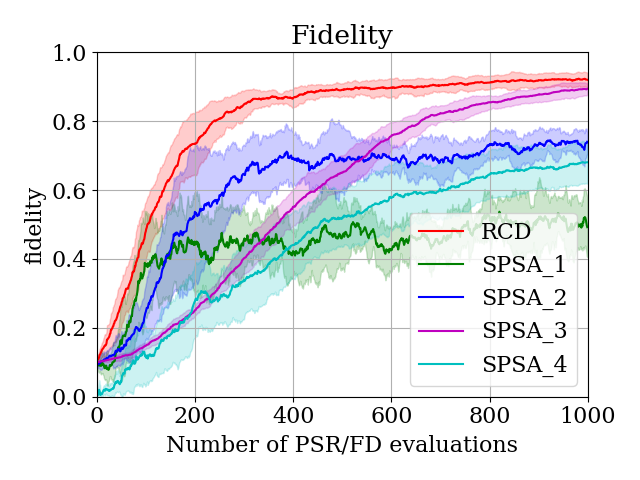}
    \label{fig:rcdvsspsa}
\end{figure}

\bibliographystyle{abbrv}
\bibliography{ref}
\end{document}